\documentclass{article}[11pt]
\usepackage{hyperref}
\usepackage{graphicx}
\usepackage{multicol}
\usepackage{color}
\usepackage{caption}
\usepackage{setspace}
\usepackage{amssymb,amsmath}  
\numberwithin{equation}{section}   
\usepackage{amsfonts,amsrefs}
\usepackage{amsthm} %added to format proofs
\usepackage{bm}    
\usepackage{bbm}      
\usepackage{tabularx}
\usepackage{enumerate}
\usepackage{dutchcal}
\usepackage{authblk}
\usepackage[linesnumbered,ruled,vlined]{algorithm2e}
\usepackage{algorithmic}
\usepackage{lscape}
\makeatletter
\def\verbatim@font{\linespread{1}\normalfont\ttfamily}
\makeatother
\usepackage{pgfpages}
\usepackage{changepage}
\usepackage{fancyhdr,a4wide}

\usepackage[utf8]{inputenc}
\usepackage[T1]{fontenc}
\usepackage{amsmath}
\usepackage{tensor}
\usepackage{amssymb}
\usepackage{graphicx}
\usepackage{subfigure}
\usepackage{mathrsfs}
\usepackage{microtype}
\usepackage{rotating}
\usepackage{xspace}
\usepackage{mathtools}
\usepackage{tikz,pgfplots}
\usepackage{comment}
\pgfplotsset{compat=1.9}
\usetikzlibrary{fit}
\tikzset{every label/.style={font=\footnotesize,inner sep=1.5pt}}
\allowdisplaybreaks

\def\jf#1{\relax}

\newcommand{\Tbb}{\mathbb{T}{}}
\newcommand{\Zbb}{\mathbb{Z}{}}
\newcommand{\Rbb}{\mathbb{R}{}}

\newcommand*{\dd}{\mathrm{d}}

\newcommand*{\del}{\partial}

\newcommand*{\eps}{\epsilon}

\newcommand{\frd}{\mathfrak{d}}
\newcommand{\frv}{\mathfrak{v}}
\newcommand{\frg}{\mathfrak{u}}

\newcommand{\fru}{\mathfrak{u}}

\newcommand{\ut}{\tilde{u}}
\newcommand{\Ht}{\tilde{H}}
\newcommand{\Kt}{\tilde{K}}

\newcommand{\Fc}{\mathcal{F}}
\newcommand{\nuh}{\hat{\nu}}
\newcommand{\wh}{\hat{w}}
\newcommand{\uh}{\hat{u}}
\newcommand{\Hct}{\tilde{\mathcal{H}}}
\newcommand{\alphat}{\tilde{\alpha}}
\newcommand{\alphah}{\hat{\alpha}}

\newcommand{\afrak}{\mathfrak{a}}
\newcommand{\Dt}{\tilde{D}}

\newcommand{\tDLd}{\tilde{\Delta}_{L\delta}}
\newcommand{\et}{\tilde{e}}
\newcommand{\sigmat}{\tilde{\sigma}}

\newcommand{\thetat}{\tilde{\theta}}
\newcommand{\gt}{\tilde{g}}

\newcommand{\vt}{\tilde{v}}
\newcommand{\htl}{\tilde{h}}

\usepackage[normalem]{ulem}

\theoremstyle{plain}
\newtheorem{thm}{Theorem}[section]

\newtheorem{prop}[thm]{Proposition}

\theoremstyle{definition}

\theoremstyle{remark}
\newtheorem{rem}[thm]{Remark}

\usepackage{placeins}
\usepackage{multirow}

\newcommand*{\R}[1]{\mbox{\sffamily\ifcase#1\relax\or I\or II\or III\or IV\fi}\xspace}

\begin{document}
% \title[Future Stability of Tilted Two-Fluid Bianchi I Spacetimes]{Future Stability of Tilted Two-Fluid Bianchi I Spacetimes} 

% \author[G. ~Fournodavlos]{Grigorios Fournodavlos}
% \address{Department of Mathematics and Applied Mathematics\\
% University of Crete, 70013 Heraklion\\ Greece}
% \email{gfournodavlos@uoc.gr}

% \author[E.~Marshall]{Elliot Marshall}
% \address{School of Mathematics\\
% 9 Rainforest Walk\\
% Monash University, VIC 3800\\ Australia}
% \email{elliot.marshall@monash.edu}

% \author[T.A.~Oliynyk]{Todd A.~Oliynyk}
% \address{School of Mathematics\\
% 9 Rainforest Walk\\
% Monash University, VIC 3800\\ Australia}
% \email{todd.oliynyk@monash.edu}

\title{Future Stability of Tilted Two-Fluid Bianchi I Spacetimes} 

% \author{Grigorios Fournodavlos\footnote{Department of Mathematics and Applied Mathematics, University of Crete, 70013 Heraklion, Greece, gfournodavlos@uoc.gr}, Elliot Marshall\footnote{School of Mathematics, 9 Rainforest Walk, Monash University, VIC 3800, Australia} \footnote{Corresponding author: elliot.marshall@monash.edu} , Todd A. Oliynyk\footnote{School of Mathematics, 9 Rainforest Walk, Monash University, VIC 3800, Australia, todd.oliynyk@monash.edu}}

\author[1,2]{Grigorios Fournodavlos}
\author[2,3,a]{Elliot Marshall}
\author[3]{Todd A. Oliynyk}
\affil[1]{Department of Mathematics and Applied Mathematics, University of Crete, 70013 Heraklion, Greece.}
\affil[2]{Institute of Applied and Computational Mathematics, FORTH, 70013 Heraklion, Greece.}
\affil[3]{School of Mathematics, 9 Rainforest Walk, Monash University, VIC 3800, Australia.}
\affil[a]{Corresponding author: elliot.marshall@iacm.forth.gr}

\maketitle
% \address{Department of Mathematics and Applied Mathematics\\
% University of Crete, 70013 Heraklion\\ Greece}
% \email{gfournodavlos@uoc.gr}

% \author[E.~Marshall]{Elliot Marshall}
% \address{School of Mathematics\\
% 9 Rainforest Walk\\
% Monash University, VIC 3800\\ Australia}
% \email{elliot.marshall@monash.edu}

% \author[T.A.~Oliynyk]{Todd A.~Oliynyk}
% \address{School of Mathematics\\
% 9 Rainforest Walk\\
% Monash University, VIC 3800\\ Australia}
% \email{todd.oliynyk@monash.edu}

\maketitle

\begin{abstract} 
We establish the nonlinear stability to the future of tilted two-fluid Bianchi I solutions to the Einstein-Euler equations with positive cosmological constant and linear equations of state $p_{(\afrak)}=K_{(\afrak)}\rho_{(\afrak)}$, $\afrak\in\{1,2\}$, where $\frac{1}{3}<K_{(\afrak)}<\frac{5}{7}$. 
\end{abstract}

\maketitle

\section{Introduction}

Exponentially expanding, fluid-filled cosmological models governed by linear equations of state of the form $p = K\rho$, where $0 \leq K \leq 1$, play a fundamental role in modern cosmology. Owing to their physical significance, considerable effort has been devoted to rigorously analyzing their future stability. Foundational results in this direction were obtained by Rodnianski and Speck \cites{RodnianskiSpeck:2013,Speck:2012}, who established the nonlinear future stability of FLRW (spatially homogeneous and isotropic) solutions to the Einstein-Euler equations with a positive cosmological constant, for the parameter range $0 < K < \frac{1}{3}$. The endpoint cases $K = 0$ and $K = \frac{1}{3}$ were subsequently treated in \cite{HadzicSpeck:2015} and \cite{LubbeKroon:2013}, respectively. Related studies have explored models with hyperbolic spatial slices \cites{Mondal:2021,Mondal:2024}, alternative equations of state \cites{LeFlochWei:2021,LiuWei:2021}, and other classes of expanding spacetimes \cites{Speck:2013,Ringstrom:2009,FOW:2021,Wei:2018,FOOW:2023,Fajman_et_al:2025}. \newline \par

Until recently, the stability of cosmological models with super-radiative equations of state ($K > \frac{1}{3}$) remained an open question. Influential work by Rendall \cite{Rendall:2004} suggested that super-radiative FLRW models are \emph{unstable}- a phenomenon now known as the ``Rendall instability''. This instability has since been observed numerically \cites{BMO:2023,BMO:2024} and rigorously demonstrated in the restricted setting of $\mathbb{T}^2$-symmetric fluids evolving on a fixed de Sitter background \cite{Oliynyk:2024}. In contrast, the existence of a homogeneous but \textit{anisotropic} family of solutions to the relativistic Euler equations, that is future stable for $\frac{1}{3}<K<\frac{1}{2}$ was established by the third author in \cite{Oliynyk:2021}. This was subsequently extended to the full parameter range $\frac{1}{3}<K<1$ by the second and third authors in \cite{MarshallOliynyk:2022}. \newline \par

The anisotropic family of solutions studied in \cites{Oliynyk:2021,MarshallOliynyk:2022} is referred to as \textit{tilted fluids}. The concept of tilt is defined relative to a fixed timelike vector field $u = u^\mu \partial_\mu$. A fluid with four-velocity $v = v^\mu \partial_\mu$ is said to be \textit{orthogonal} with respect to $u$ if $v$ is aligned with $u$; otherwise, the fluid is classified as \textit{tilted} with respect to $u$.\newline \par

For spacetimes without symmetries, there is no preferred choice for $u$. However, in spatially homogeneous cosmologies, there exists a preferred spacelike foliation by hypersurfaces of homogeneity. In this setting, it is natural to take $u$ to be the timelike normal vector to the foliation. \newline \par

The tilt classification is introduced because homogeneous tilted and orthogonal fluids can exhibit markedly different asymptotic behaviour. This difference has been extensively studied in the dynamical systems literature; see, for example, \cites{EllisWainwright:1997,Sandin:2009,SandinUggla:2008,Hervik_et_al:2010,Hewitt_et_al:2001}. Notably, these works show that $K = \frac{1}{3}$ is a \emph{bifurcation point} in the dynamics of homogeneous cosmological models. For $K < \frac{1}{3}$, solutions asymptotically approach orthogonal fluid configurations, whereas for $K > \frac{1}{3}$, the fluid exhibits \emph{extreme tilt} asymptotically, that is, the fluid velocity becomes null at future timelike infinity. \newline \par

A natural extension of the stability results from \cites{Oliynyk:2021,MarshallOliynyk:2022} is to determine whether tilted fluid solutions to the Einstein-Euler system with a positive cosmological constant remain future stable over the parameter range $\frac{1}{3} < K < 1$. This question was partially addressed by the authors in \cite{Fournodavlos_et_al:2024}, where it was shown that tilted fluid solutions with spatial topology $\mathbb{S}^3$ are future stable for $\frac{1}{3} < K < \frac{3}{7}$. The main aim of this article is to extend these results to tilted Bianchi I spacetimes and increase the parameter range over which stability holds. Due to the momentum constraint, Bianchi I models can exhibit tilt only in the presence of two or more fluids \cite{SandinUggla:2008}. Consequently, we restrict our attention to analysing the stability of \emph{two fluid} models where each fluid satisfies a linear equation of state $p_{(\afrak)}=K_{(\afrak)}\rho_{(\afrak)}$, $\afrak = 1,2$. The principal result of this article is the proof that tilted two-fluid Bianchi I solutions are future stable for $\tfrac{1}{3} < K_{(1)} \leq K_{(2)} < \tfrac{5}{7}$, extending the parameter range established in \cite{Fournodavlos_et_al:2024}.

\subsection{Einstein-Euler Equations}
The Einstein-Euler equations, in (3+1) dimensions, for two non-interacting barotropic perfect fluids with a positive cosmological constant $\Lambda>0$ are given by
\begin{gather}
\label{eqn:Einstein_Physical}
\tilde{R}_{\mu\nu} -\Lambda \tilde{g}_{\mu\nu} = \tilde{T}_{\mu\nu} - \frac{1}{2}\tilde{T} \tilde{g}_{\mu\nu}, \\
\label{eqn:Euler_Physical}
\tilde{\nabla}^{\mu}\tilde{T}^{(1)}_{\mu\nu} = 0, \quad
\tilde{\nabla}^{\mu}\tilde{T}^{(2)}_{\mu\nu} = 0,
\end{gather}
where the fluid stress energy tensors are defined by
\begin{align*}
\tilde{T}^{(\afrak)}_{\mu\nu} = (\rho_{(\afrak)}+p_{(\afrak)})\tilde{v}^{(\afrak)}_{\mu}\tilde{v}^{(\afrak)}_{\nu} + p_{(\afrak)}\tilde{g}_{\mu\nu}, \quad \afrak=1,2.
\end{align*}
The fluid four-velocities $\tilde{v}^{(\afrak)}_{\mu}$  are normalised by $\tilde{g}^{\mu\nu}\tilde{v}^{(\afrak)}_{\mu}\tilde{v}^{(\afrak)}_{\nu} = -1$, 
and $\rho_{(\afrak)}$ and $p_{(\afrak)}$ are the fluid densities and pressures, respectively.  The total stress-energy tensor of the fluid matter is
\begin{align} \label{Tt-def}
\tilde{T}_{\mu\nu} = \tilde{T}^{(1)}_{\mu\nu} + \tilde{T}^{(2)}_{\mu\nu},
\end{align}
which, by \eqref{eqn:Euler_Physical}, satisfies $\tilde{\nabla}_{\mu}\tilde{T}^{\mu\nu} = 0$.
We assume linear equations of state
\begin{align} \label{eos}
p_{(\afrak)} &= K_{(\afrak)}\rho_{(\afrak)}
\end{align}
for each fluid where
\begin{equation} \label{K-range}
   \frac{1}{3}< K_{(1)} \leq K_{(2)} <1.
\end{equation}

\begin{rem}
The following combinations of the sound speed parameters $K_{(\afrak)}$ will play a critical role in the analysis:
\begin{align} \label{mu-def}
    \mu_{(\afrak)} = \frac{3K_{(\afrak)}-1}{1-K_{(\afrak)}}.
\end{align} 
By \eqref{K-range}, we note that
    \begin{align}
    \label{eqn:Mu1<Mu2}
        0 <\mu_{(1)} \leq \mu_{(2)} \quad \text{and}\quad 
        \frac{4K_{(\afrak)}}{1-K_{(\afrak)}} > 2.
    \end{align}
\end{rem}

Instead of working directly with the physical metric $\tilde{g}_{\mu\nu}$, we find it advantageous to employ a conformal metric $g_{\mu\nu}$ defined by
\begin{align}
\label{eqn:conformal_metric_def}
    g_{\mu\nu} = e^{-2\Phi}\tilde{g}_{\mu\nu},
\end{align}
where 
\begin{align} \label{Phi-def}
\Phi = -\ln(t). 
\end{align}
This transformation compactifies the time domain to a finite interval of the form $(0,T_{1}]$ where $t=0$ corresponds to future timelike infinity. That is to say, the future is now located in the direction of \textit{decreasing} $t$. \newline \par

In terms of the conformal metric, the Einstein equations \eqref{eqn:Einstein_Physical} become\footnote{This follows from the identities in \cite{Wald:1984}*{Appendix D}.}
\begin{align}
\label{eqn:ConformalEinstein}
R_{\mu\nu} &= e^{2\Phi}T_{\mu\nu} - \frac{1}{2}Te^{2\Phi}g_{\mu\nu} +\Lambda e^{2\Phi}g_{\mu\nu} + 2\nabla_{\mu}\nabla_{\nu}\Phi \nonumber \\
&+g_{\mu\nu}(\Box_g\Phi +2|\nabla\Phi|^{2}_{g}) -2\nabla_{\mu}\Phi\nabla_{\nu}\Phi,
\end{align}
where $\nabla$ and $R_{\mu\nu}$ are the Levi-Civita connection and Ricci tensor, respectively, of the conformal metric $g$, $\Box_g\Phi = \nabla_{\mu}\nabla^{\mu}\Phi$ is the wave operator, $|\nabla\Phi|^{2}_{g} = g^{\mu\nu}\nabla_{\mu}\Phi\nabla_{\nu}\Phi$, and we have introduced the conformal stress-energy tensor $T_{\mu\nu}$ defined by
\begin{equation}
\begin{aligned}
\label{eqn:Conformal_Tmunu_def}
T_{\mu\nu} &= T^{(1)}_{\mu\nu} + T^{(2)}_{\mu\nu}, \\
\tilde{T}^{(\afrak)}_{\mu\nu} &= e^{2\Phi}T^{(\afrak)}_{\mu\nu} = e^{2\Phi}\Big[(\rho^{(\afrak)}+p^{(\afrak)})v^{(\afrak)}_{\mu}v^{(\afrak)}_{\nu}+p^{(\afrak)}g_{\mu\nu}\Big],
\end{aligned}
\end{equation}
where $v_{(\afrak)}^{\mu} = e^{\Phi}\tilde{v}_{(\afrak)}^{\mu}$ are the conformal fluid four-velocities.

\subsection{Tilted Two-Fluid Bianchi I Spacetimes}
Two-fluid Bianchi I spacetimes are spatially homogeneous solutions of the Einstein-Euler equations \eqref{eqn:Einstein_Physical}-\eqref{eqn:Euler_Physical} with flat spatial geometry. These cosmological models with linear equations of states have been extensively studied in the literature \cites{Sandin:2009,SandinUggla:2008,EllisWainwright:1997}. In particular, it was shown in \cite{SandinUggla:2008} that, in the presence of a positive cosmological constant, such models isotropise towards the future for the range of sound speeds specified in \eqref{K-range}. As we demonstrate, this isotropisation property enables a reduction of the stability analysis to a simpler class of solutions \textemdash \;  namely, tilted two-fluid solutions of the relativistic Euler equations on a fixed de Sitter background. \newline \par 

Further details on two-fluid Bianchi I spacetimes are provided in Appendix \ref{Appendix:two-fluid}. There, the future asymptotic behaviour is rigorously analysed, with quantitative estimates that make precise the sense in which these spacetimes isotropise and can be approximated by solutions of the relativistic Euler equations on a prescribed de Sitter spacetime.

\subsection{Informal Statement of the Main Theorem}
The main results of this article are a proof that tilted two-fluid Bianchi I spacetimes are nonlinearly stable to the future for 
\begin{equation*}
    \frac{1}{3} < K_{(1)} \leq K_{(2)} < \frac{5}{7}\quad \Longleftrightarrow \quad 0<\mu_{(1)}\leq \mu_{(2)} < 4,
\end{equation*} and  moreover, that the fluid velocities of the perturbed solutions develop extreme tilts near future timelike infinity. An informal version of our main result is provided below; see Theorem \ref{mainthm} for a precise statement.

\begin{thm}[Future Stability of Tilted Two-Fluid Bianchi I Spacetimes\label{maintheorem:informal}] 
     Solutions to the conformal Einstein-Euler equations \eqref{eqn:Einstein_Physical}-\eqref{eqn:Euler_Physical} with positive cosmological constant that are generated from sufficiently differentiable initial data imposed on $\Sigma_{T_{1}} = \{T_{1}\}\times \mathbb{T}^{3}$, $T_1>0$,  exist globally to the future on the spacetime region $M \cong (0,T_{1}] \times \mathbb{T}^{3}$ provided  $\frac{1}{3} < K_{(1)} \leq K_{(2)} < \frac{5}{7}$ and the initial data is sufficiently close to homogeneous two-fluid Bianchi I initial data. Moreover, the fluid velocities exhibit extreme tilt asymptotically at future timelike infinity.   
\end{thm}

\begin{rem}
    Recall that $t=0$ corresponds to future timelike infinity in our conventions. In particular, the interval $(0,T_{1}]$ in the theorem statement corresponds to a time domain of $[\tilde{T}_{1},\infty)$ with respect to the physical time coordinate.
\end{rem}

\subsection{Future Directions}
A natural next step is to investigate the stability of tilted Bianchi solutions over the entire parameter range $\frac{1}{3} < K < 1$. Although the present analysis provides an upper bound of $K=\frac{5}{7}$, this restriction appears to be technical rather than fundamental. Indeed, we expect that tilted solutions remain stable throughout the full range $\frac{1}{3} < K < 1$. This expectation is supported by numerical simulations in $\mathbb{T}^2$-symmetry, which demonstrate that small perturbations of tilted two-fluid Bianchi I solutions exist globally towards the future \cite{Marshall_Thesis:2025}. Furthermore, in Gowdy-symmetry,  tilted two-fluid Bianchi I solutions have been shown to be future stable over the entire range $\frac{1}{3} < K_{(1)} \leq K_{(2)} < 1$ \cite{Marshall_Thesis:2025}. \newline \par

Additionally, it is interesting to note that the bifurcation at $K=\frac{1}{3}$ between tilted and orthogonal homogeneous solutions also occurs in the past towards a big bang singularity \cite{BMO:2024}. However, in this setting, the situation is more complicated since the bifurcation value for $K$ depends on the Kasner exponents. Recently, Beyer and Oliynyk \cite{BeyerOliynyk:2024} have established the nonlinear stability of the FLRW solution to the Einstein-Euler-scalar field system for $\frac{1}{3}<K<1$. On the other hand, the Rendall instability has been numerically observed \cite{BMO:2024} for perturbations of the FLRW solution when $K<\frac{1}{3}$. Based on the similarities between the asymptotic regimes towards the past and future, we expect an analogous past stability result could be established for tilted Bianchi models. We plan to investigate this further in future work.

\subsection{Proof Summary}

\subsubsection{Orthonormal Frame Formulation of Einstein-Euler Equations}
The main result of this article, namely, the future nonlinear stability of the tilted two-fluid Bianchi spacetimes, is presented in Theorem \ref{mainthm}. To establish this result, we employ a tetrad formulation of the Einstein-Euler equations, which follows from a straightforward adaptation of the formalism developed by R\"{o}hr and Uggla \cite{RöhrUggla:2005} (see also \cites{EllisWainwright:1997,ElstUggla:1997}). The specific system employed in our analysis is introduced in Section \ref{sec:Frame_Formulation}; see equations \eqref{eqn:frame_component_evo_intro}-\eqref{eqn:Hamiltonian_Constraint_intro}. \newline \par 

To fix the gauge, we set the shift vector to zero, propagate the spatial frame via Fermi-Walker transport, and determine the lapse by imposing the generalised harmonic slicing condition
\begin{align*}
\Box_g t = f
\end{align*}
relative to the conformal metric $g$. The gauge source function $f$ is chosen to drive the conformal lapse $\alpha$ toward the de Sitter background value $\sqrt{\frac{\Lambda}{3}}$; see \eqref{eqn:gauge_source_f} for its precise form. A similar harmonic time slicing condition was previously employed in \cite{Oliynyk:CMP_2016} to establish the future stability of FLRW solutions with $0 < K \leq \frac{1}{3}$.

\subsubsection{Symmetric Hyperbolic Formulation of the Einstein-Euler Equations}
The next step in the stability proof is to establish local-in-time existence for solutions to the tetrad formulation of the Einstein-Euler equations. The gauge-reduced evolution equations arising from the Röhr-Uggla formulation of the Einstein equations \cite{RöhrUggla:2005} are presented in Section \ref{sec:Frame_Formulation}; see equations \eqref{eqn:frame_form1}-\eqref{eqn:alpha_evo_2}. In their initial form, these equations are not manifestly well-posed. \newline \par 

To address this, we derive a symmetric hyperbolic formulation in Section \ref{sec:Sym_Hyp_Einstein}. The derivation proceeds in four steps: \textbf{(i)} First, in Section \ref{sec:Sec3_3:Step1}, we specify the gauge source function $ f$ and introduce new variables $\Hct$ and $\alphat$, which improve the singular structure of the evolution equations. \textbf{(ii)} Next, in Section \ref{sec:Sec3_3:Step2}, we modify the evolution equations for $ \alphat$, $A_{A}$, and $\dot{U}_{A}$ by adding suitable multiples of the constraint equations. The evolution equations for  $N_{AB}$ and $\Sigma_{AB}$ are also reformulated in a form amenable to later analysis. \textbf{(iii)} Then, in Section \ref{sec:Sec3_3:Step3}, we show that the full system of evolution equations can be written in symmetric hyperbolic form, with a singular structure compatible with the global existence theory developed in \cite{BOOS:2021}. \textbf{(iv)} Finally, in Section \ref{sec:Subtract_BG_Einstein_Sec3}, we subtract the homogeneous background solution being perturbed, which yields the symmetric hyperbolic formulation of the (gauge-reduced) Einstein equations  given by \eqref{eqn:SH_Einstein_Sec3}.\newline \par

In Section \ref{sec:Sym_Hyp_Euler}, we derive a symmetric hyperbolic formulation of the relativistic Euler equations \eqref{eqn:Euler_Physical}. This section involves a \textit{significant} amount of computations. In an attempt to simplify the derivation, we break it into five steps:
\textbf{(i)} First, in Section \ref{sec:Conformal_Euler_deriv}, we express the physical Euler equations \eqref{eqn:Euler_Physical} with respect to the conformal stress-energy tensor \eqref{eqn:Conformal_Tmunu_def}. 
\textbf{(ii)} Next, in Section \ref{sec:Euler_3+1_decomp}, we derive the 3+1 decomposition of the Euler equations and introduce the modified density variable $\zeta$, which will ultimately yield a better singular structure for the equations.
\textbf{(iii)} We then perform, in Section \ref{sec:Fluid_2+1_Decomp}, a radial decomposition of the spatial fluid velocity. This generalises the change of variables used in \cites{Oliynyk:2021,MarshallOliynyk:2022} and is necessary to obtain the correct Fuchsian structure for the equations. \textbf{(iv)} New variables are introduced in Section \ref{sec:SH_Conformal_Euler_Derivation} that allow us to express the conformal Euler equations in symmetric hyperbolic form. \textbf{(v)} Finally, in Section \ref{sec:Euler_Subtract_Background}, we subtract the homogeneous background to yield the final symmetric hyperbolic formulation of the relativistic Euler equations given by \eqref{eqn:Euler_coord_deriv1}. \newline \par

Together, equations \eqref{eqn:SH_Einstein_Sec3} and \eqref{eqn:Euler_coord_deriv1} yield a symmetric hyperbolic formulation of the gauge-reduced Einstein-Euler equations. The local-in-time existence and uniqueness of solutions to this system, along with a continuation principle and constraint propagation, are established in Section \ref{sec:Local_Existence_ConstraintProp}. A precise statement of these results is provided in Proposition \ref{prop:Local_ExistenceUniqueness_Prop}.

\subsubsection{Global Existence via Fuchsian Techniques}
\label{sec:Proof_Outline_SpatialDerivatives}

Our proof of global existence relies on expressing the Einstein-Euler system as a singular, symmetric hyperbolic system of the form
\begin{align}
\label{eqn:FuchsianPDE_Example}
    B^{0}(t,U)\del_{t}U + B^{\Omega}(t,U)\del_{\Omega}U = \frac{1}{t}\mathcal{B}(t,U)\mathbb{P}U + F(t,U).
\end{align}
We say \eqref{eqn:FuchsianPDE_Example} is \textit{Fuchsian} if the right-hand side is formally singular at $t=0$. In recent years, a global existence theory for Fuchsian systems has been developed and applied to a wide range of stability problems in general relativity.  \newline \par 

Historically, analysis of Fuchsian systems has been focused on the singular initial value problem (SIVP) in which one prescribes asymptotic data at the singular time $t=0$. Using \eqref{eqn:FuchsianPDE_Example}, this data is then evolved \textit{away} from the singular time. For our purposes, however, it is the standard initial value problem (IVP) which is more relevant. Here, the goal is to prescribe initial data at some finite time $T_{1}>0$ and construct a solution up to the singular time $t=0$. That is to say, a solution on the whole interval $(0,T_{1}]$.  \newline \par

In \cite{Oliynyk:CMP_2016}, it was shown by Oliynyk that for a certain class of Fuchsian systems one can always construct solutions on the interval $(0,1]$ provided the initial data at $t=1$ is taken to be suitably small. Oliynyk then applied this result to a Fuchsian formulation of the Einstein-Euler equations with positive cosmological constant in order to prove the future stability of FLRW solutions for $0<K\leq\frac{1}{3}$. More recently, Beyer, Oliynyk, and Olvera-Santamaría  have generalised the small data global existence theory of \cite{Oliynyk:CMP_2016} to a wider class of Fuchsian systems in \cite{BOOS:2021}. This extended theory has subsequently been used to establish stability results for a variety of systems, see for example \cites{BeyerOliynyk:2023,BeyerOliynyk:2024,BOZ:2025, FOW:2021, FOOW:2023}. \newline \par

Now, while the symmetric hyperbolic formulation of the gauge-reduced Einstein-Euler equations derived in Sections \ref{sec:Sym_Hyp_Einstein} and \ref{sec:Sym_Hyp_Euler} is locally well-posed, it does not have the appropriate Fuchsian form for analysing the behaviour of solutions near future timelike infinity using Fuchsian techniques. The difficulty arises from certain singular terms in the Euler equations, whose structure is not manifestly favourable. Specifically, the symmetrised formulation of the Euler equations derived in Section \ref{sec:Sym_Hyp_Euler} has the schematic structure
\begin{align}
\label{eqn:euler_schematic_form_intro}
   \del_{t}W + \mathbf{C}^{\Sigma}\del_{\Sigma}W = \frac{-\mu}{t}\Pi W + t^{2\mu-1}\mathbf{S} + \mathbf{H}
\end{align}
where $\mu$ is a positive constant. Already at the ODE level\footnote{That is, by setting spatial derivatives to zero in \eqref{eqn:euler_schematic_form_intro}.}, one can see that the singular term $\frac{-\mu}{t}\Pi W$ results in polynomial growth instead of decay! Terms of this form were first observed in \cite{Oliynyk:2021} for the relativistic Euler equations on fixed de Sitter backgrounds. In that work, it was shown that the structure of the singular terms could be improved by spatially differentiating the equations and renormalising the spatial derivatives by appropriate powers of $t$, while using the unsymmetrised equations for the undifferentiated variables. In particular,  \cite{Oliynyk:2021} demonstrates that one can trade regularity for improved singular structure. \newline \par

However, in the coupled Einstein-Euler system, this approach is not optimal because renormalising the first derivatives leads to unnecessary restrictions on the fluid sound speeds $K_{(\afrak)}$. This is because we now have to differentiate both the fluid and gravitational variables. In particular, upon differentiating in space, certain terms involving products of gravitational and fluid variables become \textit{quadratic} in the re-scaled first derivatives resulting in terms that are too singular to apply the Fuchsian theory from \cites{Oliynyk:CMP_2016,BOOS:2021,BeyerOliynyk:2024b}. To overcome this, we instead differentiate twice in space. The second derivatives, which appear \textit{linearly}, are then renormalised. This leads to the system \eqref{eqn:EinsteinEuler_NoDeriv}-\eqref{eqn:EinsteinEuler_2ndDeriv}, which is better suited for establishing global existence over a larger range of sound speeds. \newline \par

In Section \ref{sec:Coefficient_Assumptions}, we verify that the twice differentiated system \eqref{eqn:EinsteinEuler_NoDeriv}-\eqref{eqn:EinsteinEuler_2ndDeriv} possesses a favourable singular structure by demonstrating that it satisfies the conditions outlined in \cite{BOOS:2021}*{\S 3.4} for the range $\frac{1}{3} < K_{(1)} \leq K_{(2)} < \frac{5}{7}$. This verification is a crucial step, as it enables the application of the Fuchsian global existence theory developed in \cite{BOOS:2021}*{\S 3.4} to our system. As a result, we are able to establish the main stability result of this article stated in Theorem \ref{mainthm}, whose proof is given in Section \ref{sec:MainTheoremProof}. \newline \par

\begin{rem}
    It is important to emphasise that the original \textit{unsymmetrised} Einstein-Euler system \eqref{eqn:frame_form1}, \eqref{eqn:Hct_evo}, \eqref{eqn:alphat_constraintmodevo}, \eqref{eqn:Udot_constraintmodevo}, \eqref{eqn:A_PDE3}, \eqref{eqn:N_evo_2}, \eqref{eqn:Sigma_evo_2}, and \eqref{eqn:Euler2+1_1}-\eqref{eqn:Euler2+1_3} has good ODE structure for global existence. However, in order to apply the Fuchsian theorem, we need a symmetric hyperbolic form of the equations. In particular, symmetrising the Euler equations requires us to rescale certain fluid variables which, in turn, is directly responsible for the singular term $\frac{-\mu}{t}\Pi W$. As mentioned above, we ultimately differentiate the Euler equations in space to resolve this issue. Specifically, as discussed in Section \ref{sec:Differentiated_Euler_deriv}, we treat the zeroth order equations \eqref{eqn:Fuchsian_noderiv} as ODEs and `remove' the re-scaling which was necessary for symmetrisation. It is unclear whether there exist other symmetrisations of the equations that impose less restrictive conditions on the sound speeds and, in particular, whether there exists a symmetrisation that leads to a good Fuchsian formulation over the range $\frac{1}{3} < K_{(1)} \leq K_{(2)} < 1$. We are currently investigating this question.
\end{rem}

\subsubsection{Existence of Perturbed Initial Data}
We conclude the article by establishing the existence of an open set of initial data around tilted and homogeneous two fluid solutions of the constraint equations in Section \ref{sec:PerturbedInitialData}.

\subsection{Indexing Conventions}
\label{sec:indexing}
In this article, we restrict our attention to spacetimes of the form $M = (0,T_{1}]\times \mathbb{T}^{3}$. On $\mathbb{T}^3$, we use $(x^{\Omega})$, $(\Omega = 1,2,3)$ to denote periodic spatial coordinates, while $t = x^{0}$ will be used to denote a Cartesian time coordinate on the interval $(0,T_{1}]$ where $T_{1}>0$. 
Lowercase Greek letters e.g. $\mu$, $\nu$, $\gamma$, run from $0$ to $3$ and label spacetime coordinate indices, while uppercase Greek indices e.g. $\Sigma$, $\Omega$, $\Gamma$, run from $1$ to $3$ and label spatial coordinate indices. 
Moreover, the future lies in the direction of \textit{decreasing} $t$ and future timelike infinity is located at $t=0$. \newline \par 

We also employ frames $e_{a} = e^{\mu}_{a}\del_{\mu}$, and  lower case Latin letters, e.g. $a$, $b$, $c$, that run from $0$ to $3$ will label spacetime frame indices while spatial frame indices will be labelled by upper case
Latin letters, e.g. $A$, $B$, $C$, and run from $1$ to $3$. Fraktur font indices, e.g. $\afrak, \mathfrak{b}, \mathfrak{c}$, are used to label the fluids and run from $1$ to $2$. 

\subsubsection{Index Operations}
The \textit{symmetrisation}, \textit{anti-symmetrisation}, and \textit{symmetric trace free} operations on pairs of spatial frame indices are defined by
\begin{gather*}
    \Sigma_{(AB)} = \frac{1}{2}(\Sigma_{AB} + \Sigma_{BA}), \quad \Sigma_{[AB]} = \frac{1}{2}(\Sigma_{AB}-\Sigma_{BA}),\intertext{and}
    \Sigma_{\langle AB \rangle} = \Sigma_{(AB)}- \frac{1}{3}\delta^{CD}\Sigma_{CD}\delta_{AB},
\end{gather*}
respectively.

\subsection*{Ethics Declarations}
The authors declare no conflict of interests.

\subsection*{Acknowledgements}
G.F. and E.M. gratefully acknowledge the support of the ERC starting grant 101078061 SINGinGR, under the European
Union’s Horizon Europe program for research and innovation.

\section{Frame Formulation}
\label{sec:Frame_Formulation}
As discussed in the introduction, our global existence result relies on the initial value formulation obtained by expressing the conformal Einstein-Euler equations in terms of an orthonormal frame. These equations, presented below, are based on the formulation derived in 
\cite{RöhrUggla:2005}, see also the closely related works \cites{UEWE:2003,ElstUggla:1997}. \newline \par 

To begin, we introduce an orthonormal frame $e_{a}$ with respect to the conformal metric $g$, defined above by \eqref{eqn:conformal_metric_def}, via
\begin{align} \label{e-def}
    e_{0} = \alpha^{-1}\del_{t}, \quad e_{A} = e^{\Sigma}_{A}\del_{\Sigma},
\end{align}
where $e_{0}$ is assumed to be hypersurface orthogonal, $\alpha$ is the lapse associated with the conformal metric, the shift vector vanishes, and the spatial frame $e_{A}$ is tangent to the $t=\textrm{constant}$ hypersurfaces.  Expressed in this frame, the frame components of the conformal metric are given by
\begin{align*}
    g_{ab}:= g(e_{a},e_{b}) = \eta_{ab}
\end{align*}
where $\eta_{ab}$ is the Minkowski metric\footnote{Frame indices will always be raised and lowered using the conformal frame metric $g_{ab} = \eta_{ab}$.}, and we can express the conformal metric as
\begin{equation*}
g= -\alpha^2 dt\otimes dt + \delta_{AB}\theta^A\otimes\theta^B 
\end{equation*}
where $\theta^B=\theta^B_\Lambda dx^\Lambda$ is the dual spatial frame, i.e.~$\theta^B_\Lambda e^\Lambda_A=\delta^B_A$.
We always assume that the spatial frame $e_{A}$ is evolved via Fermi-Walker transport, that is,
\begin{align*}
\nabla_{e_{0}}e_{A} = \frac{-g(\nabla_{e_{0}}e_{0},e_{A})}{g(e_{0},e_{0})}e_{0}
\end{align*}
where the spatial frame $e_{A}$ is chosen initially to satisfy 
\begin{align*}
g(e_{0},e_{A})|_{t_{0}} = 0 \quad \text{and} \quad  g(e_{A},e_{B})|_{t_{0}} = \delta_{AB}.
\end{align*}

\subsection{Connection and Commutation Coefficients}
The connection coefficients $\tensor{\omega}{_a^b_c}$ of the conformal metric associated to the orthonormal frame $e_a$ are defined by
\begin{equation*}
\begin{aligned}
\nabla_{e_{a}}e_{b} = \tensor{\omega}{_a^c_b}e_{c},
\end{aligned}
\end{equation*}
and we set
\begin{equation*}
\omega_{acb} = \eta_{cd}\tensor{\omega}{_a^d_b}.
\end{equation*}
Similarly, the Lie brackets $[e_a,e_b]$ determine the commutator coefficients $\tensor{c}{_a^c_b}$ by
\begin{align*}
    [e_{a},e_{b}] = \tensor{c}{_a^c_b}e_{c}.
\end{align*}
Since the frame $e_{i}$ is orthonormal, the connection coefficients satisfy
\begin{equation*}
\omega_{abc}=-\omega_{acb}.
\end{equation*}
Moreover, the connection and commutator coefficients uniquely determine each other via the relations
\begin{align*}
    \tensor{c}{_a^c_b} = \tensor{\omega}{_a^c_b} -\tensor{\omega}{_b^c_a}, \quad \tensor{\omega}{_{abc}} = \frac{1}{2}(c_{bac}-c_{cba}-c_{acb}).
\end{align*}

In the notation of \cite{RöhrUggla:2005}, our gauge choices (hypersurface orthogonal normal vector $e_{0}$ with a zero-shift coordinate gauge and Fermi-Walker transported spatial frame) correspond to setting 
 \begin{align}
 \label{eqn:RöhrUggla_GaugeSimplifications}
     M_{i}=\mathcal{M}_{\alpha}=W_{\alpha}=R_{\alpha}=0.
 \end{align}
 With these choices, it follows from \cite{RöhrUggla:2005}*{Eqns.~(25)-(26)} that the commutators $[e_{0},e_{A}]$ and $[e_{A},e_{B}]$ can be expressed as 
\begin{equation*}
\begin{aligned}
    [e_{0},e_{A}] =  \dot{U}_{A}e_{0} - (\mathcal{H}\delta^{B}_{A} + \tensor{\Sigma}{_A^B})e_{B}, \\
    [e_{A},e_{B}] = (2A_{[A}\delta^{C}_{B]}+\eps_{ABD}N^{DC})e_{C},
\end{aligned}
\end{equation*}
where $N_{AB}$ is symmetric and $\Sigma_{AB}$ is symmetric and trace-free, that is,
\begin{align*}
    N_{(AB)} = N_{AB} \quad \text{and} \quad \Sigma_{\langle AB \rangle} = \Sigma_{AB}.
\end{align*}
Similarly, from \cite{RöhrUggla:2005}*{Eqns.~(35)-(36)}, the connection coefficients are given by
\begin{equation}
\begin{gathered}
\label{eqn:connection_identities}
\omega_{0A0} = -\omega_{00A} = \dot{U}_{A}, \quad \omega_{A0B} = -\omega_{AB0} = -(\mathcal{H}\delta_{AB} + \Sigma_{AB}), \\
\omega_{ABC} = 2A_{[B}\eta_{C]A} + \epsilon_{BCE}\tensor{\tilde{N}}{^{E}_{A}}, \quad \tensor{\tilde{N}}{^E_A} = \tensor{N}{^E_A} -\frac{1}{2}\tensor{N}{^D_D}\delta^{E}_{A}, \\
\omega_{000} = \omega_{A00} = \omega_{0AB} = 0.
\end{gathered}
\end{equation}
For use below, we set
\begin{align} \label{ra-def}
   r_{a} := e_{a}(\Phi) = -\frac{1}{\alpha t}\delta_a^0,
\end{align}
where the second equality follows from \eqref{Phi-def} and \eqref{e-def}.

\subsection{Frame Equations}
In \cite{RöhrUggla:2005}, it is shown that the conformal Einstein equations, frame commutators, and the Jacobi identities lead to a system of evolution and constraint equations. For our gauge choices \eqref{eqn:RöhrUggla_GaugeSimplifications}, these equations, see \cite{RöhrUggla:2005}*{Eqns.~(41)-(53)}\footnote{Note, there is a sign error in the term $\epsilon^{\gamma\delta}{}_{\langle \alpha}(2\Sigma_{\beta\rangle \gamma} R_\delta - N_{\beta\rangle \gamma}\dot{U}_\delta)$ from \cite{RöhrUggla:2005}*{Eqn.~(45)}. The `$-$' sign in front of this term should be replaced with a `$+$' sign; cf. \cite{ElstUggla:1997}*{Eqn.~(34)}.}, reduce to the evolution equations 
\begin{align}
\label{eqn:frame_component_evo_intro}
e_{0}(e_{A}^{\Sigma}) &= -(\mathcal{H}\delta_{A}^{B} + \Sigma_{A}^{B})e_{B}^{\Sigma}, \\
e_{0}(A_{A}) &= -e_{A}(\mathcal{H}) + \frac{1}{2}e_{B}(\tensor{\Sigma}{_{A}^{B}}) - \mathcal{H}(\dot{U}_{A}+A_{A}) + \tensor{\Sigma}{_A^B}(\frac{1}{2}\dot{U}_{B}-A_{B}), \\
e_{0}(N^{AB}) &= -\mathcal{H}N^{AB} + 2\tensor{N}{^{(A}_{C}}\Sigma^{B)C} - \epsilon^{CD(A}\Big[e_{C}(\tensor{\Sigma}{_D^{B)}})+\dot{U}_{C}\tensor{\Sigma}{_D^{B)}}\Big], \\
\label{eqn:Hcal_Evo_Intro}
e_{0}(\mathcal{H}) &= -\mathcal{H}^{2}+\frac{1}{3}e_{A}(\dot{U}^{A}) + \frac{1}{3}\dot{U}_{A}(\dot{U}^{A}-2A^{A})-\frac{1}{3}\Sigma_{AB}\Sigma^{AB}\nonumber \\
&-\frac{e^{2\Phi}}{6}(T_{00}+\tensor{T}{_A^A})+\frac{e^{2\Phi}\Lambda}{3} -\frac{1}{3}U_{00}, \\
\label{eqn:Sigma_evo_intro}
e_{0}(\Sigma_{AB}) &= -3\mathcal{H}\Sigma_{AB} +e_{\langle A}(\dot{U}_{B \rangle}) +\dot{U}_{\langle A}\dot{U}_{B \rangle} -e_{\langle A}(A_{B \rangle})+A_{\langle A}\dot{U}_{B \rangle} \nonumber \\
&+\epsilon_{CD(A}\Big[e^{C}(\tensor{N}{_{B)}^D})+\tensor{N}{_{B)}^D}\dot{U}^{C}-2\tensor{N}{_{B)}^D}A^{C}\Big] -2\tensor{N}{_{\langle A}^C}\tensor{N}{_{B \rangle}_{C}} +NN_{\langle AB \rangle} \nonumber \\
&+e^{2\Phi}T_{\langle AB \rangle} + U_{\langle AB \rangle}, 
\end{align}
and constraint equations
\begin{align}
0 &= \mathcal{C}_{1} := 2e_{[A}(e_{B]}^{\Sigma}) - 2A_{[A}e^{\Sigma}_{B]} - \epsilon_{ABD}N^{DC}e^{\Sigma}_{C}, \\
0 &= \mathcal{C}_{2} := e_{A}(\tensor{N}{^{AD}}) + \tensor{\eps}{^{ABD}}e_{A}(A_{B}) -2A_{A}N^{AD}, \\
\label{eqn:C3_constraint_intro}
0 &= \mathcal{C}_{3} := \dot{U}_{A} - \alpha^{-1}e_{A}(\alpha), \\
\label{eqn:C4_constraint_intro}
0 &= \mathcal{C}_{4} :=2e_{[B}(\dot{U}_{A]}) + 2A_{[A}\dot{U}_{B]} + \epsilon_{ABC}N^{CD}\dot{U}_{D}, \\
0 &= \mathcal{C}_{M} := e_{B}(\tensor{\Sigma}{_A^B}) -2e_{A}(\mathcal{H}) -3\tensor{\Sigma}{_A^B}A_{B} - \epsilon_{ABC}N^{BD}\tensor{\Sigma}{_D^C} - e^{2\Phi}T_{0A} -U_{0A}, \label{eqn:Momentum_Constraint_intro} \\
\label{eqn:Hamiltonian_Constraint_intro}
0 &= \mathcal{C}_{H} := 4e_{A}(A^{A}) +6\mathcal{H}^{2} - 6A^{A}A_{A} -N^{AB}N_{AB} +\frac{1}{2}N^{2} -\Sigma_{AB}\Sigma^{AB} \nonumber \\
&-2e^{2\Phi}T_{00}-2e^{2\Phi}\Lambda -\big(U_{00}+\tensor{U}{_A^A}\big), 
\end{align}
where $\mathcal{C}_{M}$ and $\mathcal{C}_{H}$ are the momentum and Hamiltonian constraints, respectively, and $U_{ab}$ is defined by\footnote{Note that our definition of $U_{ab}$ is the negative of the definition given in \cite{RöhrUggla:2005}.}
\begin{align}
\label{eqn:U_conf_def}
U_{ab} := 2\nabla_{a}\nabla_{b}\Phi +g_{ab}(\Box\Phi +2|\nabla\Phi|^{2}_{g}) -2\nabla_{a}\Phi\nabla_{b}\Phi.
\end{align}
\begin{rem}
    For our gauge choices \eqref{eqn:RöhrUggla_GaugeSimplifications}, equations (41), (42), (46), and (53) from \cite{RöhrUggla:2005} simplify in the following manner: (i) The evolution equation (41) for the shift $\mathcal{M}_{\alpha}$ becomes the constraint  \eqref{eqn:C3_constraint_intro}. (ii)  Equations (42) and (53) identically vanish.
    (iii) The evolution equation (46) for the vorticity  $W_{\alpha}$ becomes the constraint \eqref{eqn:C4_constraint_intro}.
\end{rem}

\subsection{Fixing the Lapse}
\label{sec:lapse_fixing}
The choice of lapse $\alpha$ is the last remaining gauge freedom. We fix it by requiring our time coordinate $t$ to satisfy the generalised harmonic time slicing condition
\begin{align}
\label{eqn:Harmonic_TimeSlicing}
    \Box_g t = f,
\end{align}
where, for now, $f$ is an arbitrary function. This leads to the following evolution equation for $\alpha$:
\begin{align}
\label{eqn:alpha_evo}
e_{0}(\alpha) &= 3\mathcal{H}\alpha +\alpha^{2}f.
\end{align}
Since $\dot{U}_{A}$ is determined by the lapse via the constraint \eqref{eqn:C4_constraint_intro}, the above choice of lapse leads to an evolution equation for $\dot{U}_{A}$. In particular, by allowing the frame commutator $[e_0,e_A]$ to act on $\alpha$, we obtain
\begin{align}
\label{eqn:Udot_evo}
e_{0}(\dot{U}_{A}) &= 3e_{A}(\mathcal{H}) + \alpha e_{A}(f) + 2\mathcal{H}\dot{U}_{A} +2\alpha\dot{U}_{A}f -\tensor{\Sigma}{^B_A}\dot{U}_{B}. 
\end{align}

\subsection{Gauge-Reduced Einstein Field Equations}
Combining \eqref{eqn:frame_component_evo_intro}-\eqref{eqn:Hamiltonian_Constraint_intro} and \eqref{eqn:alpha_evo}-\eqref{eqn:Udot_evo} yields a complete gauge-reduced system of evolution and constraint equations. Specialising to the conformal factor \eqref{Phi-def}, the evolution equations are given by 
\begin{align}
\label{eqn:frame_form1}
e_{0}(e_{A}^{\Sigma}) &= -(\mathcal{H}\delta_{A}^{B} + \Sigma_{A}^{B})e_{B}^{\Sigma}, \\
\label{eqn:A_evo}
e_{0}(A_{A}) &= -e_{A}(\mathcal{H}) + \frac{1}{2}e_{B}(\tensor{\Sigma}{_{A}^{B}}) - \mathcal{H}(\dot{U}_{A}+A_{A}) + \tensor{\Sigma}{_A^B}(\frac{1}{2}\dot{U}_{B}-A_{B}), \\
\label{eqn:N_evo}
e_{0}(N^{AB}) &= -\mathcal{H}N^{AB} + 2\tensor{N}{^{(A}_{C}}\Sigma^{B)C} - \epsilon^{CD(A}\Big[e_{C}(\tensor{\Sigma}{_D^{B)}})+\dot{U}_{C}\tensor{\Sigma}{_D^{B)}}\Big], \\
\label{eqn:H_evo}
e_{0}(\mathcal{H}) &= -\mathcal{H}^{2}+\frac{1}{3}e_{A}(\dot{U}^{A}) + \frac{1}{3}\dot{U}_{A}(\dot{U}^{A}-2A^{A})-\frac{1}{3}\Sigma_{AB}\Sigma^{AB}\nonumber \\
&-\frac{1}{6t^{2}}(T_{00}+\tensor{T}{_A^A})+\frac{\Lambda}{3t^{2}} -\frac{2\mathcal{H}}{\alpha t}-\frac{1}{\alpha^{2}t^{2}}-\frac{f}{t}, \\
\label{eqn:Sigma_evo}
e_{0}(\Sigma_{AB}) &= -3\mathcal{H}\Sigma_{AB} +e_{\langle A}(\dot{U}_{B \rangle}) +\dot{U}_{\langle A}\dot{U}_{B \rangle} -e_{\langle A}(A_{B \rangle})+A_{\langle A}\dot{U}_{B \rangle} \nonumber \\
&+\epsilon_{CD(A}\Big[e^{C}(\tensor{N}{_{B)}^D})+\tensor{N}{_{B)}^D}\dot{U}^{C}-2\tensor{N}{_{B)}^D}A^{C}\Big] -2\tensor{N}{_{\langle A}^C}\tensor{N}{_{B \rangle}_{C}} +NN_{\langle AB \rangle} \nonumber \\
&+\frac{1}{t^{2}}T_{\langle AB \rangle} + \frac{2}{\alpha t}\Sigma_{AB}, \\
\label{eqn:Udot_evo_f}
e_{0}(\dot{U}_{A}) &= 3e_{A}(\mathcal{H}) + \alpha e_{A}(f) + 2\mathcal{H}\dot{U}_{A} +2\alpha\dot{U}_{A}f -\tensor{\Sigma}{^B_A}\dot{U}_{B}, \\
\label{eqn:alpha_evo_2}
e_{0}(\alpha) &= 3\mathcal{H}\alpha +\alpha^{2}f,
\end{align}
where explicit expressions for the components of $U_{ab}$ in terms of the connection coefficients are given in Appendix \ref{sec:ConformalFactor_Appendix}.
Similarly, the constraint equations are given by 
\begin{align}
\label{eqn:C1_constraint}
0 &= \mathcal{C}_{1} :=  2e_{[A}(e_{B]}^{\Sigma}) - 2A_{[A}e^{\Sigma}_{B]} - \epsilon_{ABD}N^{DC}e^{\Sigma}_{C}, \\
\label{eqn:C2_constraint}
0 &= \mathcal{C}_{2} :=e_{A}(N^{AD}) - \epsilon^{ADB}e_{A}(A_{B}) - 2A_{A}N^{AD}, \\
\label{eqn:C3_constraint}
0 &= \mathcal{C}_{3} :=\dot{U}_{A} - \alpha^{-1}e_{A}(\alpha), \\
\label{eqn:Udot_commutatorconstraint}
0 &= \mathcal{C}_{4} :=2e_{[B}(\dot{U}_{A]}) + 2A_{[A}\dot{U}_{B]} + \epsilon_{ABC}N^{CD}\dot{U}_{D}, \\
\label{eqn:MomentumConstraint}
0 &= \mathcal{C}_{M} := e_{B}(\tensor{\Sigma}{_A^B}) -2e_{A}(\mathcal{H}) -3\tensor{\Sigma}{_A^B}A_{B} - \epsilon_{ABC}N^{BD}\tensor{\Sigma}{_D^C} - \frac{1}{t^{2}}T_{0A} - \frac{2}{\alpha t}\dot{U}_{A}, \\
\label{eqn:Hamiltonian_Constraint}
0 &= \mathcal{C}_{H} :=4e_{A}(A^{A}) +6\mathcal{H}^{2} - 6A^{A}A_{A} -N^{AB}N_{AB} +\frac{1}{2}N^{2} -\Sigma_{AB}\Sigma^{AB} \nonumber \\
&\hspace{1.5cm}-\frac{2}{t^{2}}T_{00}-\frac{2}{t^{2}}\Lambda -\frac{12\mathcal{H}}{\alpha t} + \frac{6}{\alpha^{2}t^{2}}.
\end{align}

\section{A Symmetric Hyperbolic Formulation of The Einstein Equations}
\label{sec:Sym_Hyp_Einstein}

\subsection{Gauge Source}
\label{sec:Sec3_3:Step1}
Motivated by the stability results from \cite{Oliynyk:CMP_2016}, we fix the gauge source function by setting
\begin{align}
\label{eqn:gauge_source_f}
    f = \frac{2}{t}\Bigl(-\frac{1}{\alpha^{2}}+\frac{\Lambda}{3}\Bigr).
\end{align}
\begin{rem}
The role of the harmonic slicing condition \eqref{eqn:Harmonic_TimeSlicing} and gauge source \eqref{eqn:gauge_source_f} is to drive the lapse towards the de Sitter background lapse. In exponentially expanding spacetimes, we expect to approach de Sitter space at late times (cf. Appendix \ref{Appendix:two-fluid}), so it is natural to use such a gauge in the context of stability results. 
\end{rem}

Next, we introduce new variables $\Hct$ and $\alphat$ via  
\begin{equation}
\begin{aligned}
\label{eqn:Hct_Alphat_defs}
    \Hct &= \mathcal{H} - \frac{1}{3}\alphah, \quad \alphat = 2\mathcal{H} - \alphah,
\end{aligned}
\end{equation}
where
\begin{align*}
    \alphah = \frac{1}{t}\Big(\frac{1}{\alpha}-\frac{\alpha\Lambda}{3}\Big).
\end{align*}
\begin{rem}
   This choice of new variables is motivated by the fact that working with $\alphat$ and $\Hct$, as opposed to $\alpha$ and $\mathcal{H}$, results in the correct singular structure necessary to apply the Fuchsian theory from \cite{BOOS:2021}. 
\end{rem}
\noindent We observe  that $\mathcal{H}$ and $\alphah$ can be recovered from $\Hct$ and $\alphat$ using
\begin{equation}
\begin{aligned}
\label{eqn:Hct_identities}
    \mathcal{H}= -\alphat +3\Hct \quad \text{and} \quad  \alphah=-3\alphat +6\Hct,
\end{aligned}
\end{equation}
while $\alpha$ can be recovered via
\begin{align}
\label{eqn:lapse_Hct_identity}
    \alpha &= \frac{\Big(3\alphat -6\Hct\Big) t + \sqrt{\Big(\big(-3\alphat +6\Hct\big)  t\Big)^{2}+\frac{4\Lambda}{3}}}{\frac{2\Lambda}{3}}.
\end{align}
After a lengthy but straightforward calculation, we obtain the following evolution equations for $\Hct$ and $\alphat$:
\begin{align}
\label{eqn:Hct_evo}
e_{0}(\Hct) &=\frac{1}{3}\delta^{AB}e_{A}(\dot{U}_{B})+ \frac{1}{3}\dot{U}_{A}(\dot{U}^{A}-2A^{A})-\frac{1}{3}\Sigma_{AB}\Sigma^{AB}-\Hct^{2} -\frac{1}{6t^{2}}(T_{00}+\tensor{T}{_A^A}) \nonumber \\
&-\frac{5}{3}\alphah\Hct +\frac{2}{9}\alphah^{2}, \\
\label{eqn:alphat_evo}
    e_{0}(\alphat) &= \frac{2}{3}\delta^{AB}e_{A}(\dot{U}_{B})+\frac{2}{3}\dot{U}_{A}\Big(\dot{U}_{A}-2A^{A}\Big)-\frac{2}{3}\Sigma_{AB}\Sigma^{AB}-\frac{\alphat^{2}}{2}-\frac{1}{3t^{2}}\Big(T_{00}+\tensor{T}{_A^A}\Big) \nonumber \\
    & +\frac{\alpha\Lambda}{3t}\alphat-\frac{3}{2}\alphah\alphat.
\end{align}

\subsection{Addition of Constraints}
\label{sec:Sec3_3:Step2}
In order to bring the evolution equations \eqref{eqn:frame_form1}-\eqref{eqn:N_evo}, \eqref{eqn:Sigma_evo}-\eqref{eqn:Udot_evo_f} and \eqref{eqn:Hct_evo}-\eqref{eqn:alphat_evo} into a symmetric hyperbolic form, we will need to add suitable multiples of the constraints to the evolution equations \eqref{eqn:A_evo}, \eqref{eqn:Udot_evo_f} and \eqref{eqn:alphat_evo}. To this end, we express the Hamiltonian constraint quantity $\mathcal{C}_{H}$, see \eqref{eqn:Hamiltonian_Constraint}, in terms of $\alphat$ and $\Hct$
as follows 
\begin{align}
    \label{eqn:Hamiltonian_Constraint2}
     \mathcal{C}_{H} &= 4e_{A}(A^{A}) +6(-\alphat+3\Hct)^{2} - 6A^{A}A_{A} -N^{AB}N_{AB} +\frac{1}{2}N^{2} -\Sigma_{AB}\Sigma^{AB} \nonumber \\
&-\frac{2}{t^{2}}T_{00}-\frac{6\alphat}{\alpha t}.
\end{align}
Similarly, using the identity 
\begin{align*}
    \alpha^{-1}e_{A}(\alpha) = -\frac{\alpha^{-1}e_{A}(\alphah)}{\frac{1}{t\alpha^{2}}+\frac{\Lambda}{3t}}
\end{align*}
and \eqref{eqn:Hct_identities}, the constraint quantity $\mathcal{C}_3$ (defined by \eqref{eqn:C3_constraint}) can be expressed as
\begin{align*}
   0 =  \tilde{\mathcal{C}}_{3} := \Big(\frac{1}{t\alpha}+\frac{\Lambda\alpha}{3t}\Big)\dot{U}_{A}+3e_{A}(\alphat)-6e_{A}(\Hct).
\end{align*}
\newline \par
Now, by adding $\theta\mathcal{C}_{H}$ to \eqref{eqn:alphat_evo} we obtain the equation
\begin{align}
\label{eqn:alphat_constraintmodevo}
    e_{0}(\alphat) &= \frac{2}{3}\delta^{AB}e_{A}(\dot{U}_{B})+ 4\theta e_{A}(A^{A})+\frac{2}{3}\dot{U}_{A}\Big(\dot{U}_{A}-2A^{A}\Big)-(\frac{2}{3}+\theta)\Sigma_{AB}\Sigma^{AB}\nonumber \\
    &-\frac{\alphat^{2}}{2}-\frac{1}{3t^{2}}\Big(T_{00}+\tensor{T}{_A^A}\Big) +\Big(\frac{\alpha\Lambda}{3t}-\frac{6\theta}{\alpha t}\Big)\alphat-\frac{3}{2}\alphah\alphat  \nonumber \\
    &+6\theta(-\alphat+3\Hct)^{2} - 6\theta A^{A}A_{A} -\theta N^{AB}N_{AB} +\frac{\theta}{2}N^{2} -\frac{2\theta}{t^{2}}T_{00},
\end{align} 
where $\theta$ is a constant to be determined. Similarly, by adding  $\lambda\mathcal{C}_{H}$ and $\beta\tilde{\mathcal{C}}_{3}$ to \eqref{eqn:Udot_evo} and $\gamma\mathcal{C}_{H}$ and $\kappa \tilde{\mathcal{C}_{3}}$ to \eqref{eqn:A_evo}, we obtain the evolution equations
\begin{align}
    \label{eqn:Udot_constraintmodevo}
    e_{0}(\dot{U}_{A}) &= \Big(-(3-2\lambda)+3\beta\Big)e_{A}(\alphat) +\Big(3(3-2\lambda)-6\beta\Big)e_{A}(\Hct) +\lambda \delta^{BC}\pi^{EF}_{AC}e_{B}(\tensor{\Sigma}{_{EF}}) \nonumber \\
    &+ 2(-\alphat+3\Hct)\dot{U}_{A} - \tensor{\Sigma}{^B_A}\dot{U}_{B} +\Big(\frac{\beta-2\lambda}{t\alpha}+\frac{(4+\beta)\alpha\Lambda}{3t}\Big)\dot{U}_{A} \nonumber \\
    &+\lambda\Big(-3\tensor{\Sigma}{_A^B}A_{B} - \epsilon_{ABC}N^{BD}\tensor{\Sigma}{_D^C} - \frac{1}{t^{2}}T_{0A} \Big), \\
    \label{eqn:A_PDE3}
    e_{0}(A_{A}) &= \Big(-(-1-2\gamma)+3\kappa\Big)e_{A}(\alphat) + \Big(3(-1-2\gamma)-6\kappa\Big)e_{A}(\Hct) -3\gamma\tensor{\Sigma}{_A^B}A_{B} \nonumber \\
    &+ (\frac{1}{2}+\gamma)\delta^{BC}\pi^{EF}_{AC}e_{B}(\tensor{\Sigma}{_{EF}}) -(-\alphat+3\Hct)(\dot{U}_{A}+A_{A}) + \tensor{\Sigma}{_A^B}(\frac{1}{2}\dot{U}_{B}-A_{B})  \nonumber \\
    &-\gamma\epsilon_{ABC}N^{BD}\tensor{\Sigma}{_D^C}-\gamma \frac{1}{t^{2}}T_{0A} + \Big(\frac{\kappa-2\gamma}{t\alpha}+\frac{\kappa\Lambda\alpha}{3t}\Big)\dot{U}_{A}, 
\end{align}
respectively, where $\lambda$, $\beta$, $\gamma$, and $\kappa$ are constants to be determined. \newline \par

Finally, we turn to expressing the evolution equations for 
$N_{AB}$ and $\Sigma_{AB}$ in a form that clearly reveals their symmetric hyperbolic structure. To this end, we define the symmetric trace-free projection operator
\begin{equation}
\label{eqn:pi_projection_tracefree}
    \pi^{CD}_{AB} = \delta^{C}_{(A}\delta^{D}_{B)} - \frac{1}{3}\delta^{CD}\delta_{AB}, 
\end{equation}
i.e.~$\pi^{CD}_{AB}X_{CD} = X_{\langle AB \rangle}$.
Using $\pi^{DE}_{AB}$, we can express \eqref{eqn:N_evo} and \eqref{eqn:Sigma_evo} as 
\begin{align*}
    e_{0}(N_{AB}) &= -\tensor{\epsilon}{^{C(D}_{(A}}\delta^{E)}_{B)}e_{C}\big(\Sigma_{DE}\big)-\tensor{\epsilon}{^{C(D}_{(A}}\delta^{E)}_{B)}\dot{U}_{C}\Sigma_{DE} -\mathcal{H}N_{AB}+2N_{C(A}\tensor{\Sigma}{_{B)}^C}, \\
    e_{0}(\Sigma_{AB}) &= \tensor{\epsilon}{^{C(D}_{(A}}\delta^{E)}_{B)}e_{C}(N_{DE}) +\pi^{DE}_{AB}e_{D}(\dot{U}_{E}) - \pi^{DE}_{AB}e_{D}(A_{E}) \nonumber \\
    &\quad -3\mathcal{H}\Sigma_{AB} +\epsilon_{CD(A}\tensor{N}{_{B)}^D}\dot{U}^{C}-2\tensor{\epsilon}{_{CD(A}}\tensor{N}{_{B)}^D}A^{C}\nonumber \\ &\quad +\pi^{DE}_{AB}\Big(\dot{U}_{D}\dot{U}_{E}+A_{D}\dot{U}_{E}-2\tensor{N}{_D^C}N_{EC}+NN_{DE}+\frac{1}{t^{2}}T_{DE}+U_{DE}\Big).
\end{align*}
Setting
\begin{align}
\label{eqn:Scal_operator_defn}
    \tensor{\mathcal{S}}{^C_{AB}^{DE}} &:= -\tensor{\epsilon}{^{C(D}_{(A}}\delta^{E)}_{B)} = -\frac{1}{4}\Big(\tensor{\epsilon}{^{CD}_{B}}\delta^{E}_{A}+\tensor{\epsilon}{^{CD}_{A}}\delta^{E}_{B}+\tensor{\epsilon}{^{CE}_{A}}\delta^{D}_{B}+\tensor{\epsilon}{^{CE}_{B}}\delta^{D}_{A}\Big),
\end{align}
a straightforward calculation shows
\begin{align*}
    \tensor{\mathcal{S}}{^C_{ABDE}} = -\frac{1}{4}\Big(\tensor{\epsilon}{^{C}_{DB}}\delta_{EA}+\tensor{\epsilon}{^{C}_{DA}}\delta_{EB}+\tensor{\epsilon}{^{C}_{EA}}\delta_{DB}+\tensor{\epsilon}{^{C}_{EB}}\delta_{DA}\Big).
\end{align*}
Swapping the index pairs $AB$ and $DE$
\begin{align*}
    \tensor{\mathcal{S}}{^C_{DEAB}} = \frac{1}{4}\Big(\tensor{\epsilon}{^{C}_{EA}}\delta_{BD}+\tensor{\epsilon}{^{C}_{DA}}\delta_{BE}+\tensor{\epsilon}{^{C}_{DB}}\delta_{AE}+\tensor{\epsilon}{^{C}_{EB}}\delta_{AD}\Big)
\end{align*}
then implies 
\begin{align*}
    \tensor{\mathcal{S}}{^{CDE}_{AB}} &= \frac{1}{2}\Big(\tensor{\epsilon}{^{CD}_{(A}}\delta^{E}_{B)} + \tensor{\epsilon}{^{CE}_{(A}}\delta^{D}_{B)}\Big)=\tensor{\epsilon}{^{C(D}_{(A}}\delta^{E)}_{B)}.
\end{align*}
Now, using $\mathcal{S}$, we can express the evolution equations for $N_{AB}$ and $\Sigma_{AB}$ as
\begin{align}
    \label{eqn:N_evo_2}
    e_{0}(N_{AB}) &= \tensor{\mathcal{S}}{^C_{AB}^{DE}}e_{C}(\Sigma_{DE})+\tensor{\mathcal{S}}{^C_{AB}^{DE}}\dot{U}_{C}\Sigma_{DE} -\mathcal{H}N_{AB}+2N_{C(A}\tensor{\Sigma}{_{B)}^C}, \\
    \label{eqn:Sigma_evo_2}
    e_{0}(\Sigma_{AB}) &= \tensor{\mathcal{S}}{^{CDE}_{AB}}e_{C}(N_{DE}) +\pi^{DE}_{AB}e_{D}(\dot{U}_{E}) - \pi^{DE}_{AB}e_{D}(A_{E}) \nonumber \\
    &-3\mathcal{H}\Sigma_{AB} +\epsilon_{CD(A}\tensor{N}{_{B)}^D}\dot{U}^{C}-2\tensor{\epsilon}{_{CD(A}}\tensor{N}{_{B)}^D}A^{C}\nonumber \\ &+\pi^{DE}_{AB}\Big(\dot{U}_{D}\dot{U}_{E}+A_{D}\dot{U}_{E}-2\tensor{N}{_D^C}N_{EC}+NN_{DE}+e^{2\Phi}T_{DE}+U_{DE}\Big). 
\end{align}

\subsection{Symmetric Hyperbolic System}
\label{sec:Sec3_3:Step3}
 The full set of evolution equations is given by \eqref{eqn:frame_form1}, \eqref{eqn:Hct_evo}, \eqref{eqn:alphat_constraintmodevo}, \eqref{eqn:Udot_constraintmodevo}, \eqref{eqn:A_PDE3}, \eqref{eqn:N_evo_2}, and \eqref{eqn:Sigma_evo_2}. This system can be written in matrix form as follows 
\begin{align}
\label{eqn:Einstein_MatrixForm}
    B^0 \del_{t}\tilde{V} &= \alpha B^{\Sigma}\del_{\Sigma}\tilde{V} + \alpha\frac{1}{t}\tilde{\mathcal{B}}\tilde{V} + \alpha\frac{1}{t}\tilde{\mathcal{C}}\tilde{V}+ \alpha\mathcal{F},
\end{align}
where 
\begin{align}
\label{eqn:V_vector}
\tilde{V} &= (\Hct, \alphat, \dot{U}_{P}, A_{P}, N_{PI}, \Sigma_{PI}, e_{P}^{\Sigma})^{\text{T}}, \\
\label{eqn:B0_Einstein}
B^0 &= \begin{pmatrix}
a & c & 0 & 0 & 0 & 0 & 0 \\
c & b & 0 & 0 & 0 & 0 & 0 \\
0 & 0 &d\delta^{P}_{A} & g\delta_{A}^{P} & 0 & 0 & 0 \\ 
0 & 0 & g\delta_{A}^{P} & \delta^{P}_{A} & 0 & 0 & 0\\
0 & 0 & 0 & 0 & \delta^{P}_{A}\delta^{I}_{B} & 0 & 0\\ 
0 & 0 & 0 & 0 & 0 & \delta^{P}_{A}\delta^{I}_{B} & 0 \\
0 & 0 & 0 & 0 & 0 & 0 & \delta^{P}_{A} \end{pmatrix}, \\
\label{eqn:BSigma_Einstein}
B^{\Sigma} &= \begin{pmatrix} 0 & 0 & B_{13}\delta^{KP} & B_{14}\delta^{KP} & 0 & 0 & 0 \\
0 & 0 & B_{23}\delta^{KP} & B_{24}\delta^{KP}  & 0 & 0 & 0 \\
B_{31}\delta^{K}_{A} &  B_{32}\delta^{K}_{A}& 0 & 0 & 0 & B_{36} \delta^{KC}\pi^{PI}_{AC} & 0 \\ 
B_{41}\delta^{K}_{A} &  B_{42}\delta^{K}_{A}& 0 & 0 & 0 & B_{46}\delta^{KC}\pi^{PI}_{AC} & 0\\
0 & 0 & 0 & 0 & 0 & \tensor{\mathcal{S}}{^K_{AB}^{PI}} & 0\\ 
0 & 0 & \pi^{KP}_{AB} & -\pi^{KP}_{AB} & \tensor{\mathcal{S}}{^{KPI}_{AB}} & 0 & 0\\
0 & 0 & 0 & 0 & 0 & 0 & 0\end{pmatrix}e_{K}^{\Sigma}, \\
\label{tildeBcal-cmpts}
\tilde{\mathcal{B}} &= \text{diag}\bigl(0, 0, 0, 0, 0, \frac{2}{\alpha}\delta^{P}_{A}\delta^{I}_{B}, 0\bigr), \\
\label{tildeCcal-cmpts}
\tilde{\mathcal{C}} &= \begin{pmatrix} 0 & \tilde{\mathcal{C}}_{12} & 0 & 0 & 0 & 0 & 0\\
0 & \tilde{\mathcal{C}}_{22} & 0 & 0 & 0 & 0 & 0\\
0 & 0 & \tilde{\mathcal{C}}_{33}\delta^{P}_{A} & 0 & 0 & 0 & 0\\ 
0 & 0 & \tilde{\mathcal{C}}_{43}\delta^{P}_{A} & 0 & 0 & 0 & 0\\
0 & 0 & 0 & 0 & 0 & 0 & 0\\ 
0 & 0 & 0 & 0 & 0 & 0 & 0\\
0 & 0 & 0 & 0 & 0 & 0 & 0 \end{pmatrix},
\intertext{and}
\label{eqn:Fcal_SourceTerm}
\mathcal{F} &=\bigl( \mathcal{F}_{1} , \mathcal{F}_{2} , \mathcal{F}_{3} , \mathcal{F}_{4} , \mathcal{F}_{5} , \mathcal{F}_{6} , \mathcal{F}_{7})^T.
\end{align}
The coefficients of $B^{\Sigma}$ and $\tilde{\mathcal{C}}$ are given by
\begin{equation} \label{eqn:BSigma_Components}
\begin{aligned}
    B_{13} &= \frac{a+2c}{3}, \quad B_{14} = 4c\theta,\quad B_{23} = \frac{2b+c}{3}, \;\; B_{24} = 4b\theta, \\
    B_{31} &= d\big(3(3-2\lambda)-6\beta\big) + g\big(3(-1-2\gamma)-6\kappa\big), \\
    B_{32} &= d\big(-(3-2\lambda)+3\beta\big)+g\big(-(-1-2\gamma)+3\kappa\big), \\
    B_{36} &= d\lambda + g(\frac{1}{2}+\gamma) ,  \quad B_{46} = (\frac{1}{2}+\gamma) +g\lambda,\\
    B_{41} &= \big(3(-1-2\gamma)-6\kappa\big) +g\big(3(3-2\lambda)-6\beta\big), \\
    B_{42} &= \big(-(-1-2\gamma)+3\kappa\big) + g\big(-(3-2\lambda)+3\beta),
\end{aligned}
\end{equation}
and
\begin{align}
    \tilde{\mathcal{C}}_{12} &= c\big(\frac{\alpha\Lambda}{3}-\frac{6\theta}{\alpha}),\quad \tilde{\mathcal{C}}_{22} = b\big(\frac{\alpha\Lambda}{3}-\frac{6\theta}{\alpha}), \nonumber \\
    \label{tildeCcal-def}
    \tilde{\mathcal{C}}_{33} &= d\Big(\frac{\beta-2\lambda}{\alpha}+\frac{(4+\beta)\alpha\Lambda}{3}\Big) + g\Big(\frac{\kappa-2\gamma}{\alpha}+\frac{\kappa\alpha\Lambda}{3}\Big) ,  \\
    \tilde{\mathcal{C}}_{43} &=\Big(\frac{\kappa-2\gamma}{\alpha}+\frac{\kappa\alpha\Lambda}{3}\Big) +g\Big(\frac{\beta-2\lambda}{\alpha}+\frac{(4+\beta)\alpha\Lambda}{3}\Big), \nonumber
\end{align}
respectively, while
the components of $\mathcal{F}$ are given by 
\begin{align}
    \mathcal{F}_{1} &= c\Big[\frac{2}{3}\dot{U}_{A}\Big(\dot{U}_{A}-2A^{A}\Big)-(\frac{2}{3}+\theta)\Sigma_{AB}\Sigma^{AB} -\frac{\alphat^{2}}{2}-\frac{1}{3t^{2}}\Big(T_{00}+\tensor{T}{_A^A}\Big) -\frac{3}{2}\alphah\alphat   \nonumber\\
    &+6\theta(-\alphat+3\Hct)^{2} - 6\theta A^{A}A_{A} -\theta N^{AB}N_{AB} +\frac{\theta}{2}N^{2} -\frac{2\theta}{t^{2}}T_{00}\Big]  \nonumber\\
    &+ a\Big[\frac{1}{3}\dot{U}_{A}(\dot{U}^{A}-2A^{A})-\frac{1}{3}\Sigma_{AB}\Sigma^{AB}-\Hct^{2} -\frac{1}{6t^{2}}(T_{00}+\tensor{T}{_A^A})  -\frac{5}{3}\alphah\Hct +\frac{2}{9}\alphah^{2}\Big], \nonumber\\
    \mathcal{F}_{2} &= b\Big[\frac{2}{3}\dot{U}_{A}\Big(\dot{U}_{A}-2A^{A}\Big)-(\frac{2}{3}+\theta)\Sigma_{AB}\Sigma^{AB} -\frac{\alphat^{2}}{2}-\frac{1}{3t^{2}}\Big(T_{00}+\tensor{T}{_A^A}\Big) -\frac{3}{2}\alphah\alphat   \nonumber\\
    &+6\theta(-\alphat+3\Hct)^{2} - 6\theta A^{A}A_{A} -\theta N^{AB}N_{AB} +\frac{\theta}{2}N^{2} -\frac{2\theta}{t^{2}}T_{00}\Big]  \nonumber\\
    &+ c\Big[\frac{1}{3}\dot{U}_{A}(\dot{U}^{A}-2A^{A})-\frac{1}{3}\Sigma_{AB}\Sigma^{AB}-\Hct^{2} -\frac{1}{6t^{2}}(T_{00}+\tensor{T}{_A^A})-\frac{5}{3}\alphah\Hct +\frac{2}{9}\alphah^{2}\Big], \nonumber \\
\label{eqn:Fcal_Components}
    \mathcal{F}_{3} &= d\Big[2(-\alphat+3\Hct)\dot{U}_{A} - \tensor{\Sigma}{^B_A}\dot{U}_{B}  +\lambda\big(-3\tensor{\Sigma}{_A^B}A_{B} - \epsilon_{ABC}N^{BD}\tensor{\Sigma}{_D^C} - \frac{1}{t^{2}}T_{0A} \big)\Big] \\
    &+ g\Big[(\frac{1}{2}+\gamma)\delta^{BC}\pi^{EF}_{AC}e_{B}(\tensor{\Sigma}{_{EF}}) -(-\alphat+3\Hct)(\dot{U}_{A}+A_{A}) + \tensor{\Sigma}{_A^B}(\frac{1}{2}\dot{U}_{B}-A_{B}) \nonumber \\
    &-3\gamma\tensor{\Sigma}{_A^B}A_{B} -\gamma\epsilon_{ABC}N^{BD}\tensor{\Sigma}{_D^C}-\gamma \frac{1}{t^{2}}T_{0A}\Big], \nonumber \\
    \mathcal{F}_{4} &= g\Big[2(-\alphat+3\Hct)\dot{U}_{A} - \tensor{\Sigma}{^B_A}\dot{U}_{B}  +\lambda\big(-3\tensor{\Sigma}{_A^B}A_{B} - \epsilon_{ABC}N^{BD}\tensor{\Sigma}{_D^C} - \frac{1}{t^{2}}T_{0A} \big)\Big] \nonumber \\
    &+ \Big[(\frac{1}{2}+\gamma)\delta^{BC}\pi^{EF}_{AC}e_{B}(\tensor{\Sigma}{_{EF}}) -(-\alphat+3\Hct)(\dot{U}_{A}+A_{A}) + \tensor{\Sigma}{_A^B}(\frac{1}{2}\dot{U}_{B}-A_{B})  \nonumber \\
    &-3\gamma\tensor{\Sigma}{_A^B}A_{B} -\gamma\epsilon_{ABC}N^{BD}\tensor{\Sigma}{_D^C}-\gamma \frac{1}{t^{2}}T_{0A}\Big], \nonumber \\
    \mathcal{F}_{5} &= \tensor{\mathcal{S}}{^C_{AB}^{DE}}\dot{U}_{C}\Sigma_{DE} -\mathcal{H}N_{AB}+2N_{C(A}\tensor{\Sigma}{_{B)}^C} ,\nonumber  \\
    \mathcal{F}_{6} &= -3\mathcal{H}\Sigma_{AB} +\epsilon_{CD(A}\tensor{N}{_{B)}^D}\dot{U}^{C}-2\tensor{\epsilon}{_{CD(A}}\tensor{N}{_{B)}^D}A^{C} +\pi^{DE}_{AB}\Big(\dot{U}_{D}\dot{U}_{E}+A_{D}\dot{U}_{E} \nonumber \\
    &-2\tensor{N}{_D^C}N_{EC}+NN_{DE}+e^{2\Phi}T_{DE}+U_{DE}\Big), \nonumber\\
    \mathcal{F}_{7} &= -(\mathcal{H}\delta_{A}^{B} + \Sigma_{A}^{B})e_{B}^{\Sigma}. \nonumber
\end{align}
From \eqref{eqn:B0_Einstein}, \eqref{eqn:BSigma_Einstein}, and \eqref{eqn:BSigma_Components}, it is clear that for the system \eqref{eqn:Einstein_MatrixForm} to be symmetric hyperbolic, the following conditions must be satisfied:
\begin{gather}
    \frac{a+2c}{3} = d\big(3(3-2\lambda)-6\beta\big) + g\big(3(-1-2\gamma)-6\kappa\big), \nonumber \\
    \frac{2b+c}{3} =d\big(-(3-2\lambda)+3\beta\big)+g\big(-(-1-2\gamma)+3\kappa\big), \nonumber \\
    \label{eqn:symmetrisation_conditions_einstein}
    4c\theta = \big(3(-1-2\gamma)-6\kappa\big) +g\big(3(3-2\lambda)-6\beta\big), \\
    4b\theta = \big(-(-1-2\gamma)+3\kappa\big) + g\big(-(3-2\lambda)+3\beta), \nonumber \\
    1 = d\lambda + g(\frac{1}{2}+\gamma),\quad  -1 = (\frac{1}{2}+\gamma) +g\lambda. \nonumber 
\end{gather}
We further observe, using \eqref{eqn:Hct_identities} and \eqref{eqn:lapse_Hct_identity}, that the ostensibly singular terms $\frac{1}{t}\tilde{\mathcal{C}}$, see \eqref{tildeCcal-cmpts} and\eqref{tildeCcal-def}, will, in fact, be regular in $t$ and given by
\begin{equation}
\begin{gathered}
\label{eqn:Analysis_Ccal_tilde}
    \frac{1}{t}\tilde{\mathcal{C}}_{12} = -c\alphah, \quad 
    \frac{1}{t}\tilde{\mathcal{C}}_{22} = -b\alphah, \quad 
    \frac{1}{t}\tilde{\mathcal{C}}_{33} = \bigl(-d(\lambda+2) -g\gamma\bigr)\alphah, \\ 
    \frac{1}{t}\tilde{\mathcal{C}}_{43} = \bigl(-\gamma - g(\lambda+2)\bigr)\alphah, 
\end{gathered}
\end{equation}
for the parameter choices $\theta = \frac{1}{6}$, $\beta = \lambda-2$, and $\kappa=\gamma$. Consequently, if we can solve \eqref{eqn:symmetrisation_conditions_einstein} with $\theta = \frac{1}{6}$, $\beta = \lambda-2$, and $\kappa=\gamma$, we will obtain a symmetric system and the singular term $\frac{\alpha}{t}\tilde{\mathcal{C}}$ will be absorbed into the source term $\mathcal{F}$. It is simple to verify that
\begin{align}
    a &= -28+25\sqrt{2}, \quad b= \frac{3}{4}(-7+6\sqrt{2}), \quad c= 12-\frac{21}{\sqrt{2}} , \quad d= \frac{1}{9}(5+\sqrt{2}),\quad  \nonumber\\
    \label{eqn:Symmetrisation_parameter_values}
    g &= \frac{1}{3}(-1-\sqrt{2}) ,\quad \gamma = -\frac{1}{2}+\sqrt{2} , \quad \theta = \frac{1}{6}, \quad \kappa = -\frac{1}{2}+\sqrt{2}, \quad \lambda = 3, \quad \beta = 1,
\end{align}
is such a solution.
Furthermore, the matrix $B^{0}$ is positive definite for these values as can be seen from its eigenvalues
\begin{equation*}
\begin{aligned}
    \lambda_{1} &\approx 8.46039, \quad \lambda_{2} \approx 1.6738, \quad \lambda_{3} = 1, \quad \lambda_{4} = 1, \\
    \lambda_{5} &= 1, \quad \lambda_{6} \approx 0.0388859, \quad \lambda_{7} \approx 0.00890961,
\end{aligned}
\end{equation*}
which are all positive. Thus, for these parameter choices, \eqref{eqn:Einstein_MatrixForm} provides us with a symmetric hyperbolic Fuchsian formulation of the Einstein equations. Additionally, we observe that the singular matrix $\tilde{\mathcal{B}}$ can be expressed as 
\begin{align*}
\tilde{\mathcal{B}} = \mathcal{B}\mathcal{P}
\end{align*}
where 
\begin{align}
\label{eqn:Bcal_Einstein}
    \mathcal{B} &= \text{diag}(1,1,\delta^{P}_{A},\delta^{P}_{A},\delta^{P}_{A}\delta^{I}_{B},\frac{2}{\alpha}\delta^{P}_{A}\delta^{I}_{B}, \delta^{P}_{A}), \\
\label{eqn:P_projector_def}
    \mathcal{P} &= \text{diag}(0,0,0,0,0,\delta_{A}^P\delta_{B}^I,0),
\end{align}
and the projection operator $\mathcal{P}$ satisfies
\begin{align}
\label{eqn:Pcal_perp_defn}
\mathcal{P}^{\perp} = \mathbbm{1}-\mathcal{P}, \;\; \mathcal{P}^{2} = \mathcal{P}, \;\; (\mathcal{P}^{\perp})^{2} = \mathcal{P}^{\perp}, \;\; \mathcal{P}\mathcal{P}^{\perp} = \mathcal{P}^{\perp}\mathcal{P} = 0, \;\; \mathcal{P}+\mathcal{P}^{\perp} = \mathbbm{1}.
\end{align}
Finally, absorbing the now regular terms \eqref{eqn:Analysis_Ccal_tilde} into the source term $\mathcal{F}$ changes its components as follows:
\begin{align}
    \mathcal{F}_{1} &= c\Big[\frac{2}{3}\dot{U}_{A}\Big(\dot{U}_{A}-2A^{A}\Big)-(\frac{2}{3}+\theta)\Sigma_{AB}\Sigma^{AB} -\frac{\alphat^{2}}{2}-\frac{1}{3t^{2}}\Big(T_{00}+\tensor{T}{_A^A}\Big) -\frac{3}{2}\alphah\alphat   \nonumber \\
    &+6\theta(-\alphat+3\Hct)^{2} - 6\theta A^{A}A_{A} -\theta N^{AB}N_{AB} +\frac{\theta}{2}N^{2} -\frac{2\theta}{t^{2}}T_{00}\Big]  \nonumber \\
    &+ a\Big[\frac{1}{3}\dot{U}_{A}(\dot{U}^{A}-2A^{A})-\frac{1}{3}\Sigma_{AB}\Sigma^{AB}-\Hct^{2} -\frac{1}{6t^{2}}(T_{00}+\tensor{T}{_A^A})  \nonumber \\
&-\frac{5}{3}\alphah\Hct +\frac{2}{9}\alphah^{2}\Big] -c\alphah\alphat\alpha, \nonumber \\
    \mathcal{F}_{2} &= b\Big[\frac{2}{3}\dot{U}_{A}\Big(\dot{U}_{A}-2A^{A}\Big)-(\frac{2}{3}+\theta)\Sigma_{AB}\Sigma^{AB} -\frac{\alphat^{2}}{2}-\frac{1}{3t^{2}}\Big(T_{00}+\tensor{T}{_A^A}\Big) -\frac{3}{2}\alphah\alphat   \nonumber \\
    &+6\theta(-\alphat+3\Hct)^{2} - 6\theta A^{A}A_{A} -\theta N^{AB}N_{AB} +\frac{\theta}{2}N^{2} -\frac{2\theta}{t^{2}}T_{00}\Big]  \nonumber\\
    &+ c\Big[\frac{1}{3}\dot{U}_{A}(\dot{U}^{A}-2A^{A})-\frac{1}{3}\Sigma_{AB}\Sigma^{AB}-\Hct^{2} -\frac{1}{6t^{2}}(T_{00}+\tensor{T}{_A^A}) \nonumber \\
&-\frac{5}{3}\alphah\Hct +\frac{2}{9}\alphah^{2}\Big] -b\alphah\alphat \alpha, \nonumber \\
\label{eqn:Fcal_components_Final}
    \mathcal{F}_{3} &= d\Big[2(-\alphat+3\Hct)\dot{U}_{A} - \tensor{\Sigma}{^B_A}\dot{U}_{B}  +\lambda\big(-3\tensor{\Sigma}{_A^B}A_{B} - \epsilon_{ABC}N^{BD}\tensor{\Sigma}{_D^C} - \frac{1}{t^{2}}T_{0A} \big)\Big] \\
    &+ g\Big[(\frac{1}{2}+\gamma)\delta^{BC}\pi^{EF}_{AC}e_{B}(\tensor{\Sigma}{_{EF}}) -(-\alphat+3\Hct)(\dot{U}_{A}+A_{A}) + \tensor{\Sigma}{_A^B}(\frac{1}{2}\dot{U}_{B}-A_{B})  \nonumber \\
    &-3\gamma\tensor{\Sigma}{_A^B}A_{B} -\gamma\epsilon_{ABC}N^{BD}\tensor{\Sigma}{_D^C}-\gamma \frac{1}{t^{2}}T_{0A}\Big] +\big(-d(\lambda+2)-g\gamma\big)\alpha\alphah\dot{U}_{A}, \nonumber \\
    \mathcal{F}_{4} &= g\Big[2(-\alphat+3\Hct)\dot{U}_{A} - \tensor{\Sigma}{^B_A}\dot{U}_{B}  +\lambda\big(-3\tensor{\Sigma}{_A^B}A_{B} - \epsilon_{ABC}N^{BD}\tensor{\Sigma}{_D^C} - \frac{1}{t^{2}}T_{0A} \big)\Big] \nonumber \\
    &+ \Big[(\frac{1}{2}+\gamma)\delta^{BC}\pi^{EF}_{AC}e_{B}(\tensor{\Sigma}{_{EF}}) -(-\alphat+3\Hct)(\dot{U}_{A}+A_{A}) + \tensor{\Sigma}{_A^B}(\frac{1}{2}\dot{U}_{B}-A_{B})  \nonumber \\
    &-3\gamma\tensor{\Sigma}{_A^B}A_{B} -\gamma\epsilon_{ABC}N^{BD}\tensor{\Sigma}{_D^C}-\gamma \frac{1}{t^{2}}T_{0A}\Big] +\big(-\gamma -g(\lambda+2)\big)\alpha\alphah\dot{U}_{A},\nonumber \\
    \mathcal{F}_{5} &= \tensor{\mathcal{S}}{^C_{AB}^{DE}}\dot{U}_{C}\Sigma_{DE} -\mathcal{H}N_{AB}+2N_{C(A}\tensor{\Sigma}{_{B)}^C} , \nonumber \\
    \mathcal{F}_{6} &= -3\mathcal{H}\Sigma_{AB} +\epsilon_{CD(A}\tensor{N}{_{B)}^D}\dot{U}^{C}-2\tensor{\epsilon}{_{CD(A}}\tensor{N}{_{B)}^D}A^{C} +\pi^{DE}_{AB}\Big(\dot{U}_{D}\dot{U}_{E}+A_{D}\dot{U}_{E} \nonumber \\
    &-2\tensor{N}{_D^C}N_{EC}+NN_{DE}+\frac{1}{t^{2}}T_{DE}+U_{DE}\Big), \nonumber \\
    \mathcal{F}_{7} &= -(\mathcal{H}\delta_{A}^{B} + \Sigma_{A}^{B})e_{B}^{\Sigma}. \nonumber
\end{align}
With these changes, we then drop the term $\frac{\alpha}{t}\tilde{\mathcal{C}}\tilde{V}$ from \eqref{eqn:Einstein_MatrixForm}.

\subsection{Subtracting the Homogeneous Gravitational Background} 
\label{sec:Subtract_BG_Einstein_Sec3}
We recall that our approximate background metric \eqref{eqn:Einstein_Euler_BG_Solution1} is given by the conformal de Sitter metric 
\begin{align*}
    g_{\text{bg}} = -\frac{3}{\Lambda}dt^{2} +\delta_{\Omega\Gamma}dx^{\Omega}dx^{\Gamma}.
\end{align*}
In particular, for this conformal metric and the orthonormal spatial frame
$e_P^{\text{bg}}=\delta^{\Sigma}_{P}\del_{\Sigma}$,
the associated lapse and connection coefficient variables all vanish, that is,
\begin{align*}
    \Hct_{\text{bg}} = \alphat_{\text{bg}}  = \dot{U}^{\text{bg}}_{P} = A^{\text{bg}}_{P} = N^{\text{bg}}_{PI} = \Sigma^{\text{bg}}_{PI} = 0.
\end{align*}
Next, for perturbed solutions, we define 
\begin{align}
\label{eqn:f_modified_frame_defn}
    f^{\Sigma}_{P} = e^{\Sigma}_{P}-\delta^{\Sigma}_{P},
\end{align}
which measures the spatial frame deviation from the background solution.  It is straightforward to verify that $f^{\Sigma}_{P}$ satisfies the evolution equation
\begin{align*}
    e_{0}(f_{A}^{\Sigma}) &= -(\mathcal{H}\delta_{A}^{B} + \Sigma_{A}^{B})(f_{B}^{\Sigma}+\delta^{\Sigma}_{B}).
\end{align*}
Setting
\begin{align*}
    V &= (\Hct, \alphat, \dot{U}_{P}, A_{P}, N_{PI}, \Sigma_{PI}, f_{P}^{\Sigma})^{\text{T}}, \\
    \mathcal{F}_{7} &= -(\mathcal{H}\delta_{A}^{B} + \Sigma_{A}^{B})(f_{B}^{\Sigma}+\delta^{\Sigma}_{B}),
\end{align*}
we can express \eqref{eqn:Einstein_MatrixForm} as
\begin{align}
\label{eqn:SH_Einstein_Sec3}
    B^0 \del_{t}V &= \alpha B^{\Omega}\del_{\Omega}V + \frac{1}{t}\alpha\mathcal{B}\mathcal{P}V + \alpha\mathcal{F}.
\end{align}

\begin{rem}
\label{rem:errorterms} 
As discussed in the introduction, the de Sitter approximate background, defined by $V = 0$, does not satisfy the gravitational equations \eqref{eqn:SH_Einstein_Sec3} due to the presence of source terms involving the stress-energy tensor. However, as we show, these stress-energy contributions vanish asymptotically as $t \searrow 0$ for sufficiently small nonlinear perturbations, which is consistent with the expectation that cosmological solutions with a positive cosmological constant isotropise at late times.
\end{rem}

\section{A Symmetric Hyperbolic Formulation of the Euler Equations}
\label{sec:Sym_Hyp_Euler}

\subsection{Euler Equations}
\label{sec:Conformal_Euler_deriv}
We now turn to deriving a symmetric hyperbolic formulation of the relativistic Euler equations \eqref{eqn:Euler_Physical}. For ease of presentation, we will drop the Fraktur superscripts that label the two fluids for this derivation. Ultimately, the final Einstein-Euler system will contain two copies of the fluid equations derived here. \newline \par

To begin, we recall that the conservation of fluid stress-energy is given by
\begin{align*}
    \tilde{\nabla}_{a}\tilde{T}^{ab} = 0
\end{align*}
where 
\begin{align*}
    \tilde{T}^{ab} = \tilde{g}^{ac}\tilde{g}^{bd}\tilde{T}_{cd} = e^{-2\Phi}g^{ac}g^{bd}T_{cd} = e^{-2\Phi}T^{ab}.
\end{align*}
In terms of the covariant derivative $\nabla$ associated with the conformal metric $g$, this becomes
\begin{align*}
\nabla_{a}\tilde{T}^{ab} = -6\tilde{T}^{ab}\nabla_{a}\Phi + \eta_{ac}\eta^{bd}\tilde{T}^{ac}\nabla_{d}\Phi,
\end{align*}
where the conformal factor $\Phi$ is determined by \eqref{Phi-def}. Expressing the matter equations in terms of the conformal stress-energy tensor $T^{ab}$ (cf. \eqref{eqn:Conformal_Tmunu_def}) then gives
\begin{align}
\label{eqn:conformal_euler_frame1}
e_{a}(T^{ab}) + \tensor{\omega}{_a^a_c}T^{cb} + \tensor{\omega}{_a^b_c}T^{ac} +4T^{ab}r_{a} - \eta_{ac}T^{ac}r^{b} = 0
\end{align}
where $r_a$ is defined above by \eqref{ra-def}.
Lowering the indices of $T^{ab}$ and splitting up time and spatial derivatives in \eqref{eqn:conformal_euler_frame1} leads to
\begin{align}
\label{eqn:conformal_euler_frame2}
&\eta^{0c}\eta^{bd}e_{0}(T_{cd}) + \eta^{Ac}\eta^{bd}e_{A}(T_{cd}) +\eta^{ac}\eta^{bd}\tensor{\omega}{_a_c_e}\tensor{T}{^e_d} \nonumber \\
&+ \eta^{bd}\tensor{\omega}{_a_d_e}T^{ae} +4\eta^{ac}\eta^{bd}T_{cd}r_{a} - \eta^{bd}\eta_{ac}T^{ac}r_{d} = 0.
\end{align}

Setting $b=0$ in \eqref{eqn:conformal_euler_frame2} and replacing the connection coefficients using \eqref{eqn:connection_identities} yields
\begin{align*}
(e_{0}+3\mathcal{H})(T_{00}) - (e_{F} +2\dot{U}_{F} - 2A_{F})(\tensor{T}{^F_0}) +\mathcal{H}\tensor{T}{^F_F}+\Sigma_{IF}\tensor{T}{^{IF}} + C_{0} = 0,
\end{align*}
where $C_{0} = (3T_{00}+\tensor{T}{_B^B})r_{0} - 4\tensor{T}{^B_0}r_{B}$. Similarly, for $b=J$, we obtain
\begin{align*}
&-\eta^{JL}e_{0}(T_{0L})  + \eta^{JL}e_{I}(\tensor{T}{^I_L}) +4\mathcal{H}T^{J0} - 3A_{K}T^{KJ} +\dot{U}_{F}T^{FJ} + T^{00}\dot{U}^{J} + A^{J}\tensor{T}{^F_F} +T^{I0}\tensor{\Sigma}{_I^J} \nonumber \\
&+\eta^{JL}T^{IF}\tensor{N}{^E_I}\epsilon_{LFE} +\bigl(-4\eta^{JL}T_{0L}r_{0} + 4\eta^{JL}\tensor{T}{^I_L}r_{I} + \eta^{JL}T_{00}r_{L}-\eta^{JL}\tensor{T}{^I_I}r_{L}\big) = 0.
\end{align*}
Multiplying the above by $-\delta_{JP}$ then gives
\begin{align*}
&(e_{0}+4\mathcal{H})(T_{0P})  - \big(e_{I} +\dot{U}_{I} -3A_{I}\big)(\tensor{T}{^I_P})  - T_{00}\dot{U}_{P} - A_{P}\tensor{T}{^F_F} +T_{I0}\tensor{\Sigma}{^I_P} \nonumber \\
&+T_{IF}\tensor{N}{_E^I}\tensor{\epsilon}{_P^{EF}} +C_{P} = 0,
\end{align*}
where $C_{P} = 4T_{0P}r_{0} -(T_{00}-\tensor{T}{^I_I})r_{P} - 4\tensor{T}{^I_P}r_{I}$. Collecting the fluid equations together, we have
\begin{align}
\label{eqn:T_00_evo}
(e_{0}+3\mathcal{H})(T_{00}) - (e_{F} +2\dot{U}_{F} - 2A_{F})(\tensor{T}{^F_0}) +\mathcal{H}\tensor{T}{^F_F}+\Sigma_{IF}\tensor{T}{^{IF}} + C_{0} &= 0, \\
\label{eqn:T_0P_evo}
(e_{0}+4\mathcal{H})(T_{0P})  - \big(e_{I} +\dot{U}_{I} -3A_{I}\big)(\tensor{T}{^I_P})  - T_{00}\dot{U}_{P} - A_{P}\tensor{T}{^F_F} &\nonumber\\+T_{I0}\tensor{\Sigma}{^I_P} 
+T_{IF}\tensor{N}{_E^I}\tensor{\epsilon}{_P^{EF}} +C_{P} &= 0.
\end{align}

\subsection{3+1 Decomposition of the Euler Equations}
\label{sec:Euler_3+1_decomp}
The 3+1 decomposition of the fluid stress-energy tensor is given by
\begin{equation}
\begin{aligned}
\label{eqn:Tab_3+1_lorentz}
T_{00} &= (\Gamma^{2}(K+1)-K)\rho, \\
T_{0A} &= \Gamma^{2}(K+1)\rho\nu_{A} = (T_{00}+K\rho)\nu_{A} , \\
T_{AB} &= \Gamma^{2}(K+1)\rho\nu_{A}\nu_{B} + K\rho\eta_{AB} = (T_{00}+K\rho)\nu_{A}\nu_{B} + K\rho\eta_{AB}, 
\end{aligned}
\end{equation}
where the spatial fluid velocity $\nu_{A}$ and the Lorentz factor $\Gamma$ are defined by 
\begin{align}
\label{eqn:velocity_split}
v_{a} &= \Gamma(n_{a}+\nu_{a}), \\
\label{eqn:nu_magnitude}
|\nu|^{2} &= \delta^{AB}\nu_{A}\nu_{B}, \\
\label{eqn:Gamma_def}
\Gamma &= \frac{1}{(1-|\nu|^{2})^{\frac{1}{2}}},
\end{align}
where $n=e_{0}$ is the normal vector of our foliation and $n^{a}\nu_{a} = \nu_0=0$. Substituting the above definitions into \eqref{eqn:T_00_evo} and \eqref{eqn:T_0P_evo} yields
\begin{align}
\label{eqn:T_00_evo_expanded}
&(\Gamma^{2}(K+1)-K)e_{0}(\rho) + (K+1)\rho e_{0}(\Gamma^{2}) -e_{F}(\Gamma^{2}(K+1)\rho\nu^{F}) \nonumber \\ 
&- (2\dot{U}_{F} - 2A_{F})\tensor{T}{^F_0} +3\mathcal{H}T_{00} +\mathcal{H}\tensor{T}{^F_F}+\Sigma_{IF}\tensor{T}{^{IF}} + C_{0} = 0
\end{align}
and
\begin{align}
\label{eqn:T_0P_evo_expanded}
&e_{0}\big((T_{00}+K\rho)\nu_{P}\big) - e_{I}\big((T_{00}+K\rho)\nu^{I}\nu_{P}+K\rho\delta^{I}_{P}\big) - \big(\dot{U}_{I} -3A_{I}\big)\tensor{T}{^I_P}\nonumber \\ 
&+4\mathcal{H}T_{0P} - T_{00}\dot{U}_{P} - A_{P}\tensor{T}{^F_F} +T_{I0}\tensor{\Sigma}{^I_P} +T_{IF}\tensor{N}{_E^I}\tensor{\epsilon}{_P^{EF}} +C_{P} = 0,
\end{align}
respectively. Replacing $T_{00}$ in \eqref{eqn:T_0P_evo_expanded} using \eqref{eqn:T_00_evo} gives
\begin{align}
\label{eqn:T_0P_evo_expanded2}
& K\nu_{P}e_{0}(\rho) + (K\rho+T_{00})e_{0}(\nu_{P}) - Ke_{I}(\rho\delta^{I}_{P}) - (T_{00}+K\rho)\nu^{F}e_{F}(\nu_{P})  \nonumber \\
&+ 4\mathcal{H}T_{0P} +(3A_{I}-\dot{U}_{I})\tensor{T}{^I_P} -T_{00}\dot{U}_{P} -A_{P}\tensor{T}{^F_F} +T_{I0}\tensor{\Sigma}{^I_P} + T_{IF}\tensor{N}{_E^I}\tensor{\epsilon}{_P^{EF}} +C_{P}\nonumber \\
& -3\nu_{P}\mathcal{H}T_{00} +2\nu_{P}(\dot{U}_{F}-A_{F})\tensor{T}{^F_0} -\nu_{P}\mathcal{H}\tensor{T}{^F_F} -\nu_{P}\Sigma_{IF}T^{IF} -\nu_{P}C_{0} = 0.
\end{align}
Similarly, using \eqref{eqn:Gamma_def}, we can express \eqref{eqn:T_00_evo_expanded} as
\begin{align}
\label{eqn:T_00_evo_expanded2}
&\Big(\Gamma^{2}(K+1)-K\Big)e_{0}(\rho) + \frac{2(K+1)\rho }{(1-|\nu|^{2})^{2}}\nu^{C}e_{0}(\nu_{C}) -\Gamma^{2}(K+1)\nu^{F}e_{F}(\rho) \nonumber \\
&- (K+1)\rho\Big(\Gamma^{2}\eta^{FC}+ \frac{2}{(1-|\nu|^{2})^{2}}\nu^{C}\nu^{F}\Big)e_{F}(\nu_{C}) \nonumber \\
& -2(\dot{U}_{F}- A_{F})\tensor{T}{^F_0} + 3\mathcal{H}T_{00} +\mathcal{H}\tensor{T}{^F_F}+\Sigma_{IF}\tensor{T}{^{IF}} + C_{0} = 0.
\end{align}
Motivated by the results of \cite{Oliynyk:CMP_2016}, we introduce the modified density $\hat{\zeta}$ by 
\begin{align}
\label{eqn:zeta_hat_defn}
    \hat{\zeta} = \frac{1}{K+1}\ln\left(\frac{\rho}{\rho_{0}}\right) +3\Phi.
\end{align}
 Expanding the stress-energy terms in \eqref{eqn:T_0P_evo_expanded2}-\eqref{eqn:T_00_evo_expanded2} using \eqref{eqn:Tab_3+1_lorentz}, replacing $\Gamma$ using \eqref{eqn:Gamma_def}, and expressing the density in terms of $\hat{\zeta}$, we then obtain the following form of the Euler equations
 \begin{align}
\label{eqn:Vector_Euler2}
& K\nu_{P}e_{0}(\hat{\zeta}) + \frac{1}{1-|\nu|^{2}}e_{0}(\nu_{P}) - Ke_{P}(\hat{\zeta}) - \frac{1}{1-|\nu|^{2}}\nu^{F}e_{F}(\nu_{P})  +(1-3K)\nu_{P}r_{0} \nonumber \\
&+  \Bigg(\mathcal{H}\nu_{P} +\frac{1}{1-|\nu|^{2}}(A_{I}+\dot{U}_{I})\nu^{I}\nu_{P} -\frac{1}{1-|\nu|^{2}}\dot{U}_{P} -\frac{|\nu|^{2}}{1-|\nu|^{2}}A_{P} +\frac{1}{1-|\nu|^{2}}\tensor{\Sigma}{^I_P}\nu_{I} \nonumber \\
&+ \frac{1}{1-|\nu|^{2}}\tensor{N}{_E^I}\tensor{\epsilon}{_P^{EF}}\nu_{I}\nu_{F} -\frac{1}{1-|\nu|^{2}}\Sigma_{IF}\nu^{I}\nu^{F}\nu_{P}\Bigg) = 0,
\end{align}
and 
\begin{align}
\label{eqn:Scalar_Euler2}
&\frac{1+K|\nu|^{2}}{1-|\nu|^{2}}e_{0}(\hat{\zeta}) + \frac{2}{(1-|\nu|^{2})^{2}} \nu^{C}e_{0}(\nu_{C}) -\frac{K+1}{1-|\nu|^{2}}\nu^{F}e_{F}(\hat{\zeta}) \nonumber \\
&- \Big(\frac{1}{1-|\nu|^{2}}\eta^{FC}+ \frac{2}{(1-|\nu|^{2})^{2}}\nu^{C}\nu^{F}\Big)e_{F}(\nu_{C})  + (1-3K )\frac{|\nu|^{2}}{1-|\nu|^{2}}r_{0} \nonumber \\
&+ \bigg[-2(\dot{U}_{F}- A_{F})\frac{1}{1-|\nu|^{2}}\nu^{F} + \frac{\mathcal{H}(3+|\nu|^{2})}{1-|\nu|^{2}}+\frac{1}{1-|\nu|^{2}}\Sigma_{IF}\nu^{I}\nu^{F}\bigg] = 0.
\end{align}

\subsection{Radial Decomposition of the Spatial Fluid Velocity}
\label{sec:Fluid_2+1_Decomp}
Motivated by \cite{Oliynyk:2021}, we further decompose the spatial fluid velocity $\nu_{C}$ as follows
\begin{align}
\label{eqn:nuh_defn}
\nuh_{C} = \frac{\nu_{C}}{|\nu|}, \quad |\nu|^{2} = \delta_{AB}\nu^{A}\nu^{B}.
\end{align}
This decomposition results in the constraint
\begin{align}
\label{eqn:C5_nu_normalisation_constraint}
   \mathcal{C}_{5} := \nuh^{A}\nuh_{A} - 1=0,
\end{align}
which we will assume holds for the derivations carried out in this section and the following one. 
\begin{rem}
    The radial decomposition \eqref{eqn:nuh_defn} generalises the decomposition of the fluid velocity used by the third author in \cite{Oliynyk:2021}. The essential reason for this split is that the norm and the unit vector decay at different rates and must be scaled differently to obtain an appropriate Fuchsian system.
\end{rem}

Letting the frame vector $e_{i}$ act on $\nu_{C}$ gives
\begin{align*}
e_{i}(\nu_{C}) = e_{i}(|\nu|\nuh_{C}) = e_{i}(|\nu|)\nuh_{C} + |\nu|e_{i}(\nuh_{C}).
\end{align*}
Introducing the operator
\begin{align}
\label{eqn:P_def}
P^{A}_{B} = \delta^{A}_{B} - \nuh^{A}\nuh_{B}
\end{align}
that projects into the subspace orthogonal to $\nuh_A$,
we find that 
\begin{align*}
e_{i}(\nuh_{C}) = \frac{1}{|\nu|}P^{B}_{C}e_{i}(\nu_{B}).
\end{align*}
From this, we obtain
\begin{align*}
\nuh^{C}e_{i}(\nuh_{C}) = \frac{1}{|\nu|}\nuh^{C}P^{B}_{C}e_{i}(\nu_{B}) = 0,
\end{align*}
since it is clear from \eqref{eqn:P_def} that
\begin{align}
\label{eqn:P_contract_nu}
\nuh^{B}P^{A}_{B} = P^{A}_{B}\nuh_{A} = 0.
\end{align}
In terms of $|\nu|$ and $\nuh_{A}$, the Euler equations  \eqref{eqn:Vector_Euler2}-\eqref{eqn:Scalar_Euler2} can be expressed as 
\begin{align}
\label{eqn:Vector_Euler_2+1}
    & K\nu_{P}e_{0}(\hat{\zeta}) + \frac{1}{1-|\nu|^{2}}\nuh_{P}e_{0}(|\nu|) + \frac{|\nu|}{1-|\nu|^{2}}e_{0}(\nuh_{P}) \nonumber \\
    &- Ke_{P}(\hat{\zeta}) - \frac{|\nu|^{2}}{1-|\nu|^{2}}\nuh^{F}e_{F}(\nuh_{P}) - \frac{|\nu|}{1-|\nu|^{2}}\nuh^{F}\nuh_{P}e_{F}(|\nu|) +(1-3K)\nu_{P}r_{0} \nonumber \\
&+  \Bigg[\mathcal{H}\nu_{P} +\frac{1}{1-|\nu|^{2}}(A_{I}+\dot{U}_{I})\nu^{I}\nu_{P} -\frac{1}{1-|\nu|^{2}}\dot{U}_{P} -\frac{|\nu|^{2}}{1-|\nu|^{2}}A_{P} +\frac{1}{1-|\nu|^{2}}\tensor{\Sigma}{^I_P}\nu_{I} \nonumber \\
&+ \frac{1}{1-|\nu|^{2}}\tensor{N}{_E^I}\tensor{\epsilon}{_P^{EF}}\nu_{I}\nu_{F} -\frac{1}{1-|\nu|^{2}}\Sigma_{IF}\nu^{I}\nu^{F}\nu_{P}\Bigg] = 0,
\end{align}
and
\begin{align}
\label{eqn:Scalar_Euler_2+1}
    &\frac{1+K|\nu|^{2}}{1-|\nu|^{2}}e_{0}(\hat{\zeta}) + \frac{2|\nu|}{(1-|\nu|^{2})^{2}} e_{0}(|\nu|) -\frac{K+1}{1-|\nu|^{2}}\nu^{F}e_{F}(\hat{\zeta}) \nonumber \\
    &- \Big(\frac{1}{1-|\nu|^{2}}+ \frac{2|\nu|^{2}}{(1-|\nu|^{2})^{2}}\Big)\nuh^{F}e_{F}(|\nu|) - \frac{1}{1-|\nu|^{2}}P^{FC}e_{F}(\nuh_{C}) + \Big(1-3K \Big)\frac{|\nu|^{2}}{1-|\nu|^{2}}r_{0}\nonumber \\
    &+ \Bigg[-2(\dot{U}_{F}- A_{F})\frac{1}{1-|\nu|^{2}}\nu^{F} + \frac{\mathcal{H}(3+|\nu|^{2})}{1-|\nu|^{2}}+\frac{1}{1-|\nu|^{2}}\Sigma_{IF}\nu^{I}\nu^{F}\Bigg] = 0.
\end{align}
Applying $P^{P}_{Q}$ and $\nuh^{P}$ to the vector equation \eqref{eqn:Vector_Euler_2+1}  yields\footnote{In deriving these equations we have used \eqref{eqn:P_contract_nu} and the fact $P^{P}_{Q}e_{i}(\nuh_{P}) = e_{i}(\nuh_{Q})$.} 
\begin{align}
\label{eqn:Euler_Vector_Pcontract}
    & \frac{|\nu|}{1-|\nu|^{2}}e_{0}(\nuh_{Q}) - KP^{P}_{Q}e_{P}(\hat{\zeta}) - \frac{|\nu|^{2}}{1-|\nu|^{2}}\nuh^{F}e_{F}(\nuh_{Q}) -\frac{1}{1-|\nu|^{2}}P^{P}_{Q}\dot{U}_{P} \nonumber \\
    &-\frac{|\nu|^{2}}{1-|\nu|^{2}}P^{P}_{Q}A_{P} +\frac{|\nu|}{1-|\nu|^{2}}P^{P}_{Q}\tensor{\Sigma}{^I_P}\nuh_{I} + \frac{|\nu|^{2}}{1-|\nu|^{2}}P^{P}_{Q}\tensor{N}{_E^I}\tensor{\epsilon}{_P^{EF}}\nuh_{I}\nuh_{F}  = 0.
\end{align}
and
\begin{align}
\label{eqn:Vector_Euler_nucontract}
    & K|\nu|e_{0}(\hat{\zeta}) + \frac{1}{1-|\nu|^{2}}e_{0}(|\nu|) - K\nuh^{P}e_{P}(\hat{\zeta}) - \frac{|\nu|}{1-|\nu|^{2}}\nuh^{F}e_{F}(|\nu|) +(1-3K)|\nu|r_{0} \nonumber \\
&+  \Bigg[\mathcal{H}|\nu| +\frac{|\nu|^{2}}{1-|\nu|^{2}}(A_{I}+\dot{U}_{I})\nuh^{I} -\frac{1}{1-|\nu|^{2}}\nuh^{P}\dot{U}_{P} -\frac{|\nu|^{2}}{1-|\nu|^{2}}\nuh^{P}A_{P} +\frac{|\nu|}{1-|\nu|^{2}}\tensor{\Sigma}{^I_P}\nuh_{I}\nuh^{P} \nonumber \\
&+ \frac{|\nu|^{2}}{1-|\nu|^{2}}\tensor{N}{_E^I}\tensor{\epsilon}{_P^{EF}}\nuh_{I}\nuh_{F}\nuh^{P} -\frac{|\nu|^{3}}{1-|\nu|^{2}}\Sigma_{IF}\nuh^{I}\nuh^{F}\Bigg] = 0,
\end{align}
respectively.  \newline \par

Next, we introduce the variables $u$ and $\zeta$ by
\begin{align}
\label{eqn:u_zeta_defns}
    u = \frac{|\nu|}{\sqrt{1-|\nu|^{2}}} \quad \text{and} \quad
    \zeta = \hat{\zeta} + \ln(\sqrt{1+u^{2}}),
\end{align}
and note that
\begin{align*}
    |\nu|  = \frac{u}{\sqrt{1+u^{2}}}, \quad e_{i}(\hat{\zeta}) = e_{i}(\zeta) -\frac{u}{1+u^{2}}e_{i}(u).
\end{align*}
\begin{rem}
    The variables \eqref{eqn:u_zeta_defns} generalise the variable change in equation (2.12) from \cite{Oliynyk:2021}. The purpose of this change is that it will ultimately block diagonalise the coefficient matrix in front the time derivatives for the Euler equations.
\end{rem}
When expressed in terms of $u$ and $\zeta$, the Euler equations \eqref{eqn:Scalar_Euler_2+1}, \eqref{eqn:Euler_Vector_Pcontract}, and \eqref{eqn:Vector_Euler_nucontract} become
\begin{align}
\label{eqn:Euler2+1_1}
     &(1+(1+K)u^{2})e_{0}(\zeta) + \frac{u+(1-K)u^{3}}{1+u^{2}} e_{0}(u) -\big((1+K)u\sqrt{1+u^{2}}\big)\nuh^{F}e_{F}(\zeta) \nonumber \\
    &+ \frac{-1+(K-1)u^{2}}{\sqrt{1+u^{2}}}\nuh^{F}e_{F}(u) - (1+u^{2})P^{FC}e_{F}(\nuh_{C}) + \frac{(3K-1)}{\alpha t}u^{2} + \texttt{F} = 0,
\end{align}
\begin{align}
\label{eqn:Euler2+1_2}
   & u\sqrt{1+u^{2}}e_{0}(\nuh_{Q}) - KP^{P}_{Q}e_{P}(\zeta) +\frac{Ku}{1+u^{2}}P^{P}_{Q}e_{P}(u) - u^{2}\nuh^{F}e_{F}(\nuh_{Q}) +\texttt{G} = 0, 
\end{align}
and 
\begin{align}
\label{eqn:Euler2+1_3}
    & \frac{Ku}{\sqrt{1+u^{2}}}e_{0}(\zeta) + \frac{1+(1-K)u^{2}}{(1+u^{2})^{\frac{3}{2}}}e_{0}(u) - K\nuh^{P}e_{P}(\zeta) + \frac{u(K-1)}{1+u^{2}}\nuh^{F}e_{F}(u) \nonumber \\
    &+\frac{(3K-1)}{\alpha t}\frac{u}{\sqrt{1+u^{2}}} +\texttt{R} = 0,
\end{align}
where 
\begin{align*}
\texttt{F} &:= \Bigg[-2(\dot{U}_{F}- A_{F})u\sqrt{1+u^{2}}\nuh^{F} + \mathcal{H}(3+4u^{2})+u^{2}\Sigma_{IF}\nuh^{I}\nuh^{F}\Bigg], \\
\texttt{G} &:= -(1+u^{2})P^{P}_{Q}\dot{U}_{P} -u^{2}P^{P}_{Q}A_{P} +u\sqrt{1+u^{2}}P^{P}_{Q}\tensor{\Sigma}{^I_P}\nuh_{I} + u^{2}P^{P}_{Q}\tensor{N}{_E^I}\tensor{\epsilon}{_P^{EF}}\nuh_{I}\nuh_{F}, \\
    \texttt{R} &:= \Bigg[\mathcal{H}\frac{u}{\sqrt{1+u^{2}}}  +u^{2}(A_{I}+\dot{U}_{I})\nuh^{I} -(1+u^{2})\nuh^{P}\dot{U}_{P} -u^{2}\nuh^{P}A_{P} +u\sqrt{1+u^{2}}\tensor{\Sigma}{^I_P}\nuh_{I}\nuh^{P} \nonumber \\
&\quad + u^{2}\tensor{N}{_E^I}\tensor{\epsilon}{_P^{EF}}\nuh_{I}\nuh_{F}\nuh^{P} -\frac{u^{3}}{\sqrt{1+u^{2}}}\Sigma_{IF}\nuh^{I}\nuh^{F}\Bigg].
\end{align*}

\subsection{Symmetric Hyperbolic Form of the Euler Equations}
\label{sec:SH_Conformal_Euler_Derivation}
It remains to express the relativistic Euler equations as a symmetric hyperbolic system.  First, by multiplying \eqref{eqn:Euler2+1_3} by $u\sqrt{1+u^{2}}$ and subtracting the resulting expression from \eqref{eqn:Euler2+1_1}, we obtain an evolution equation for $\zeta$ given by 
\begin{align*}
    &(1+u^{2})e_{0}(\zeta)  -\big(u\sqrt{1+u^{2}}\big)\nuh^{F}e_{F}(\zeta) - \frac{1}{\sqrt{1+u^{2}}}\nuh^{F}e_{F}(u) \nonumber \\
    &- (1+u^{2})P^{FC}e_{F}(\nuh_{C}) + \texttt{F} -u\sqrt{1+u^{2}}\texttt{R} = 0.
\end{align*}
Dividing the above equation by $(1+u^{2})$, we then find
\begin{align}
\label{eqn:Euler2+1_a}
    &e_{0}(\zeta) - \frac{u}{\sqrt{1+u^{2}}}\nuh^{F}e_{F}(\zeta) - \frac{1}{(1+u^{2})^{\frac{3}{2}}}\nuh^{F}e_{F}(u) \nonumber \\
    &-P^{FC}e_{F}(\nuh_{C}) + \frac{1}{1+u^{2}}\texttt{F} - \frac{u}{\sqrt{1+u^{2}}}\texttt{R} = 0.
\end{align}
Similarly, multiplying \eqref{eqn:Euler2+1_1} by $\frac{Ku}{\sqrt{1+u^{2}}(1+(1+K)u^{2})}$ and subtracting the result from \eqref{eqn:Euler2+1_3}, we get  
\begin{align}
\label{eqn:Euler2+1_1_b}
    & \frac{1-(K-1)u^{2}}{\sqrt{1+u^{2}}(1+(1+K)u^{2})}e_{0}(u) - \frac{K}{1+(1+K)u^{2}}\nuh^{P}e_{P}(\zeta) \nonumber \\
    &+ u\Big(-\frac{1}{1+u^{2}}+\frac{2K}{1+(1+K)u^{2}}\Big)\nuh^{F}e_{F}(u) +\frac{Ku\sqrt{1+u^{2}}}{1+(1+K)u^{2}}P^{FQ}e_{F}(\nuh_{Q})\nonumber \\
    &+\frac{(3K-1)}{\alpha t}\frac{u\sqrt{1+u^{2}}}{1+(1+K)u^{2}} +\texttt{R}-\frac{Ku}{\sqrt{1+u^{2}}(1+(1+K)u^{2})}\texttt{F} = 0.
\end{align}
Multiplying \eqref{eqn:Euler2+1_1_b} by $\frac{\sqrt{1+u^{2}}(1+(1+K)u^{2})}{1-(K-1)u^{2}}$ then yields 
\begin{align}
\label{eqn:Euler2+1_b}
    e_{0}(u) &+\frac{K\sqrt{1+u^{2}}}{-1+(K-1)u^{2}}\nuh^{F}e_{F}(\zeta) + \frac{u(1-2K-(K-1)u^{2})}{\sqrt{1+u^{2}}(-1+(K-1)u^{2})}\nuh^{F}e_{F}(u) \nonumber \\
    &+\frac{Ku(1+u^{2})}{1-(-1+K)u^{2}}P^{FQ}e_{F}(\nuh_{Q}) + \frac{(3K-1)}{\alpha t}\frac{u(1+u^{2})}{1-(K-1)u^{2}} \nonumber \\
    &+\frac{\sqrt{1+u^{2}}(1+(1+K)u^{2})}{1-(K-1)u^{2}}\texttt{R}-\frac{Ku}{1-(K-1)u^{2}}\texttt{F} = 0.
\end{align}
Finally, dividing \eqref{eqn:Euler2+1_2} by $u\sqrt{1+u^{2}}$ gives
\begin{align}
\label{eqn:Euler2+1_c}
    e_{0}(\nuh_{Q}) &- \frac{K}{u\sqrt{1+u^{2}}}P^{F}_{Q}e_{F}(\zeta) + \frac{K}{(1+u^{2})^{\frac{3}{2}}}P^{F}_{Q}e_{F}(u) \nonumber \\
    &- \frac{u}{\sqrt{1+u^{2}}}\nuh^{F}e_{F}(\nuh_{Q}) + \frac{1}{u\sqrt{1+u^{2}}}\texttt{G} = 0.
\end{align}

Next, following \cites{Oliynyk:2021,MarshallOliynyk:2022}, we introduce the re-scaled variables $\ut$ and $\wh_{Q}$ defined by 
\begin{align}
\label{eqn:Rescaled_FluidVars}
    u = t^{-\mu}\ut, \;\;\; \nuh_{Q} = t^{\mu}\wh_{Q},
\end{align}
where $\mu \in \mathbb{R}$ is a constant to be fixed below. It is straightforward to verify that $\ut$ and $\wh_{Q}$ satisfies the following equations
\begin{equation*}
\begin{aligned}
    e_{0}(u) &= \frac{1}{t^{\mu}}\Big(e_{0}(\ut) - \frac{\mu \ut}{\alpha t}\Big), \quad e_{F}(u) = \frac{1}{t^{\mu}}e_{F}(\ut), \\
    e_{0}(\nuh_{Q}) &= t^{\mu}\Big(e_{0}(\wh_{Q}) +\frac{\mu}{\alpha t}\wh_{Q}\Big), \quad e_{F}(\nuh_{Q}) = t^{\mu}e_{F}(\wh_{Q}).
\end{aligned}
\end{equation*}
Then expressing the Euler equations \eqref{eqn:Euler2+1_a}, \eqref{eqn:Euler2+1_b} and \eqref{eqn:Euler2+1_c} in terms of $\ut$ and $\wh_{Q}$ and multiplying on the left by the matrix Q
\begin{align*}
    Q := \begin{pmatrix} \frac{K}{\sqrt{t^{2\mu}+\ut^{2}}} & 0 & 0 \\ 0 & \frac{t^{2\mu}+(1-K)\ut^{2}}{(t^{2\mu}+\ut^{2})^{\frac{5}{2}}} & 0 \\ 0 & 0 & t^{-\mu}\ut
    \end{pmatrix},
\end{align*}
we obtain the following symmetric hyperbolic formulation of the Euler equations
\begin{align}
\label{eqn:MatrixForm_Euler2}
\hat{C}^{0}\del_{t}\begin{pmatrix}\zeta \\ \ut \\ \wh_{P}\end{pmatrix} +\hat{C}^{\Omega}\del_{\Omega}\begin{pmatrix}\zeta \\ \ut \\ \wh_{P}\end{pmatrix} +\frac{1}{t}\hat{\texttt{C}}\begin{pmatrix}\zeta \\ \ut \\ \wh_{P}\end{pmatrix} +\hat{\texttt{H}} = 0
\end{align}
where
\begin{align}
\label{eqn:Ahat_0}
    \hat{C}^{0} &= \alpha^{-1}\begin{pmatrix}\frac{K}{\sqrt{t^{2\mu}+\ut^{2}}} &0 &0 \\
    0 & \frac{t^{2\mu}+(1-K)\ut^{2}}{(t^{2\mu}+\ut^{2})^{\frac{5}{2}}} & 0 \\ 0 & 0 & \ut\delta^{P}_{Q} \end{pmatrix},  \\
\label{eqn:Ahat_F}
    \hat{C}^{\Omega} &= \begin{pmatrix} \frac{-K\ut}{t^{2\mu}+\ut^{2}}\nuh^{F} & \frac{-Kt^{2\mu}}{(t^{2\mu}+\ut^{2})^{2}}\nuh^{F} & \frac{-t^{\mu}K}{\sqrt{t^{2\mu}+\ut^{2}}}P^{FP}\\
    \frac{-Kt^{2\mu}}{(t^{2\mu}+\ut^{2})^{2}}\nuh^{F} & \frac{(-1+2K)t^{2\mu}\ut+(K-1)\ut^{3}}{(t^{2\mu}+\ut^{2})^{3}}\nuh^{F} & \frac{Kt^{\mu}\ut}{(t^{2\mu}+\ut^{2})^{\frac{3}{2}}}P^{FP} \\ \frac{-Kt^{\mu}}{\sqrt{t^{2\mu}+\ut^{2}}}P^{F}_{Q} & \frac{Kt^{\mu}\ut}{(t^{2\mu}+\ut^{2})^{\frac{3}{2}}}P^{F}_{Q} & -\frac{\ut^{2}}{\sqrt{t^{2\mu}+\ut^{2}}}\nuh^{F}\delta^{P}_{Q}
    \end{pmatrix}e^{\Omega}_{F}, \\
    \label{eqn:Acal_def}
    \hat{\texttt{C}} &= \begin{pmatrix}0&0&0 \\ 0&\frac{-t^{2\mu}(1-3K+\mu)+\ut^{2}(-1-\mu+K(3+\mu))}{(t^{2\mu}+\ut^{2})^{\frac{5}{2}}\alpha}& 0 \\ 0&0& \frac{\mu \ut}{\alpha}\delta^{P}_{Q} \end{pmatrix},\\
    \label{eqn:textttH_def}
    \hat{\texttt{H}} &=  \begin{pmatrix} \frac{\texttt{F}Kt^{2\mu} - K\texttt{R}\ut\sqrt{t^{2\mu}+\ut^{2}}}{(t^{2\mu}+\ut^{2})^{\frac{3}{2}}} \\ \frac{\texttt{F}Kt^{2\mu}\ut + \texttt{R}(t^{2\mu}+(1+K)\ut^{2})\sqrt{t^{2\mu}+\ut^{2}}}{(t^{2\mu}+\ut^{2})^{\frac{5}{2}}} \\ \frac{t^{\mu}}{\sqrt{t^{2\mu}+\ut^{2}}}\texttt{G} \end{pmatrix}.
\end{align}
In terms of the variables $\ut$ and $\wh_{Q}$, the source terms $\texttt{F}$, $\texttt{R}$, and $\texttt{G}$ are given by
\begin{equation}
\begin{aligned}
\label{eqn:Source_ttt_defns}
    \texttt{F} &= -2(\dot{U}_{F}-A_{F})t^{-\mu}\ut\sqrt{t^{2\mu}+\ut^{2}}\wh^{F} + \mathcal{H}t^{-2\mu}(3t^{2\mu}+4\ut^{2}) + \ut^{2}\Sigma_{IF}\wh^{I}\wh^{F}, \\
    \texttt{R} &= \mathcal{H}\frac{\ut}{\sqrt{t^{2\mu}+\ut^{2}}} - t^{\mu}\wh^{P}\dot{U}_{P} + \frac{t^{2\mu}\ut}{\sqrt{t^{2\mu}+\ut^{2}}}\Sigma_{IF}\wh^{I}\wh^{F}, \\
    \texttt{G} &= -t^{-2\mu}(t^{2\mu}+\ut^{2})P^{P}_{Q}\dot{U}_{P} - t^{-2\mu}\ut^{2}P^{P}_{Q}A_{P} + t^{-\mu}\ut\sqrt{t^{2\mu}+\ut^{2}}P^{P}_{Q}\tensor{\Sigma}{^I_P}\wh_{I}  \\
    &+ \ut^{2}P^{P}_{Q}\tensor{N}{_E^I}\tensor{\epsilon}{_P^{EF}}\wh_{I}\wh_{F},
\end{aligned}
\end{equation}
respectively.

\subsection{Subtracting the Homogeneous Fluid Background}
\label{sec:Euler_Subtract_Background}
Let us now introduce the new variables $\uh$, $\frg(t)$, and $w_{Q}$ by
\begin{align}
\label{eqn:uh_w_defns}
    \ut = e^{\uh+\frg(t)}, \;\; \wh_{Q} = w_{Q} + t^{-\mu}\xi_{Q},
\end{align}
where $\xi_{Q}$ is a constant vector satisfying $|\xi|=1$. Below, we will choose $\frg(t)$ and $\xi_{Q}$ so that they define a homogeneous solution to the Euler equations on a fixed de Sitter background. That is to say, by switching to the variables $\uh$ and $w_{Q}$, we are subtracting a homogeneous solution to Euler equations on a de Sitter background. \newline \par 

Now, after a short calculation, we find
\begin{equation*}
\begin{aligned}
    e_{0}(\ut) &= e^{\uh+\frg(t)}\big(e_{0}(\uh) + e_{0}(\frg(t))\big), \quad e_{F}(\ut) = e^{\uh+\frg(t)}e_{F}(\uh), \\
    e_{0}(\wh_{Q}) &= e_{0}(w_{Q}) - \frac{\mu t^{-\mu}}{\alpha t}\xi_{Q}, \quad e_{F}(\wh_{Q}) = e_{F}(w_{Q}).
\end{aligned}
\end{equation*}
Then expressing \eqref{eqn:MatrixForm_Euler2} in terms of $\uh$ and multiplying on the left by the matrix M
\begin{align*}
   M := \begin{pmatrix} 1 & 0 & 0 \\ 0 & e^{\uh+\frg} & 0 \\ 0 & 0 & 1 \end{pmatrix},
\end{align*}
we arrive at the following form of the Euler equations
\begin{align} 
\label{eqn:MatrixForm_Euler3}
{C}^{0}\del_{t}W +C^{\Omega}\del_{\Omega}W = -\frac{\mu \ut}{t\alpha}\Pi W  + \texttt{M} + \texttt{H}
\end{align}
where
\begin{align}
    W &= ( \zeta , \uh ,w_{P})^T, \\
    \label{eqn:A0_final}
    C^{0} &= \alpha^{-1}\begin{pmatrix}\frac{K}{\sqrt{t^{2\mu}+\ut^{2}}} &0 &0 \\
    0 & \frac{t^{2\mu}+(1-K)\ut^{2}}{(t^{2\mu}+\ut^{2})^{\frac{5}{2}}}\ut^{2} & 0 \\ 0 & 0 & \ut\delta^{P}_{Q} \end{pmatrix},  \\
    \label{eqn:AF_final}
    C^{\Omega} &= \begin{pmatrix} \frac{-K\ut}{t^{2\mu}+\ut^{2}}\nuh^{F} & \frac{-Kt^{2\mu}\ut }{(t^{2\mu}+\ut^{2})^{2}}\nuh^{F} & \frac{-t^{\mu}K}{\sqrt{t^{2\mu}+\ut^{2}}}P^{FP}\\
    \frac{-Kt^{2\mu}\ut}{(t^{2\mu}+\ut^{2})^{2}}\nuh^{F} & \frac{(-1+2K)t^{2\mu}\ut+(K-1)\ut^{3}}{(t^{2\mu}+\ut^{2})^{3}}\ut^{2}\nuh^{F} & \frac{Kt^{\mu}\ut^{2}}{(t^{2\mu}+\ut^{2})^{\frac{3}{2}}}P^{FP} \\ \frac{-Kt^{\mu}}{\sqrt{t^{2\mu}+\ut^{2}}}P^{F}_{Q} & \frac{Kt^{\mu}\ut^{2}}{(t^{2\mu}+\ut^{2})^{\frac{3}{2}}}P^{F}_{Q} & -\frac{\ut^{2}}{\sqrt{t^{2\mu}+\ut^{2}}}\nuh^{F}\delta^{P}_{Q}
    \end{pmatrix}e^{\Omega}_{F}, \\
    \label{eqn:Pi_def}
    \Pi &= \begin{pmatrix}0&0&0\\0&0&0\\0&0&\delta^{P}_{Q}\end{pmatrix}, \\
    \texttt{M} &= \begin{pmatrix}0 \\  \Big(\frac{t^{2\mu}(1-3K+\mu)-\ut^{2}(-1-\mu+K(3+\mu)) -t\frg^{\prime}(t^{2\mu}+(1-K)\ut^{2})}{t(t^{2\mu}+\ut^{2})^{\frac{5}{2}}\alpha}\Big)\ut^{2}  \\ 0 \end{pmatrix}, \\
    \label{eqn:Httt_def}
    \texttt{H} &=  \begin{pmatrix} \frac{\texttt{F}Kt^{2\mu} -K\texttt{R}\ut\sqrt{t^{2\mu}+\ut^{2}}}{(t^{2\mu}+\ut^{2})^{\frac{3}{2}}} \\ \frac{\texttt{F}Kt^{2\mu}\ut^{2} + \texttt{R}\ut(t^{2\mu}+(1+K)\ut^{2})\sqrt{t^{2\mu}+\ut^{2}}}{(t^{2\mu}+\ut^{2})^{\frac{5}{2}}} \\ \frac{t^{\mu}}{\sqrt{t^{2\mu}+\ut^{2}}}\texttt{G} \end{pmatrix}.
\end{align}
Additionally, we observe that the projection operator $\Pi$ \eqref{eqn:Pi_def} satisfies the following conditions
\begin{align}
\label{eqn:Pi_perp_defn}
\Pi^{\perp} = \mathbbm{1}-\Pi, \;\; \Pi^{2} = \Pi, \;\; (\Pi^{\perp})^{2} = \Pi^{\perp}, \;\; \Pi\Pi^{\perp} = \Pi^{\perp}\Pi = 0, \;\; \Pi+\Pi^{\perp} = \mathbbm{1}.
\end{align}

Multiplying \eqref{eqn:MatrixForm_Euler3} on the left by $(C^{0})^{-1}$ we then obtain
\begin{align}
\label{eqn:MatrixForm_Euler4}
    \del_{t}W + \mathbf{C}^{\Omega}\del_{\Omega}W = -\frac{\mu}{t}\Pi W + \mathbf{M} + \mathbf{H}
\end{align}
where
\begin{align}  
\label{eqn:Abf_def}
\mathbf{C}^{\Omega} &:= (C^{0})^{-1}C^{\Omega} \nonumber \\
&= \alpha\begin{pmatrix} \frac{-\ut}{\sqrt{t^{2\mu}+\ut^{2}}}\nuh^{F} & \frac{-t^{2\mu}\ut}{(t^{2\mu}+\ut^{2})^{\frac{3}{2}}}\nuh^{F} & -t^{\mu}P^{FP}\\ \frac{-Kt^{2\mu}\ut^{-1}\sqrt{t^{2\mu}+\ut^{2}}}{t^{2\mu}+(1-K)\ut^{2}}\nuh^{F} & \frac{(-1+2K)t^{2\mu}\ut+(K-1)\ut^{3}}{\sqrt{t^{2\mu}+\ut^{2}}(t^{2\mu}+(1-K)\ut^{2})}\nuh^{F} & \frac{Kt^{\mu}(t^{2\mu}+\ut^{2}) }{t^{2\mu}+(1-K)\ut^{2}}P^{FP} \\ \frac{-Kt^{\mu}}{\ut\sqrt{t^{2\mu}+\ut^{2}}}P^{F}_{Q} & \frac{Kt^{\mu} \ut}{(t^{2\mu}+\ut^{2})^{\frac{3}{2}}}P^{F}_{Q} & -\frac{\ut}{\sqrt{t^{2\mu}+\ut^{2}}}\nuh^{F}\delta^{P}_{Q}
    \end{pmatrix}e^{\Omega}_{F}, \\
\mathbf{M} &:= (C^{0})^{-1}\texttt{M} = \alpha\begin{pmatrix}0 \\  \Big(\frac{t^{2\mu}(1-3K+\mu)-\ut^{2}(-1-\mu+K(3+\mu)) -t\frg^{\prime}(t^{2\mu}+(1-K)\ut^{2})}{t(t^{2\mu}+(1-K)\ut^{2})\alpha}\Big)  \\ 0 \end{pmatrix}, \\
\label{eqn:Hbf_def}
\mathbf{H} &:=(C^{0})^{-1}\texttt{H} =  \alpha\begin{pmatrix} \frac{\texttt{F}t^{2\mu} -\texttt{R}\ut\sqrt{t^{2\mu}+\ut^{2}}}{t^{2\mu}+\ut^{2}} \\ \frac{\texttt{F}Kt^{2\mu}\ut^{2} + \texttt{R}\ut(t^{2\mu}+(1+K)\ut^{2})\sqrt{t^{2\mu}+\ut^{2}}}{t^{2\mu}\ut^{2}+(1-K)\ut^{4}} \\ \frac{t^{\mu}}{\ut\sqrt{t^{2\mu}+\ut^{2}}}\texttt{G} \end{pmatrix}.
\end{align}

The source term $\mathbf{M}$ can be decomposed as 
\begin{align}
\label{eqn:M_decomp_Euler}
    \mathbf{M} = t^{2\mu-1}\mathbf{S} + \mathbf{M}_{0}
\end{align}
where
\begin{align}
   \mathbf{M}_{0} := \mathbf{M}|_{W=0} &= \alpha \begin{pmatrix}0 \\  \Big(\frac{t^{2\mu}(1-3K+\mu)-e^{2\frg}(-1-\mu+K(3+\mu)) -t\frg^{\prime}(t)(t^{2\mu}+(1-K)e^{2\frg})}{t(t^{2\mu}+(1-K)e^{2\frg}))\alpha}\Big)  \\ 0 \end{pmatrix}, \\
   \label{eqn:Sbf_def}
   \mathbf{S} &= \alpha\begin{pmatrix} 0 \\ -\frac{e^{2\frg}(e^{2\uh}-1)K(3K-1)}{\big(e^{2\frg}(K-1)-t^{2\mu}\big)\big(e^{2(\uh+\frg)}(K-1)-t^{2\mu}\big)\alpha} \\ 0 \end{pmatrix}.
\end{align}
 Observe that when $V=W=0$, which corresponds to a homogeneous reduction of the Euler equations on a de Sitter background, the Euler equations become 
\begin{align*}
   t^{2\mu}(1-3K+\mu)-e^{2\frg}(-1-\mu+K(3+\mu)) -t\frg^{\prime}(t^{2\mu}+(1-K)e^{2\frg}) = 0.
\end{align*} 
Hence, for the following choice of $\mu$
\begin{align}
\label{eqn:mu_def}
    \mu = \frac{3K-1}{1-K}, 
\end{align}
homogeneous solutions to the Euler equations are determined by the (singular) initial value problem
\begin{align*}
    \frg^{\prime}(t) &= \frac{t^{2\mu-1}(\mu-3K+1)}{t^{2\mu}+(1-K)e^{2\frg}}, \quad 0 < t \leq 1, \\
    \frg(0) &= \frg_{0}.
\end{align*}
Note that the existence of solutions to this IVP is guaranteed by Proposition \ref{prop:Asymptotic_ODE} from Appendix \ref{appendix:approximate_background}. \newline \par

For the remainder of this article, we take $\mu$ to be defined by \eqref{eqn:mu_def} and $\frg(t)$ to be a  solution of the above initial value problem. 
In this case, $\mathbf{M}_{0}$ vanishes by construction and the Euler equations \eqref{eqn:MatrixForm_Euler4} reduce to
\begin{align}
\label{eqn:Euler_coord_deriv1}
    \del_{t}W + \mathbf{C}^{\Sigma}\del_{\Sigma}W = -\frac{\mu}{t}\Pi W + t^{2\mu-1}\mathbf{S} +  \mathbf{H}.
\end{align}

\section{Local-in-time Existence and Constraint Propagation}
\label{sec:Local_Existence_ConstraintProp}
In this section, we use standard existence and uniqueness results for systems of symmetric hyperbolic equations to establish the local-in-time existence and uniqueness of solutions to the reduced Einstein-Euler system, which is the system obtained by combining \eqref{eqn:Einstein_MatrixForm} and \eqref{eqn:MatrixForm_Euler3}, that is,
\begin{align}
\label{eqn:EinsteinEuler_NotFuchsian}
    \begin{pmatrix}B^{0} & 0 \\ 0 & \alpha C_{(\afrak)}^{0} \end{pmatrix}\del_{t}\mathcal{U} + \begin{pmatrix}-\alpha B^{K}e^{\Sigma}_{K} & 0 \\ 0 & \alpha C^{K}_{(\afrak)}e^{\Sigma}_{K} \end{pmatrix}\del_{\Sigma}\mathcal{U} &= \frac{1}{t}\begin{pmatrix}\alpha\mathcal{B} & 0 \\ 0 & -\mu_{(\afrak)}\ut_{(\afrak)} \end{pmatrix}\mathscr{P}\mathcal{U}\notag \\
    &\quad + \begin{pmatrix} \alpha \mathcal{F} \\ \alpha (\texttt{M}_{(\afrak)}+\texttt{H}_{(\afrak)}) \end{pmatrix},
\end{align}
where
\begin{align}
\label{eqn:ufrak_defn}
    \mathcal{U} &= \begin{pmatrix} V \\ W_{(\afrak)} \end{pmatrix}, \;\; W_{(\afrak)} = \begin{pmatrix}\zeta_{(\afrak)} \\ \uh_{(\afrak)} \\ w^{(\afrak)}_{P}\end{pmatrix} \;\; \text{and} \;\; \mathscr{P} = \begin{pmatrix}\mathcal{P} & 0 \\ 0 & \Pi \end{pmatrix}.
\end{align}
In addition, we establish a continuation principle for these solutions and show that if the initial data is chosen so that the constraints \eqref{eqn:C1_constraint}-\eqref{eqn:Hamiltonian_Constraint} and \eqref{eqn:C5_nu_normalisation_constraint} vanish on the initial hypersurface, then they will vanish on the entire region where the solution is defined. The precise statement of these results is given in the following proposition.

\begin{prop}
\label{prop:Local_ExistenceUniqueness_Prop}
    Suppose $T_{1}>0$, $k \in \mathbb{Z}_{>\frac{3}{2}+1}$, $\frg_{(\afrak)}(t)$ is the unique solution to the IVP \eqref{eqn:Homogeneous_ODE_1}-\eqref{eqn:Homogeneous_ODE_2}, and $\mathcal{U}_{0} \in H^{k}(\mathbb{T}^{3},\mathbb{R}^{45})$ . Then there exists $T_{0} \in (0,T_{1})$ and a unique solution 
    \begin{equation*}
        \mathcal{U}\in C^{0}\big((T_{0},T_{1}],H^{k}(\mathbb{T}^{3},\mathbb{R}^{45})\big)\cap C^{1}\big((T_{0},T_{1}],H^{k-1})(\mathbb{T}^{3},\mathbb{R}^{45})\big)
    \end{equation*} 
   of the reduced Einstein-Euler equations \eqref{eqn:EinsteinEuler_NotFuchsian} 
satisfying the initial condition $\mathcal{U}|_{t=T_1}=\mathcal{U}_0$. Moreover, the following hold:
    \begin{enumerate}[(a)]
        \item If $\sup_{T_0<t\leq T_1}\|\mathcal{U}(t)\|_{W^{1,\infty}} < \infty $, then there exists a time $T_{*}\in [0,T_{0})$ such the solution $\mathcal{U}$ can be uniquely continued to the region $(T_{*},T_{1}] \times \mathbb{T}^{3}$.
        \item If the initial data $\mathcal{U}_{0}$ satisfies the constraint equations \eqref{eqn:C1_constraint}-\eqref{eqn:Hamiltonian_Constraint} and \eqref{eqn:C5_nu_normalisation_constraint}, then the solution $\mathcal{U}$ satisfies these constraints everywhere on $(T_{0},T_{1}] \times \mathbb{T}^{3}$ and determines a solution of the Einstein-Euler equations \eqref{eqn:Einstein_Physical}-\eqref{eqn:Euler_Physical} on $(T_{0},T_{1}] \times \mathbb{T}^{3}$.
    \end{enumerate}
\end{prop}
\begin{proof}$\;$
By construction, the system \eqref{eqn:EinsteinEuler_NotFuchsian} is symmetric hyperbolic. Hence, given initial data $\mathcal{U}_{0} \in H^{k}(\mathbb{T}^{3},\mathbb{R}^{45})$ with $k\in \mathbb{Z}_{>\frac{3}{2}+1}$,  an application of standard local-in-time existence and uniqueness theorems for symmetric hyperbolic systems, e.g. \cite{TaylorIII:1996}*{Ch.~16, Prop.~2.1}, guarantees, for any $T_1>0$, the existence of a $T_0\in (0,T_1)$ and a unique solution  
\begin{equation} \label{loc-sol}
    \mathcal{U}\in C^{0}\big((T_{0},T_{1}],H^{k}(\mathbb{T}^{3},\mathbb{R}^{45})\big)\cap C^{1}\big((T_{0},T_{1}],H^{k-1})(\mathbb{T}^{3},\mathbb{R}^{45})\big)
\end{equation}
of the reduced Einstein-Euler equations \eqref{eqn:EinsteinEuler_NotFuchsian} 
satisfying the initial condition $\mathcal{U}|_{t=T_1}=\mathcal{U}_0$. 
This establishes the local-in-time existence and uniqueness of solutions.

\subsubsection*{Proof of (a)}
The claim follows immediately from the continuation principle for symmetric hyperbolic systems, e.g. Propositions 1.5 and 2.1 from \cite{TaylorIII:1996}*{Ch.~16}.

\subsubsection*{Proof of (b)}
To prove this claim, we first verify that normalisation constraint \eqref{eqn:C5_nu_normalisation_constraint}  for the fluid velocity propagates. In terms of $w^{(\afrak)}_{Q}$, this constraint becomes
\begin{align*}
    \mathcal{C}^{(\afrak)}_{5} := w^{(\afrak)}_{Q}w_{(\afrak)}^{Q}+ 2t^{-\mu_{(\afrak)}}w^{(\afrak)}_{Q}\xi_{(\afrak)}^{Q} + t^{-2\mu_{(\afrak)}}\xi^{(\afrak)}_{Q}\xi_{(\afrak)}^{Q} -t^{-2\mu_{(\afrak)}}
\end{align*}
To simplify the presentation, we drop the fluid indices and observe that $e_{0}(\mathcal{C}_{5})$ is given by
\begin{align}
\label{eqn:Constraint_Prop1}
    e_{0}(\mathcal{C}_{5}) = 2(w^{Q}+t^{-\mu}\xi^{Q})e_{0}(w_{Q}) -\frac{2\mu}{\alpha}t^{-\mu-1}w_{Q}\xi^{Q} - \frac{2\mu}{\alpha}t^{-2\mu-1}\xi_{Q}\xi^{Q}+\frac{2\mu}{\alpha}t^{-2\mu-1}.
\end{align}
Then using \eqref{eqn:Source_ttt_defns}, \eqref{eqn:MatrixForm_Euler3}, and the identities
\begin{equation*}
    (w^{Q}+t^{-\mu}\xi^{Q})P^{P}_{Q} = -t^{2\mu}\mathcal{C}_{5}(w^{P}+t^{-\mu}\xi^{P}) \quad \text{and} \quad
    (w^{P}+t^{-\mu}\xi^{P})e_{F}(w_{P}) = \frac{1}{2}e_{F}(\mathcal{C}_{5}),
\end{equation*}
we can express \eqref{eqn:Constraint_Prop1} as
\begin{align*} 
    e_{0}(\mathcal{C}_{5}) + \frac{t^{2\mu}\ut}{\sqrt{t^{2\mu}+\ut^{2}}}\nuh^{F} e_{F}(\mathcal{C}_{5}) &= \Xi\mathcal{C}_{5}
\end{align*}
where 
\begin{align*}
   \Xi &=  \frac{-2K t^{3\mu}}{\ut\sqrt{t^{2\mu}+\ut^{2}}}(w^{F}+t^{-\mu}\xi^{F})e_{F}(\zeta) + \frac{2Kt^{3\mu}}{(t^{2\mu}+\ut^{2})^{\frac{3}{2}}}(w^{F}+t^{-\mu}\xi^{F})e_{F}(\ut) \nonumber \\
    &- \frac{2t^{2\mu}\ut}{\sqrt{t^{2\mu}+\ut^{2}}}(w^{P}+t^{-\mu}\xi^{P})\nuh^{F}e_{F}(w_{P})  +\frac{2t^{3\mu}}{\ut\sqrt{t^{2\mu}+\ut^{2}}}\wh^{P} \Big[-t^{-2\mu}(t^{2\mu}+\ut^{2})\dot{U}_{P} \nonumber \\
    &- t^{-2\mu}\ut^{2}A_{P} + t^{-\mu}\ut\sqrt{t^{2\mu}+\ut^{2}}\tensor{\Sigma}{^I_P}\wh_{I} + \ut^{2}\tensor{N}{_E^I}\tensor{\epsilon}{_P^{EF}}\wh_{I}\wh_{F}\Big] - \frac{2\mu}{\alpha t}.
\end{align*}
Clearly, this equation is linear and homogeneous in the constraint quantity $\mathcal{C}_{5}$. By uniqueness of solutions to transport equations, it follows immediately that if $\mathcal{C}_{5}$ vanishes initially at $t=T_1$, then it will vanish for all $t\in (T_0,T_1]$. This shows that the constraint $\mathcal{C}_{5}$ propagates. \newline \par 

It remains to verify that the constraints \eqref{eqn:C1_constraint}-\eqref{eqn:Hamiltonian_Constraint} propagate. Noting that the conformal Einstein-Euler equations, see \eqref{eqn:Euler_Physical}, \eqref{Tt-def}, \eqref{eqn:conformal_metric_def},  \eqref{eqn:ConformalEinstein} and \eqref{eqn:Conformal_Tmunu_def}, can be expressed as 
\begin{align*}
    G_{\mu\nu} &= \bar{T}_{\mu\nu},
    \\
    \nabla_{\mu}\bar{T}^{\mu\nu} &= 0, 
\end{align*}
where 
\begin{equation*}
   \bar{T}_{\mu\nu} = e^{2\Phi}T_{\mu\nu}-e^{2\Phi}\Lambda g_{\mu\nu} + 2\bigl(\nabla_\mu \nabla_\nu \Phi -\nabla_\mu \Phi \nabla_\nu \Phi\bigr)-\bigl(2\Box \Phi + |\nabla\Phi|_g^2\bigr)g_{\mu\nu}
\end{equation*}
and $G_{\mu\nu}=R_{\mu\nu}-\frac{1}{2}Rg_{\mu\nu}$ is the Einstein tensor of the conformal metric $g_{\mu\nu}$,
we can employ the BSSN formulation from \cite{BeyerSarbach:2004} to express these equations as a symmetric hyperbolic system of equations in the same gauge used in this article, that is, zero-shift and generalised harmonic slicing determined by   \eqref{eqn:Harmonic_TimeSlicing} and \eqref{eqn:gauge_source_f}. Additionally, as discussed in \cite{BeyerSarbach:2004}, the constraint propagation equations in the BSSN formalism can be expressed as a symmetric hyperbolic system of equations that is linear in the constraint quantities. In particular, this implies that the Momentum and Hamiltonian constraints, i.e. \eqref{eqn:MomentumConstraint} and \eqref{eqn:Hamiltonian_Constraint},  propagate, that is, if they vanish initial at $t=T_1$, then they vanish for all $t\in (T_1,T_0]$.

Now, given initial data
\begin{equation*}
\mathfrak{J}_{0}=\bigl(g_{\mu\nu},\rho^{(\afrak)},v_\mu^{(\afrak)}\bigr)\bigr|_{t=T_1} \in H^{k+1}(\mathbb{T}^3)\times H^k(\mathbb{T}^3)\times H^k(\mathbb{T}^3)
\end{equation*}
satisfying the Momentum and Hamiltonian constraints, which is \textit{the same physical initial data} as used to generate the solution $\mathcal{U}$,
the BSSN formulation of the conformal Einstein-Euler equations discussed above can be used to generate a solution \begin{equation*}
\mathfrak{J}= \bigl(g_{\mu\nu},\rho^{(\afrak)},v_\mu^{(\afrak)}\bigr)
\in \bigcap_{j=0}^1 C^{0}\big((T_{0},T_{1}], H^{k+1-j}(\mathbb{T}^3)\times H^{k-j}(\mathbb{T}^3)\times H^{k-j}(\mathbb{T}^3)\big)
\end{equation*}
of conformal Einstein-Euler equations satisfying the zero-shift and harmonic time slicing gauge employed in this article. Note that the Hamiltonian and Momentum constraints vanish everywhere, and if necessary, we have moved $T_0$ closer to $T_1$ in \eqref{loc-sol} so that $\mathcal{U}$ and $\mathfrak{J}$ are defined on the same spacetime region.  \newline \par

From $\mathfrak{J}$, we can construct a new solution $\mathcal{U}_*$ of \eqref{eqn:EinsteinEuler_NotFuchsian}. This starts by setting $e_{0} = \alpha^{-1}\del_{t}$, where now $\alpha$ is the lapse determined by the solution  $\mathfrak{J}$. We then solve the Fermi-Walker transport equations
\begin{align*}
    \nabla_{e_{0}}e_{A} = -\frac{g(\nabla_{e_{0}}e_{0},e_{A})}{g(e_{0},e_{0})}e_{0}
\end{align*}
to obtain an orthogonal spatial frame $e_A$ that completes $e_0$ to an orthonormal frame. Here, the initial data for the spatial frame $e_A$ is the same as the initial data used to generate the solution $\mathcal{U}$. \newline \par 

Once, we have the orthonormal frame $e_i$, then we can use it together with the solution $\mathfrak{J}$ to construct, following the calculation carried out in Sections \ref{sec:Frame_Formulation} to \ref{sec:Sym_Hyp_Euler}, a new solution $\mathcal{U}_*$ to \eqref{eqn:EinsteinEuler_NotFuchsian} of the same regularity as $\mathcal{U}$, i.e.~\eqref{loc-sol}. It is important to note here that $\mathcal{U}_*$ satisfies all of the constraints \eqref{eqn:C1_constraint}-\eqref{eqn:Hamiltonian_Constraint}. By design $\mathcal{U}|_{t=T_1}=\mathcal{U}_*|_{t=T_1}$, and so, since $\mathcal{U}$ and $\mathcal{U}_*$ satisfy the same symmetric hyperbolic equation, we deduce from the uniqueness of solutions to symmetric hyperbolic equations that $\mathcal{U}=\mathcal{U}_*$. This completes the proof of claim (b). 
\end{proof}

\section{Differentiated Einstein-Euler System}
\label{sec:Differentiated_Einstein_Euler_System}
While the Fuchsian formulation of the Einstein-Euler equations obtained from \eqref{eqn:SH_Einstein_Sec3} and \eqref{eqn:Euler_coord_deriv1}, that is,
\begin{align}
\label{eqn:SH_Einstein_NoDerivs}
    B^0 \del_{t}V &= \alpha B^{\Omega}\del_{\Omega}V + \frac{1}{t}\alpha\tilde{\mathcal{B}}V + \alpha\mathcal{F}, \\
\label{eqn:SH_Euler_Noderivs}
    \del_{t}W + \mathbf{C}^{\Sigma}\del_{\Sigma}W &= -\frac{\mu}{t}\Pi W + t^{2\mu-1}\mathbf{S} +  \mathbf{H},
\end{align}
is locally well-posed, it does not have the correct structure to be able to directly apply the global existence theory from \cite{BOOS:2021}.
As discussed in the introduction, to bring the system into the required form, we must augment equations \eqref{eqn:SH_Einstein_NoDerivs}-\eqref{eqn:SH_Euler_Noderivs} with their spatial derivatives up to second order. The goal of this section is to derive the differentiated Einstein-Euler system, which will be the focus of our subsequent analysis. Specifically, in Section \ref{sec:Differentiated_Einstein_deriv}, we derive the differentiated Einstein equations, followed by the derivation of the differentiated Euler equations in Section \ref{sec:Differentiated_Euler_deriv}.

\subsection{Differentiated Einstein Equations}
\label{sec:Differentiated_Einstein_deriv}
In this section, we derive a symmetric hyperbolic formulation of the first and second differentiated (reduced) Einstein equations. We begin by multiplying \eqref{eqn:SH_Einstein_NoDerivs} on the left by $(B^{0})^{-1}$ to get
\begin{align}
\label{eqn:Einstein_MatrixForm4}
 \del_{t}V &= \mathbf{B}^{\Sigma}\del_{\Sigma}V +\frac{1}{t}\mathbf{Z}V + \mathbf{F} 
\end{align}
where
\begin{equation}
\begin{aligned}
\label{eqn:Grav_bf_Variables_Defn}
   \mathbf{B}^{\Sigma} &:=  \alpha(B^{0})^{-1}B^{\Sigma}, \\ 
   \mathbf{Z} &:=  \alpha(B^{0})^{-1}\mathcal{B}\mathcal{P} =2(B^0)^{-1}\mathcal{P}, \\ 
   \mathbf{F} &:=  \alpha(B^{0})^{-1}\mathcal{F}.
\end{aligned}
\end{equation}

To proceed, we introduce first and second order gravitational variables by
\begin{align}
\label{eqn:V_derivative_Variable_defns}
    V_{\Sigma} = \del_{\Sigma}V \quad \text{and} \quad V_{\Omega\Sigma} = t^{\psi}\del_{\Omega}\del_{\Sigma}V,
\end{align}
where $\psi>0$ is a constant to be determined.
Taking the first and second spatial derivatives of \eqref{eqn:Einstein_MatrixForm4}, we obtain the equations
\begin{align*}
    \del_{t}\del_{\Omega}V = \mathbf{B}^{\Sigma}\del_{\Omega}\del_{\Sigma}V + \del_{\Omega}\mathbf{B}^{\Sigma}\del_{\Sigma}V + \frac{1}{t}\mathbf{Z}\del_{\Omega}V + \del_{\Omega}\mathbf{F},
\end{align*}
and
\begin{align*}
\del_{t}\del_{\Gamma}\del_{\Omega}V &= \mathbf{B}^{\Sigma}\del_{\Sigma}\del_{\Gamma}\del_{\Omega}V + \del_{\Gamma}\mathbf{B}^{\Sigma}\del_{\Omega}\del_{\Sigma}V +\del_{\Omega}\mathbf{B}^{\Sigma}\del_{\Gamma}\del_{\Sigma}V+ \del_{\Gamma}\del_{\Omega}\mathbf{B}^{\Sigma}\del_{\Sigma}V \nonumber \\
&+ \frac{1}{t}\mathbf{Z}\del_{\Gamma}\del_{\Omega}V  + \del_{\Gamma}\del_{\Omega}\mathbf{F}.
\end{align*}
Multiplying the second equation by $t^{\psi}$ and expressing both equations in terms of $V_{\Sigma}$ and $V_{\Omega\Sigma}$ yields
\begin{align}  
\label{eqn:Einstein_1stderiv_noB0}
\del_{t}V_{\Omega} = \mathbf{B}^{\Sigma}\del_{\Sigma}V_{\Omega}+ \del_{\Omega}\mathbf{B}^{\Sigma}V_{\Sigma} +\frac{1}{t}\mathbf{Z}V_{\Omega} + \del_{\Omega}\mathbf{F}
\end{align}
and 
\begin{align}
\label{eqn:Einstein_2ndderiv_noB0}
    \del_{t}V_{\Gamma\Omega} &= \mathbf{B}^{\Sigma}\del_{\Sigma}V_{\Gamma\Omega} + \del_{\Gamma}\mathbf{B}^{\Sigma}V_{\Omega\Sigma} + \del_{\Omega}\mathbf{B}^{\Sigma}V_{\Gamma\Sigma} + t^{\psi}\del_{\Gamma}\del_{\Omega}\mathbf{B}^{\Sigma}V_{\Sigma} + \frac{1}{t}\mathbf{Z}V_{\Gamma\Omega} \nonumber \\
    &+ t^{\psi}\del_{\Gamma}\del_{\Omega}\mathbf{F} +\frac{\psi}{t}V_{\Gamma\Omega}.
\end{align}
Finally, after multiplying \eqref{eqn:Einstein_1stderiv_noB0} and \eqref{eqn:Einstein_2ndderiv_noB0} by $B^{0}$, we arrive at the following symmetric hyperbolic system of equations for the first and second order gravitational variables:
\begin{align}
\label{eqn:Einstein_1stderiv}
B^{0}\del_{t}V_{\Omega} &= B^{0}\mathbf{B}^{\Sigma}\del_{\Sigma}V_{\Omega}+ B^{0}\del_{\Omega}\mathbf{B}^{\Sigma}V_{\Sigma} +\frac{1}{t}B^{0}\mathbf{Z}V_{\Omega} + B^{0}\del_{\Omega}\mathbf{F}, \\
\label{eqn:Einstein_2ndderiv}
B^{0}\del_{t}V_{\Gamma\Omega} &= B^{0}\mathbf{B}^{\Sigma}\del_{\Sigma}V_{\Gamma\Omega} + B^{0}\del_{\Gamma}\mathbf{B}^{\Sigma}V_{\Omega\Sigma} + B^{0}\del_{\Omega}\mathbf{B}^{\Sigma}V_{\Gamma\Sigma} \nonumber \\
&+ t^{\psi}B^{0}\del_{\Gamma}\del_{\Omega}\mathbf{B}^{\Sigma}V_{\Sigma} + \frac{2}{t}\mathcal{P}V_{\Gamma\Omega} + t^{\psi-1}B^{0}\del_{\Gamma}\del_{\Omega}\mathbf{F} + B^{0}\frac{\psi}{t}V_{\Gamma\Omega}.
\end{align}

\subsection{Differentiated Euler Equations}
\label{sec:Differentiated_Euler_deriv}
The Euler equations \eqref{eqn:SH_Euler_Noderivs} are given by
\begin{align}
\label{eqn:Euler_Fuchsian_1}
    \del_{t}W + \mathbf{C}^{\Sigma}\del_{\Sigma}W = -\frac{\mu}{t}\Pi W + t^{2\mu-1}\mathbf{S} +  \mathbf{H}.
\end{align}
Applying the projection operator $\Pi$ to the above equation yields
\begin{align}
\label{eqn:Euler_Pi}
    \del_{t}(\Pi W) + \Pi\mathbf{C}^{\Sigma}\del_{\Sigma}W = -\frac{\mu}{t}\Pi W  + \Pi \mathbf{H},
\end{align}
where in deriving this we have used the fact that $\Pi\mathbf{S}=0$. Equivalently, we can express this equation as
\begin{align*}
    \del_{t}(t^{\mu}\Pi W) + t^{\mu}\Pi\mathbf{C}^{\Sigma}\del_{\Sigma}W =   t^{\mu}\Pi \mathbf{H}.
\end{align*}

Similarly, by applying $\Pi^{\perp}$ to \eqref{eqn:Euler_Fuchsian_1}, we get
\begin{align}
\label{eqn:Euler_Pi_perp}
    \del_{t}(\Pi^{\perp}W) + \Pi^{\perp}\mathbf{C}^{\Sigma}\del_{\Sigma}W = t^{2\mu-1}\Pi^{\perp}\mathbf{S} + \Pi^{\perp}\mathbf{H}.
\end{align}
Adding \eqref{eqn:Euler_Pi} and \eqref{eqn:Euler_Pi_perp} together yields 
\begin{align}
\label{eqn:Euler_Fuchsian_2}
    \del_{t}(\overline{W}) + t^{\mu}\Pi\mathbf{C}^{\Sigma}\del_{\Sigma}W + \Pi^{\perp}\mathbf{C}^{\Sigma}\del_{\Sigma}W = t^{2\mu-1}\Pi^{\perp}\mathbf{S} + \Pi^{\perp}\mathbf{H} + t^{\mu}\Pi \mathbf{H},
\end{align}
where we have set
\begin{align*}
    \overline{W} := \Pi^{\perp}W + t^{\mu}\Pi W = (\zeta, \; \uh, \; t^{\mu}w_{P})^{T}.
\end{align*}

\subsubsection{First Derivative Sub-System}
Differentiating \eqref{eqn:Euler_Fuchsian_1} spatially gives
\begin{align}
\label{eqn:Euler_fuchsian_derivative}
    \del_{t}\del_{\Omega}W + \del_{\Omega}\mathbf{C}^{\Sigma}\del_{\Sigma}W +\mathbf{C}^{\Sigma}\del_{\Omega}\del_{\Sigma}W  = -\frac{\mu}{t}\Pi \del_{\Omega}W  + t^{2\mu-1}\del_{\Omega}\mathbf{S} +  \del_{\Omega}\mathbf{H}.
\end{align}
After projecting with $\Pi$ and multiplying by $t^{\mu}$, we have
\begin{align}
\label{eqn:Euler_fuchsian_derivative2}
    \del_{t}(t^{\mu}\Pi\del_{\Omega}W) + t^{\mu}\Pi\del_{\Omega}\mathbf{C}^{\Sigma}\del_{\Sigma}W +t^{\mu}\Pi\mathbf{C}^{\Sigma}\del_{\Omega}\del_{\Sigma}W  =  t^{\mu}\Pi\del_{\Omega}\mathbf{H}.
\end{align}
Multiplying \eqref{eqn:Euler_fuchsian_derivative} by $\Pi^{\perp}$ yields
\begin{align}
\label{eqn:Euler_fuchsian_derivative3}
\del_{t}(\Pi^{\perp}\del_{\Omega}W) + \Pi^{\perp}\del_{\Omega}\mathbf{C}^{\Sigma}\del_{\Sigma}W +\Pi^{\perp}\mathbf{C}^{\Sigma}\del_{\Omega}\del_{\Sigma}W  =   \Pi^{\perp}t^{2\mu-1}\del_{\Omega}\mathbf{S} +  \Pi^{\perp}\del_{\Omega}\mathbf{H}.
\end{align}
To proceed, we define 
\begin{align*}
    \overline{W}_{\Sigma} &:= \Pi^{\perp}\del_{\Sigma}W + t^{\mu}\Pi \del_{\Sigma}W = (\del_{\Sigma}\zeta, \; \del_{\Sigma}\uh, \; t^{\mu}\del_{\Sigma}w_{P})^{T},\\
    \overline{W}_{\Omega\Sigma}  &:= \ (t^{\mu}\del_{\Omega}\del_{\Sigma}\zeta, \; t^{\mu}\del_{\Omega}\del_{\Sigma}\uh, \; t^{\mu}\del_{\Omega}\del_{\Sigma}w_{P})^{T},
\end{align*}
and
observe that the term $\Pi^{\perp}\mathbf{C}^{\Sigma}\del_{\Omega}\del_{\Sigma}W $ in \eqref{eqn:Euler_fuchsian_derivative3} can be expressed as 
\begin{align*}
\Pi^{\perp}\mathbf{C}^{\Sigma}\del_{\Omega}\del_{\Sigma}W  &= \Pi^{\perp}\mathbf{C}^{\Sigma}\Pi\del_{\Omega}\del_{\Sigma}W + \Pi^{\perp}\mathbf{C}^{\Sigma}\Pi^{\perp}\del_{\Omega}\del_{\Sigma}W\nonumber \\
    &= t^{-\mu}\Pi^{\perp}\mathbf{C}^{\Sigma}\Pi\overline{W}_{\Omega\Sigma} + \Pi^{\perp}\mathbf{C}^{\Sigma}\Pi^{\perp}\del_{\Sigma}\overline{W}_{\Omega}.
\end{align*}
Substituting this into \eqref{eqn:Euler_fuchsian_derivative3} and multiplying the resulting expression by $\Pi^{\perp}C^{0}\Pi^{\perp}$ gives
\begin{align*}
    &\Pi^{\perp}C^{0}\Pi^{\perp}\del_{t}(\Pi^{\perp}\del_{\Omega}W) +  \Pi^{\perp}C^{0}\Pi^{\perp}\mathbf{C}^{\Sigma}\Pi^{\perp}\del_{\Sigma}\overline{W}_{\Omega}^{(\afrak)}  = -t^{-\mu}\Pi^{\perp}C^{0}\Pi^{\perp}\mathbf{C}^{\Sigma}\Pi\overline{W}_{\Omega\Sigma} \nonumber \\
    &-\Pi^{\perp}C^{0}\Pi^{\perp}\del_{\Omega}\mathbf{C}^{\Sigma}\del_{\Sigma}W  
    +t^{2\mu-1}\Pi^{\perp}C^{0}\Pi^{\perp}\del_{\Omega}\mathbf{S} +  \Pi^{\perp}C^{0}\Pi^{\perp}\del_{\Omega}\mathbf{H}.
\end{align*}
Adding this equation to \eqref{eqn:Euler_fuchsian_derivative2}, we obtain
\begin{align*}
    &\Pi^{\perp}C^{0}\Pi^{\perp}\del_{t}(\Pi^{\perp}\del_{\Omega}W) +\del_{t}(t^{\mu}\Pi\del_{\Omega}W) +  \Pi^{\perp}C^{0}\Pi^{\perp}\mathbf{C}^{\Sigma}\Pi^{\perp}\del_{\Sigma}\overline{W}_{\Omega}^{(\afrak)}    \nonumber \\
    &= - t^{\mu}\Pi\del_{\Omega}\mathbf{C}^{\Sigma}\del_{\Sigma}W -t^{\mu}\Pi\mathbf{C}^{\Sigma}\del_{\Omega}\del_{\Sigma}W -t^{-\mu}\Pi^{\perp}C^{0}\Pi^{\perp}\mathbf{C}^{\Sigma}\Pi\overline{W}_{\Omega\Sigma} -\Pi^{\perp}C^{0}\Pi^{\perp}\del_{\Omega}\mathbf{C}^{\Sigma}\del_{\Sigma}W \nonumber \\
    &+t^{2\mu-1}\Pi^{\perp}C^{0}\Pi^{\perp}\del_{\Omega}\mathbf{S} +  \Pi^{\perp}C^{0}\Pi^{\perp}\del_{\Omega}\mathbf{H} +t^{\mu}\Pi\del_{\Omega}\mathbf{H}.
\end{align*}
Recalling the definitions of $\overline{W}$, $\overline{W}_{\Sigma}$, and $\overline{W}_{\Omega\Sigma}$, it is straightforward to express the above equation as 
\begin{align}
\label{eqn:Fuchsian_1st_deriv}
    \mathcal{M}^{0}\del_{t}\overline{W}_{\Omega}^{(\afrak)} + \mathcal{M}^{\Sigma}\del_{\Sigma}\overline{W}_{\Omega}^{(\afrak)} &= -\Pi\del_{\Omega}\mathbf{C}^{\Sigma}\Pi\overline{W}_{\Sigma} -t^{\mu}\Pi\del_{\Omega}\mathbf{C}^{\Sigma}\Pi^{\perp}\overline{W}_{\Sigma} -\Pi\mathbf{C}^{\Sigma}\overline{W}_{\Omega\Sigma} \nonumber \\
    &- t^{-\mu}\Pi^{\perp}C^{0}\Pi^{\perp}\mathbf{C}^{\Sigma}\Pi\overline{W}_{\Omega\Sigma} - \Pi^{\perp}C^{0}\Pi^{\perp}\del_{\Omega}\mathbf{C}^{\Sigma}\Pi^{\perp}\overline{W}_{\Sigma} \nonumber \\
    &- t^{-\mu}\Pi^{\perp}C^{0}\Pi^{\perp}\del_{\Omega}\mathbf{C}^{\Sigma}\Pi\overline{W}_{\Sigma} +t^{2\mu-1}\Pi^{\perp}C^{0}\Pi^{\perp}\del_{\Omega}\mathbf{S} \nonumber \\
    &+  \Pi^{\perp}C^{0}\Pi^{\perp}\del_{\Omega}\mathbf{H} +t^{\mu}\Pi\del_{\Omega}\mathbf{H},
\end{align}
where
\begin{align*}
    \mathcal{M}^{0} = \Pi^{\perp}C^{0}\Pi^{\perp} + \Pi, \;\; \mathcal{M}^{\Sigma} = \Pi^{\perp}C^{0}\Pi^{\perp}\mathbf{C}^{\Sigma}\Pi^{\perp}.
\end{align*}
By substituting in the appropriate definitions, $\mathcal{M}^{0}$ and $\mathcal{M}^{\Sigma}$ can be expressed as\footnote{Note, here we have expressed all terms using the variables $\ut$ and $\nuh^{F}$ for brevity. It is straightforward to write this expression in terms of $\uh$ and $\wh^{F}$ using the definitions $\ut = e^{\uh+\frg(t)}$ and $\nuh^{F}=t^{\mu}\wh^{F}$.} 
\begin{align}
\label{eqn:Mcal0_defn}
    \mathcal{M}^{0} &= \begin{pmatrix}\frac{K}{\sqrt{t^{2\mu}+\ut^{2}}} &0 &0 \\
    0 & \frac{t^{2\mu}+(1-K)\ut^{2}}{(t^{2\mu}+\ut^{2})^{\frac{5}{2}}}\ut^{2} & 0 \\ 0 & 0 & \delta^{P}_{Q} \end{pmatrix}, \\
\label{eqn:McalSigma_defn}
    \mathcal{M}^{\Sigma} &= \begin{pmatrix} \frac{-K\ut}{t^{2\mu}+\ut^{2}}\nuh^{F} & \frac{-Kt^{2\mu}\ut}{(t^{2\mu}+\ut^{2})^{2}}\nuh^{F} & 0\\
    \frac{-Kt^{2\mu}\ut}{(t^{2\mu}+\ut^{2})^{2}}\nuh^{F} & \frac{(-1+2K)t^{2\mu}\ut^{3}+(K-1)\ut^{5}}{(t^{2\mu}+\ut^{2})^{3}}\nuh^{F} & 0 \\ 0 & 0 & 0
    \end{pmatrix}\alpha (f^{\Sigma}_{F}+\delta^{\Sigma}_{F}),
\end{align}
respectively. This completes the derivation of the first differentiated Euler equations. 

\subsubsection{Second Derivative Sub-System}
It remains to derive  the second differentiated Euler equations. We start by spatially differentiating \eqref{eqn:Euler_Fuchsian_1} twice to get
\begin{align*}
\del_{t}\del_{\Gamma}\del_{\Omega}W &+ \del_{\Gamma}\del_{\Omega}\mathbf{C}^{\Sigma}\del_{\Sigma}W + \del_{\Omega}\mathbf{C}^{\Sigma}\del_{\Gamma}\del_{\Sigma}W + \del_{\Gamma}\mathbf{C}^{\Sigma}\del_{\Omega}\del_{\Sigma}W + \mathbf{C}^{\Sigma}\del_{\Gamma}\del_{\Omega}\del_{\Sigma}W \nonumber \\
&= -\frac{\mu}{t}\Pi\del_{\Gamma}\del_{\Omega} W  + t^{2\mu-1}\del_{\Gamma}\del_{\Omega}\mathbf{S} +  \del_{\Gamma}\del_{\Omega}\mathbf{H}.
\end{align*}
Then using the variables $\overline{W}_{\Omega}^{(\afrak)}$ and $\overline{W}_{\Omega\Sigma}$, we can write this equation as 
\begin{align*}
    \del_{t}\overline{W}_{\Gamma\Omega} & + \mathbf{C}^{\Sigma}\del_{\Sigma}\overline{W}_{\Gamma\Omega} + \del_{\Gamma}\del_{\Omega}\mathbf{C}^{\Sigma}\Pi\overline{W}_{\Sigma} + t^{\mu}\del_{\Gamma}\del_{\Omega}\mathbf{C}^{\Sigma}\Pi^{\perp}\overline{W}_{\Sigma}+ \del_{\Omega}\mathbf{C}^{\Sigma}\overline{W}_{\Gamma\Sigma} + \del_{\Gamma}\mathbf{C}^{\Sigma}\overline{W}_{\Omega\Sigma} \nonumber \\
&= \frac{\mu}{t}\Pi^{\perp}\overline{W}_{\Gamma\Omega}  + t^{3\mu-1}\del_{\Gamma}\del_{\Omega}\mathbf{S} +  t^{\mu}\del_{\Gamma}\del_{\Omega}\mathbf{H}.
\end{align*}
Multiplying the above by $C^{0}$ \eqref{eqn:A0_final}, we then obtain the following symmetric hyperbolic formulation of the second differentiated Euler equations: 
\begin{align}
\label{eqn:Fuchsian_2nd_deriv}
    C^{0}\del_{t}\overline{W}_{\Gamma\Omega} + C^\Sigma\del_{\Sigma}\overline{W}_{\Gamma\Omega} &=- C^{0}\del_{\Gamma}\del_{\Omega}\mathbf{C}^{\Sigma}\Pi\overline{W}_{\Sigma} - C^{0}t^{\mu}\del_{\Gamma}\del_{\Omega}\mathbf{C}^{\Sigma}\Pi^{\perp}\overline{W}_{\Sigma} \nonumber \\
    &- C^{0}\del_{\Omega}\mathbf{C}^{\Sigma}\overline{W}_{\Gamma\Sigma} - C^{0}\del_{\Gamma}\mathbf{C}^{\Sigma}\overline{W}_{\Omega\Sigma} + \frac{\mu}{t}C^{0}\Pi^{\perp}\overline{W}_{\Gamma\Omega}  \nonumber \\
    &+ t^{3\mu-1}C^{0}\del_{\Gamma}\del_{\Omega}\mathbf{S} +  t^{\mu}C^{0}\del_{\Gamma}\del_{\Omega}\mathbf{H},
\end{align}
We further further note that \eqref{eqn:Euler_Fuchsian_2} can be viewed as a ODE when expressed in terms of $\overline{W}_{\Omega}^{(\afrak)}$:
\begin{align}
\label{eqn:Fuchsian_noderiv}
\del_{t}\overline{W} &+ \Pi\mathbf{C}^{\Sigma}\Pi\overline{W}_{\Sigma}+t^{\mu}\Pi\mathbf{C}^{\Sigma}\Pi^{\perp}\overline{W}_{\Sigma} + t^{-\mu}\Pi^{\perp}\mathbf{C}^{\Sigma}\Pi\overline{W}_{\Sigma} + \Pi^{\perp}\mathbf{C}^{\Sigma}\Pi^{\perp}\overline{W}_{\Sigma} \nonumber \\
&= t^{2\mu-1}\Pi^{\perp}\mathbf{S} + \Pi^{\perp}\mathbf{H} + t^{\mu}\Pi \mathbf{H}.
\end{align}

\section{The Complete Einstein-Euler System\label{sec:complete_system}}
The complete Einstein-Euler system, which we will use to establish global existence, is obtained by combining equations \eqref{eqn:SH_Einstein_NoDerivs}, \eqref{eqn:Einstein_1stderiv}, \eqref{eqn:Einstein_2ndderiv}, \eqref{eqn:Fuchsian_1st_deriv}, \eqref{eqn:Fuchsian_2nd_deriv}, and \eqref{eqn:Fuchsian_noderiv}. Upon reintroducing the fluid indices $(\afrak)$, with $\afrak = 1, 2$, the complete system can be expressed in Fuchsian form as follows:
\begin{align}
\label{eqn:EinsteinEuler_NoDeriv}
    \mathscr{A}^{0}  \del_{t}\begin{pmatrix} V \\ \overline{W}_{(1)} \\ \overline{W}_{(2)} \end{pmatrix} &+ \mathscr{A}^{\Sigma} \del_{\Sigma} \begin{pmatrix} V \\ \overline{W}_{(1)} \\ \overline{W}_{(2)}\end{pmatrix} = \frac{1}{t}\mathscr{D}\mathbb{P}_{1}\begin{pmatrix} V \\ \overline{W}_{(1)} \\ \overline{W}_{(2)}\end{pmatrix} +\mathbf{Q}, \\
\label{eqn:EinsteinEuler_1stDeriv}
    \mathscr{B}^{0} \del_{t} \begin{pmatrix} V_{\Omega} \\ \overline{W}^{(1)}_{\Omega} \\ \overline{W}^{(2)}_{\Omega}\end{pmatrix} &+ \mathscr{B}^{\Sigma} \del_{\Sigma} \begin{pmatrix} V_{\Omega} \\ \overline{W}^{(1)}_{\Omega} \\ \overline{W}^{(2)}_{\Omega}\end{pmatrix}  = \frac{1}{t}\mathscr{D}\mathbb{P}_{1}\begin{pmatrix} V_{\Omega} \\ \overline{W}^{(1)}_{\Omega} \\ \overline{W}^{(2)}_{\Omega}\end{pmatrix} + \mathbf{R}, \\
\label{eqn:EinsteinEuler_2ndDeriv}
   \mathscr{C}^{0} \del_{t} \begin{pmatrix} V_{\Gamma\Omega} \\ \overline{W}^{(1)}_{\Gamma\Omega} \\\overline{W}^{(2)}_{\Gamma\Omega} \end{pmatrix} &+ \mathscr{C}^{\Sigma} \del_{\Sigma}  \begin{pmatrix} V_{\Gamma\Omega} \\ \overline{W}^{(1)}_{\Gamma\Omega}\\\overline{W}^{(2)}_{\Gamma\Omega}  \end{pmatrix} = \frac{1}{t}\mathscr{E}\mathbb{P}_{2}\begin{pmatrix} V_{\Gamma\Omega} \\ \overline{W}^{(1)}_{\Gamma\Omega} \\ \overline{W}^{(2)}_{\Gamma\Omega} \end{pmatrix} + \mathbf{G},
\end{align}
where the matrix coefficients are defined by
\begin{align}
\label{eqn:A_scr_0}
\mathscr{A}^{0} &= \begin{pmatrix} B^{0} & 0  & 0\\ 0 & \mathbbm{1} & 0\\ 0 & 0 & \mathbbm{1} \end{pmatrix}, \\
\label{eqn:A_scr_Sigma}
\mathscr{A}^{\Sigma} &= \begin{pmatrix} -\alpha B^{\Sigma} & 0 & 0\\ 0 & 0 & 0 \\ 0&0&0\end{pmatrix}, \\
\label{eqn:D_scr}
\mathscr{D} &= \begin{pmatrix} \alpha \mathcal{B} & 0 & 0\\ 0 & \mathbbm{1} & 0 \\ 0 & 0 & \mathbbm{1}\end{pmatrix} , \\
\label{eqn:B_scr_0}
\mathscr{B}^{0} &= \begin{pmatrix} B^{0} & 0 & 0\\ 0 & \mathcal{M}_{(1)}^{0} & 0 \\0 & 0 & \mathcal{M}_{(2)}^{0}\end{pmatrix}, \\
\label{eqn:B_scr_Sigma}
\mathscr{B}^{\Sigma} &=\begin{pmatrix}  -\alpha B^{\Sigma} & 0 & 0 \\ 0 & \mathcal{M}^{\Sigma}_{(1)} & 0 \\ 0 & 0 & \mathcal{M}^{\Sigma}_{(2)}\end{pmatrix}, \\
\label{eqn:C_scr_0}
\mathscr{C}^{0} &=\begin{pmatrix} B^{0} & 0 & 0\\ 0 & C^{0}_{(1)} & 0 \\ 0 &0 &C^{0}_{(2)}\end{pmatrix} , \\
\label{eqn:C_scr_Sigma}
\mathscr{C}^{\Sigma} &=\begin{pmatrix}  -\alpha B^{\Sigma} & 0 & 0\\ 0 &  C_{(1)}^{\Sigma} & 0 \\ 0 & 0 &  C_{(2)}^{\Sigma} \end{pmatrix}, \\
\label{eqn:E_scr}
\mathscr{E} &= \begin{pmatrix} 2\mathcal{P} + \psi B^{0} & 0 & 0\\ 0 & \mu_{(1)} C^{0}_{(1)} & 0 \\ 0 & 0 & \mu_{(2)} C^{0}_{(2)} \end{pmatrix},
\end{align}
and the solution vectors are given by
\begin{equation}
\begin{aligned}
\label{eqn:Fuchsian_Solution_Vector_Defns}
    V &= (\Hct, \alphat, \dot{U}_{P}, A_{P}, N_{PI}, \Sigma_{PI}, f_{P}^{\Sigma})^{\text{T}}, \\
    V_{\Sigma} &= \del_{\Sigma}V, \\
    V_{\Omega\Sigma} &= t^{\psi}\del_{\Omega}\del_{\Sigma}V, \\
    \overline{W}_{(\afrak)} &= (\zeta_{(\afrak)}, \; \uh_{(\afrak)}, \; t^{\mu_{(\afrak)}}w^{(\afrak)}_{P})^{T}, \\
    \overline{W}^{(\afrak)}_{\Sigma} &= (\del_{\Sigma}\zeta_{(\afrak)}, \; \del_{\Sigma}\uh_{(\afrak)}, \; t^{\mu_{(\afrak)}}\del_{\Sigma}w^{(\afrak)}_{P})^{T}, \\
    \overline{W}^{(\afrak)}_{\Omega\Sigma} &= (t^{\mu_{(\afrak)}}\del_{\Omega}\del_{\Sigma}\zeta_{(\afrak)}, \; t^{\mu_{(\afrak)}}\del_{\Omega}\del_{\Sigma}\uh_{(\afrak)}, \; t^{\mu_{(\afrak)}}\del_{\Omega}\del_{\Sigma}w^{(\afrak)}_{P})^{T}.
\end{aligned}
\end{equation}
In the above, the matrices $ \mathbb{P}_{1}$, $\mathbb{P}_{1}^{\perp}$,  $\mathbb{P}_{1}$ and $\mathbb{P}_{1}^{\perp}$ are  
projection operators defined by
\begin{gather}
\label{eqn:Pbb1_Def}
   \mathbb{P}_{1} =  \begin{pmatrix}\mathcal{P} & 0 & 0 \\ 0 & 0 & 0 \\ 0 & 0 & 0 \end{pmatrix}, \quad \mathbb{P}_{1}^{\perp} = \mathbbm{1} - \mathbb{P}_{1},\\ 
\label{eqn:Pbb2_Def}
   \mathbb{P}_{2} = \begin{pmatrix}\mathbbm{1} & 0 & 0 \\ 0 & \Pi^{\perp} & 0 \\ 0 & 0 & \Pi^{\perp}  \end{pmatrix}, \quad \mathbb{P}_{2}^{\perp} = \mathbbm{1} - \mathbb{P}_{2},
\end{gather}
while the source terms $\mathbf{Q}$, $\mathbf{R}$, and $\mathbf{G}$ are defined by
\begin{align}
\label{eqn:Qbf_Source_Def}
   \mathbf{Q} &= \bigl(\mathbf{Q}_{0},\mathbf{Q}_{(1)},\mathbf{Q}_{(2)}\bigr)^T, \\
\label{eqn:Rbf_Source_Def}
   \mathbf{R} &= \bigl(\mathbf{R}_{0}, \mathbf{R}_{(1)}, \mathbf{R}_{(2)}\bigr)^T, \\
\label{eqn:Gbf_Source_Def}
   \mathbf{G} &= \bigl(\mathbf{G}_{0} , \mathbf{G}_{(1)} , \mathbf{G}_{(2)}\bigr)^T, 
\end{align}
where 
\begin{align}
\label{eqn:Qbf_0_defn}
\mathbf{Q}_{0} &= \alpha\mathcal{F}, \\
\label{eqn:Qbf_a_defn}
    \mathbf{Q}_{(\afrak)} &= - \Pi\mathbf{C}_{(\afrak)}^{\Sigma}\Pi\overline{W}^{(\afrak)}_{\Sigma}-t^{\mu_{(\afrak)}}\Pi\mathbf{C}_{(\afrak)}^{\Sigma}\Pi^{\perp}\overline{W}^{(\afrak)}_{\Sigma} - t^{-\mu_{(\afrak)}}\Pi^{\perp}\mathbf{C}_{(\afrak)}^{\Sigma}\Pi\overline{W}^{(\afrak)}_{\Sigma} \nonumber \\ &- \Pi^{\perp}\mathbf{C}_{(\afrak)}^{\Sigma}\Pi^{\perp}\overline{W}^{(\afrak)}_{\Sigma} + t^{2\mu_{(\afrak)}-1}\Pi^{\perp}\mathbf{S}_{(\afrak)} + \Pi^{\perp}\mathbf{H}_{(\afrak)} + t^{\mu_{(\afrak)}}\Pi \mathbf{H}_{(\afrak)}, \\
\label{eqn:Rbf_0_defn}
    \mathbf{R}_{0} &= B^{0}\del_{\Omega}\mathbf{B}^{\Sigma}V_{\Sigma} + B^{0}\del_{\Omega}\mathbf{F}, \\
\label{eqn:Rbf_a_defn}
    \mathbf{R}_{(\afrak)} &= -\Pi\del_{\Omega}\mathbf{C}_{(\afrak)}^{\Sigma}\Pi\overline{W}^{(\afrak)}_{\Sigma} -t^{\mu_{(\afrak)}}\Pi\del_{\Omega}\mathbf{C}_{(\afrak)}^{\Sigma}\Pi^{\perp}\overline{W}^{(\afrak)}_{\Sigma} -\Pi\mathbf{C}_{(\afrak)}^{\Sigma}\overline{W}^{(\afrak)}_{\Omega\Sigma} \nonumber \\
    &- t^{-\mu_{(\afrak)}}\Pi^{\perp}C_{(\afrak)}^{0}\Pi^{\perp}\mathbf{C}_{(\afrak)}^{\Sigma}\Pi\overline{W}^{(\afrak)}_{\Omega\Sigma} - \Pi^{\perp}C_{(\afrak)}^{0}\Pi^{\perp}\del_{\Omega}\mathbf{C}_{(\afrak)}^{\Sigma}\Pi^{\perp}\overline{W}^{(\afrak)}_{\Sigma} \nonumber \\
    &- t^{-\mu_{(\afrak)}}\Pi^{\perp}C_{(\afrak)}^{0}\Pi^{\perp}\del_{\Omega}\mathbf{C}_{(\afrak)}^{\Sigma}\Pi\overline{W}^{(\afrak)}_{\Sigma} +t^{2\mu_{(\afrak)}-1}\Pi^{\perp}C_{(\afrak)}^{0}\Pi^{\perp}\del_{\Omega}\mathbf{S}_{(\afrak)} 
    \nonumber \\
    &+  \Pi^{\perp}C_{(\afrak)}^{0}\Pi^{\perp}\del_{\Omega}\mathbf{H}_{(\afrak)} +t^{\mu_{(\afrak)}}\Pi\del_{\Omega}\mathbf{H}_{(\afrak)},\\
\label{eqn:Gbf_0_defn}
    \mathbf{G}_{0} &=  B^{0}\del_{\Gamma}\mathbf{B}^{\Sigma}V_{\Omega\Sigma} + B^{0}\del_{\Omega}\mathbf{B}^{\Sigma}V_{\Gamma\Sigma} + t^{\psi}B^{0}\del_{\Gamma}\del_{\Omega}\mathbf{B}^{\Sigma}V_{\Sigma} + t^{\psi}B^{0}\del_{\Gamma}\del_{\Omega}\mathbf{F}, \\
\label{eqn:Gbf_a_defn}
    \mathbf{G}_{(\afrak)} &=- C^{0}_{(\afrak)}\del_{\Gamma}\del_{\Omega}\mathbf{C}_{(\afrak)}^{\Sigma}\Pi\overline{W}^{(\afrak)}_{\Sigma} - C^{0}_{(\afrak)}t^{\mu_{(\afrak)}}\del_{\Gamma}\del_{\Omega}\mathbf{C}_{(\afrak)}^{\Sigma}\Pi^{\perp}\overline{W}^{(\afrak)}_{\Sigma}- C^{0}_{(\afrak)}\del_{\Omega}\mathbf{C}_{(\afrak)}^{\Sigma}\overline{W}^{(\afrak)}_{\Gamma\Sigma} \nonumber \\
    &- C^{0}_{(\afrak)}\del_{\Gamma}\mathbf{C}_{(\afrak)}^{\Sigma}\overline{W}^{(\afrak)}_{\Omega\Sigma}  + t^{3\mu_{(\afrak)}-1}C_{(\afrak)}^{0}\del_{\Gamma}\del_{\Omega}\mathbf{S}_{(\afrak)} +  t^{\mu_{(\afrak)}}C^{0}_{(\afrak)}\del_{\Gamma}\del_{\Omega}\mathbf{H}_{(\afrak)}.
\end{align}

For subsequent arguments, we find it advantageous to decompose $\mathbf{Q}$ as  
\begin{align}
\label{eqn:Qbfsource_Background_Error_Split}
    \mathbf{Q} &= \mathbf{Q}_{H} + \mathbf{Q}_{E},
\intertext{where}
\label{eqn:Qbf_H_defn}
    \mathbf{Q}_{H} &= \bigl(\bar{\mathbf{Q}},  \mathbf{Q}_{(1)}, \mathbf{Q}_{(2)}\bigr)^T,
    \intertext{and}
\label{eqn:Qbf_E_defn}
\mathbf{Q}_{E} &= \bigl(\hat{\mathbf{Q}} , 0 , 0\bigr).
\end{align}
Here, $\mathbf{Q}_{H}$ denotes the terms that vanish on the homogeneous background $(V,\overline{W}_{(\afrak)})=(0,0)$,  where $\hat{\mathbf{Q}}$ and $\bar{\mathbf{Q}}$ are defined by
\begin{align}
\label{eqn:Qbf_hat_defn}
    \hat{\mathbf{Q}} &= \mathbf{Q}_{0}|_{(V,\overline{W}_{(\afrak)})=(0,0)} = (\alpha\mathcal{F})|_{(V,\overline{W}_{(\afrak)})=(0,0)},
    \intertext{and}
\label{eqn:Qbf_bar_defn}
    \bar{\mathbf{Q}} &= \mathbf{Q}_{0} - \hat{\mathbf{Q}},
\end{align}
respectively.
Observe that $\bar{\mathbf{Q}}$ vanishes on the homogeneous background defined by $(V,\overline{W}_{(\afrak)})=(0,0)$. It is clear from  \eqref{eqn:Fcal_components_Final}  that the only terms of $\mathcal{F}$ that do not vanish on the homogeneous background are those that involve fluid stress-energy terms. Thus, $\hat{\mathbf{Q}}$ only contains terms involving the fluid variables and can be decomposed as follows
\begin{align*}
 \hat{\mathbf{Q}} = \hat{\mathbf{Q}}^{(1)} +\hat{\mathbf{Q}}^{(2)}.
\end{align*} 
The components of $\hat{\mathbf{Q}}^{(\afrak)}$ are given by
\begin{align}
\label{eqn:Qbf_hat_components}
\hat{\mathbf{Q}}^{(\afrak)}_{1} &= \sqrt{\frac{3}{2\Lambda}}\Bigg[-\frac{1}{3}(\frac{a}{2}+c+6c\theta)\frac{t^{2}(t^{2\mu_{(\afrak)}}+(1+K_{(\afrak)})e^{2\frg_{(\afrak)}})\rho_{0}}{(t^{2\mu_{(\afrak)}}+e^{2\frg_{(\afrak)}})^{\frac{K_{(\afrak)}+1}{2}}} \nonumber  \\
&- \frac{1}{3}(\frac{a}{2}+c)\Big(\frac{t^{2}(1+K_{(\afrak)})e^{2\frg_{(\afrak)}}\rho_{0}}{(t^{2\mu_{(\afrak)}}+e^{2\frg_{(\afrak)}})^{\frac{K_{(\afrak)}+1}{2}}}\xi^{(\afrak)}_{A}\xi_{(\afrak)}^{A} + \frac{3K_{(\afrak)}\rho_{0}t^{\frac{4K_{(\afrak)}}{1-K_{(\afrak)}}}}{(t^{2\mu_{(\afrak)}}+e^{2\frg_{(\afrak)}})^{\frac{K_{(\afrak)}+1}{2}}}\Big)\Bigg], \nonumber \\
\hat{\mathbf{Q}}^{(\afrak)}_{2} &= \sqrt{\frac{3}{2\Lambda}}\Bigg[-\frac{1}{3}(\frac{c}{2}+b+6b\theta)\frac{t^{2}(t^{2\mu_{(\afrak)}}+(1+K_{(\afrak)})e^{2\frg_{(\afrak)}})\rho_{0}}{(t^{2\mu_{(\afrak)}}+e^{2\frg_{(\afrak)}})^{\frac{K_{(\afrak)}+1}{2}}} \nonumber \\
&- \frac{1}{3}(\frac{c}{2}+b)\Big(\frac{t^{2}(1+K_{(\afrak)})e^{2\frg_{(\afrak)}}\rho_{0}}{(t^{2\mu_{(\afrak)}}+e^{2\frg_{(\afrak)}})^{\frac{K_{(\afrak)}+1}{2}}}\xi^{(\afrak)}_{A}\xi_{(\afrak)}^{A} + \frac{3K_{(\afrak)}\rho_{0}t^{\frac{4K_{(\afrak)}}{1-K_{(\afrak)}}}}{(t^{2\mu_{(\afrak)}}+e^{2\frg_{(\afrak)}})^{\frac{K_{(\afrak)}+1}{2}}}\Big)\Bigg], \\
\hat{\mathbf{Q}}^{(\afrak)}_{3} &= \sqrt{\frac{3}{2\Lambda}}\Bigg[-(d\lambda+g\gamma)\frac{t^{2}(1+K_{(\afrak)})e^{\frg_{(\afrak)}}\rho_{0})}{(t^{2\mu_{(\afrak)}}+e^{2\frg_{(\afrak)}})^{\frac{K_{(\afrak)}}{2}}}\xi^{(\afrak)}_{A}\Bigg], \nonumber \\
\hat{\mathbf{Q}}^{(\afrak)}_{4} &= \sqrt{\frac{3}{2\Lambda}}\Bigg[-(g\lambda+\gamma)\frac{t^{2}(1+K_{(\afrak)})e^{\frg_{(\afrak)}}\rho_{0})}{(t^{2\mu_{(\afrak)}}+e^{2\frg_{(\afrak)}})^{\frac{K_{(\afrak)}}{2}}}\xi^{(\afrak)}_{A}\Bigg], \nonumber \\
\hat{\mathbf{Q}}^{(\afrak)}_{5} &= 0, \nonumber \\
\hat{\mathbf{Q}}^{(\afrak)}_{6} &= \sqrt{\frac{3}{2\Lambda}}\Bigg[\frac{t^{2}(1+K_{(\afrak)})e^{2\frg_{(\afrak)}}\rho_{0}}{(t^{2\mu_{(\afrak)}}+e^{2\frg_{(\afrak)}})^{\frac{K_{(\afrak)}+1}{2}}}\xi^{(\afrak)}_{A}\xi^{(\afrak)}_{B} - \frac{t^{2}(1+K_{(\afrak)})e^{2\frg_{(\afrak)}}\rho_{0}}{3(t^{2\mu_{(\afrak)}}+e^{2\frg_{(\afrak)}})^{\frac{K_{(\afrak)}+1}{2}}}\xi^{(\afrak)}_{C}\xi_{(\afrak)}^{C}\eta_{AB}\Bigg], \nonumber \\
\hat{\mathbf{Q}}^{(\afrak)}_{7} &= 0, \nonumber
\end{align}
where in deriving these we have used the definitions of $\alpha$ \eqref{eqn:lapse_Hct_identity} and $\mathcal{F}$ \eqref{eqn:Fcal_components_Final} and the decomposition of the stress-energy tensor in terms of the variables $\overline{W}_{(\afrak)}$ from Appendix  \eqref{eqn:Tab_Decomp_Wbar_Vars_app}.

\section{Coefficient Analysis}
In order to apply the Fuchsian global existence theory from \cite{BOOS:2021} to the system \eqref{eqn:EinsteinEuler_NoDeriv}-\eqref{eqn:EinsteinEuler_2ndDeriv}, we
need to verify that its coefficients, e.g. $\mathscr{A}^{0}$, $\mathscr{A}^{\Sigma}$, $\mathscr{D}$, $\mathbf{Q}$, etc.,  satisfy the conditions from \cite{BOOS:2021}*{\S 3.4}. Once that is done, we can use Theorem 3.8 from  \cite{BOOS:2021}*{\S 3.4} to obtain a global existence and stability result.  

\subsection{Verifying Coefficient Assumptions}
\label{sec:Coefficient_Assumptions}

\subsubsection{Undifferentiated Sub-System}
\label{sec:CoefficientAnalysis_Undifferentiated}
We begin the verification of the coefficient conditions by first considering the undifferentiated sub-system \eqref{eqn:EinsteinEuler_NoDeriv}.  From \eqref{eqn:Bcal_Einstein} and \eqref{eqn:P_projector_def}, it is simple to verify that
\begin{align*}
    [\mathcal{P},\mathcal{B}] = 0.
\end{align*}
Using this, with \eqref{eqn:D_scr} and \eqref{eqn:Pbb1_Def}, we see that $\mathscr{D}$  satisfies 
\begin{align*}
    [\mathscr{D},\mathbb{P}_{1}] = 0.
\end{align*}

Next, by \eqref{eqn:lapse_Hct_identity}, \eqref{eqn:Bcal_Einstein}-\eqref{eqn:Pcal_perp_defn}, and \eqref{eqn:D_scr}, we notice that $\mathscr{D}$ is a matrix-valued map that depends on the gravitational variables \eqref{eqn:Fuchsian_Solution_Vector_Defns} as follows
\begin{align*}
    \mathscr{D} = \mathscr{D}(t,\mathcal{P}^{\perp}V)
\end{align*}
where for any fixed $R_0>0$ there exists a positive constant $r>0$ such that, for any $R\in (0,R_0]$, $\mathscr{D}$ is smooth on the domain 
\begin{align}
\label{eqn:Domain_Dscr}
    (-r,2) \times B_{R}(\mathbb{R}^{26})
\end{align}
and satisfies
\begin{align*}
    \mathscr{D} \gtrsim\mathbbm{1}.
\end{align*}
From \eqref{eqn:lapse_Hct_identity}, \eqref{eqn:pi_projection_tracefree}, \eqref{eqn:Scal_operator_defn}, \eqref{eqn:BSigma_Einstein}, \eqref{eqn:BSigma_Components}, \eqref{eqn:Symmetrisation_parameter_values}, \eqref{eqn:P_projector_def}-\eqref{eqn:Pcal_perp_defn}, \eqref{eqn:A_scr_0}-\eqref{eqn:A_scr_Sigma}, and \eqref{eqn:Fuchsian_Solution_Vector_Defns}, it is also clear the matrices
\begin{align*}
    \mathscr{A}^{\Sigma} = \mathscr{A}^{\Sigma}(t,\mathcal{P}^{\perp}V)
\end{align*}
are smooth and symmetric  on the domain \eqref{eqn:Domain_Dscr} and $\mathscr{A}^{0}$ is a constant, symmetric matrix satisfying
\begin{align*}
    \mathscr{A}^{0} \gtrsim \mathbbm{1}.
\end{align*} 
Additionally, from \eqref{eqn:B0_Einstein}, \eqref{eqn:P_projector_def}-\eqref{eqn:Pcal_perp_defn}, \eqref{eqn:A_scr_0}, and \eqref{eqn:Pbb1_Def}, it is straightforward to verify that
\begin{align*}
\mathbb{P}_{1}^{\perp}\mathscr{A}^{0}\mathbb{P}_{1} = \mathbb{P}_{1}\mathscr{A}^{0}\mathbb{P}_{1}^{\perp} = 0.
\end{align*}

Clearly, $\del_{t}(\mathscr{A}^{0}) = 0$
and differentiating $\mathscr{A}^{\Sigma}$ spatially yields 
\begin{align*}
    \del_{\Sigma}\big(\mathscr{A}^{\Sigma}(t,\mathcal{P}^{\perp}V)\big) = D_{2}\mathscr{A}^{\Sigma}(t,\mathcal{P}^{\perp}V)\cdot \del_{\Sigma}(\mathcal{P}^{\perp}V).
\end{align*}
Using these,  we find that the matrix divergence, as defined in \cite{BOOS:2021}*{\S 3.1}, is given by
\begin{align}
\label{eqn:Div_scrA}
\text{Div}\mathscr{A}(t,V,V_\Sigma) = D_{2}\mathscr{A}^{\Sigma}(t,\mathcal{P}^{\perp}V)\cdot \mathcal{P}^{\perp}V_\Sigma.
\end{align}
Moreover, from the smoothness of the matrices $\mathcal{A}^\Sigma$ and \eqref{eqn:Div_scrA}, we deduce  
\begin{align*}
  \bigl|\text{Div}\mathscr{A}\bigr|
  \lesssim 1
\end{align*}
for $(t,V,V_\Sigma)\in (0,1]\times B_R(\mathbb{R}^{35})\times B_R(\mathbb{R}^{35\times 3})$. \newline \par

Now, let us consider the source term $\mathbf{Q}_{E}$ given by \eqref{eqn:Qbf_E_defn}. 
Recalling that $\mu_{(\afrak)}>0$ and $4 K_{(\afrak)}/(1-K_{(\afrak)})>0$, we observe from \eqref{eqn:Qbf_hat_defn}, \eqref{eqn:Qbf_hat_components},  and Proposition \ref{prop:Asymptotic_ODE} that $\mathbf{Q}_{E} \in C^0\bigl([0,1],\mathbb{R}^{35}\times \mathbb{R}^{5}\times\mathbb{R}^{5}\bigr)$   
and 
\begin{align}
\label{eqn:Qbf_E_Bound}
\bigl|\mathbf{Q}_{E}(t) \bigr| \lesssim t^{2}
\end{align}
for $t\in (0,1]$. \newline \par 

It remains to analyse the source term $\mathbf{Q}_{H}$ given by \eqref{eqn:Qbf_H_defn}. We start by analysing the first component $\bar{\mathbf{Q}}$ defined above by \eqref{eqn:Qbf_bar_defn}. Setting
\begin{align}
\label{eqn:tbar_ttilde_defns}
    \bar{t}_{(\afrak)} = t^{\mu_{(\afrak)}}, \;\; \tilde{t}_{(\afrak)}  = t^{\frac{4K_{(\afrak)} }{1-K_{(\afrak)}}},
\end{align} 
we see from   \eqref{eqn:Hct_Alphat_defs}, \eqref{eqn:Hct_identities}-\eqref{eqn:lapse_Hct_identity}, \eqref{eqn:pi_projection_tracefree}, \eqref{eqn:Scal_operator_defn}, \eqref{eqn:Symmetrisation_parameter_values}, \eqref{eqn:f_modified_frame_defn}, \eqref{eqn:Fcal_components_Final}, \eqref{eqn:Fuchsian_Solution_Vector_Defns}, \eqref{eqn:uh_w_defs_app}-\eqref{eqn:nuh_w_identity_app}, 
and \eqref{eqn:Tab_Decomp_Wbar_Vars_app} that $\bar{\mathbf{Q}}$ can be viewed as a map 
\begin{align*}
    \bar{\mathbf{Q}} = \bar{\mathbf{Q}}\bigl(t,\bar{t}_{(\afrak)},\tilde{t}_{(\afrak)},g_{(\mathfrak{a})},V,\overline{W}_{(\afrak)}\bigr),
\end{align*}
which is smooth on the domain 
\begin{align*}
(-r,2) \times \prod_{\afrak = 1}^{2}\Bigl(-r,2^{\mu_{(\afrak)}}\Bigr) \times \prod_{\afrak = 1}^{2}\biggl(-r,2^{\frac{4K_{(\afrak)}}{1-K_{(\afrak)}}}\biggr)  
\times B_{R}(\mathbb{R}^{2}) \times B_{R}(\mathbb{R}^{35}) \times B_{R}(\mathbb{R}^{10})
\end{align*}
provided $R_0>0$ is chosen sufficiently small and $R\in (0,R_0]$. {Moreover, $\bar{\mathbf{Q}}$ vanishes for $(V, \overline{W}_{(\afrak)})=0$.
\begin{rem}
In analysis that follows, it will always be possible to choose $R_0>0$ small enough so that the stated smoothness properties of the maps consider below are valid for any choice of $R\in (0,R_0]$ and a $r>0$ that remains fixed throughout.
\end{rem}

Next, we consider source terms $\mathbf{Q}_{(\afrak)}$ defined by \eqref{eqn:Qbf_a_defn}. Using \eqref{eqn:uh_w_defs_app}-\eqref{eqn:nuh_w_identity_app}, it follows from the definitions  \eqref{eqn:Source_ttt_defns},  \eqref{eqn:Pi_def},  \eqref{eqn:Pi_perp_defn},  \eqref{eqn:Abf_def},  \eqref{eqn:Hbf_def}, and \eqref{eqn:Sbf_def} that the maps 
\begin{equation*}
\Pi^{\perp}C_{(\afrak)}^{0}\Pi^{\perp}\mathbf{S}_{(\afrak)}
\end{equation*}
and
\begin{equation*}
\mathbf{Q}_{(\afrak)}-t^{2\mu_{(\afrak)}-1}\Pi^{\perp}C_{(\afrak)}^{0}\Pi^{\perp}\mathbf{S}_{(\afrak)}
\end{equation*} 
are smooth in the the variables $(t,\bar{t}_{(\afrak)} ,\frg_{(\afrak)},V,\overline{W}_{(\afrak)})$ and $(t,\bar{t}_{(\afrak)} ,\frg_{(\afrak)},V,\overline{W}_{(\afrak)},\overline{W}^{(\afrak)}_{\Sigma})$ on the domains
\begin{align*}
(-r,2)  \times \prod_{\afrak = 1}^{2}\bigl(-r,2^{\mu_{(\afrak)}}\bigr) \times B_{R}(\mathbb{R}^{2}) 
\times 
B_{R}(\mathbb{R}^{35}) \times B_{R}(\mathbb{R}^{10}),
\end{align*}
and
\begin{align*}
(-r,2)  \times \prod_{\afrak = 1}^{2}\bigl(-r,2^{\mu_{(\afrak)}}\bigr) \times B_{R}(\mathbb{R}^{2}) 
\times B_{R}(\mathbb{R}^{35}) \times B_{R}(\mathbb{R}^{10}) \times B_{R}(\mathbb{R}^{10 \times 3}),
\end{align*}
respectively. Furthermore, both maps vanish for $\bigl(V, \overline{W}_{(\afrak)}, V_{\Sigma},\overline{W}_{\Sigma}^{(\afrak)}\bigr)=0$.
Hence, from \eqref{eqn:Mu1<Mu2} and the bounds on $\frg_{(\afrak)}$ established in Proposition \ref{prop:Asymptotic_ODE}, we see that $\mathbf{Q}_{H}$  satisfies 
\begin{align}
\label{eqn:Qbf_Bound}
\bigl|\mathbf{Q}_{H}\bigl(t,V,\overline{W}_{(\afrak)}, \overline{W}^{(\afrak)}_{\Sigma}\bigr)\bigr| \lesssim \bigl(t^{2\mu_{(1)}-1} + 1\bigr)\bigr| (V,\overline{W}_{(\afrak)},\overline{W}^{(\afrak)}_{\Sigma})\bigr|
\end{align}
for $\bigl(t,V,\overline{W}_{(\afrak)},\overline{W}^{(\afrak)}_{\Sigma}\bigr)\in (0,1]\times B_{R}(\mathbb{R}^{35}) \times B_{R}(\mathbb{R}^{10})\times B_{R}(\mathbb{R}^{10\times 3})$.

\subsubsection{First Derivative Sub-System}
\label{sec:CoefficientAnalysis_1stDeriv}
Next, we analyse the coefficients of the differentiated sub-system \eqref{eqn:EinsteinEuler_1stDeriv}. From  \eqref{eqn:B0_Einstein}, \eqref{eqn:Symmetrisation_parameter_values},\eqref{eqn:Mcal0_defn}, \eqref{eqn:tbar_ttilde_defns}, and \eqref{eqn:uh_w_defs_app}, it is clear that $\mathscr{B}^{0}$  can be viewed as matrix-valued map 
\begin{align*}
\mathscr{B}^{0} = \mathscr{B}^{0}\bigl(\bar{t}_{(\afrak)}, \frg_{(\afrak)},\overline{W}_{(\afrak)}\bigr)
\end{align*}
that is smooth and symmetric on the domain 
\begin{align*}
 \prod_{\afrak = 1}^{2}\bigl(-r,2^{\mu_{(\afrak)}}\bigr) \times B_{R}(\mathbb{R}^{2}) \times B_{R}(\mathbb{R}^{10})
\end{align*} 
and satisfies
\begin{align*}
    \mathscr{B}^{0}\gtrsim \mathbbm{1}.
\end{align*} 
Additionally, by \eqref{eqn:B0_Einstein}, \eqref{eqn:P_projector_def}-\eqref{eqn:Pcal_perp_defn}, \eqref{eqn:B_scr_0} and \eqref{eqn:Pbb1_Def}, we find that $\mathscr{B}^{0}$ satisfies
\begin{equation*}
\begin{aligned}
    \mathbb{P}_{1}\mathscr{B}^{0}\mathbb{P}_{1}^{\perp} &= \mathbb{P}_{1}^{\perp}\mathscr{B}^{0}\mathbb{P}_{1} = 0. \\
\end{aligned}
\end{equation*}
Furthermore, formally taking the time derivative of $\mathscr{B}^{0}$, we have
\begin{align}   
\label{eqn:Bscr0_time_derivative}
\del_{t}\big(\mathscr{B}^{0}\bigl(t^{\mu_{(\afrak)}},\frg_{(\afrak)}(t),\overline{W}_{(\afrak)}\bigr)\big) 
&= D\mathscr{B}^{0}\bigl(t^{\mu_{(\afrak)}},\frg_{(\afrak)}(t),\overline{W}_{(\afrak)}\bigr)\cdot\begin{pmatrix} \mu_{(\afrak)} t^{\mu_{(\afrak)}-1} \\ \frg^{\prime}_{(\afrak)}(t) \\ \del_{t}\overline{W}_{(\afrak)} \end{pmatrix}
\end{align}
where $\del_{t}\overline{W}_{(\afrak)}$ is determined using the evolution equation \eqref{eqn:EinsteinEuler_NoDeriv}, that is,
\begin{align}
\label{eqn:time_deriv_Wbar_subsystem2_analysis}
    \del_{t}\overline{W}_{(\afrak)} &= \mathbf{Q}_{(\afrak)},
\end{align}
where in deriving this we have used \eqref{eqn:A_scr_0}-\eqref{eqn:D_scr}, \eqref{eqn:Pbb1_Def} and \eqref{eqn:Qbf_Source_Def}. \newline \par

From \eqref{eqn:Hct_Alphat_defs}, \eqref{eqn:Hct_identities}-\eqref{eqn:lapse_Hct_identity},
\eqref{eqn:pi_projection_tracefree}, \eqref{eqn:Scal_operator_defn}, \eqref{eqn:BSigma_Einstein}, \eqref{eqn:BSigma_Components}, 
 \eqref{eqn:Symmetrisation_parameter_values}, \eqref{eqn:P_projector_def}-\eqref{eqn:Pcal_perp_defn}, \eqref{eqn:McalSigma_defn}, \eqref{eqn:Fuchsian_Solution_Vector_Defns}, \eqref{eqn:tbar_ttilde_defns} and \eqref{eqn:uh_w_defs_app}-\eqref{eqn:nuh_w_identity_app}, we observe that the matrices \eqref{eqn:B_scr_Sigma} can be viewed as matrix valued-maps
\begin{align*}
    \mathscr{B}^{\Sigma} = \mathscr{B}^{\Sigma}\bigl(t,\bar{t}_{(\afrak)},\frg_{(\afrak)}, \mathcal{P}^{\perp}V,\overline{W}_{(\afrak)}\bigr)
\end{align*} 
that are smooth and symmetric on the domain
\begin{align*} 
(-r,2)\times  \prod_{\afrak = 1}^{2}(-r,2^{\mu_{(\afrak)}}) \times B_{R}(\mathbb{R}^{2}) \times B_{R}(\mathbb{R}^{26}) \times B_{R}(\mathbb{R}^{10}).
 \end{align*} 
Differentiating  $\mathscr{B}^{\Sigma}$ spatially gives
\begin{align}
\label{eqn:BscrSigma_spatial_derivative}
    \del_{\Sigma}\big(\mathscr{B}^{\Sigma}(t,\bar{t}_{(\afrak)},\frg_{(\afrak)}, \mathcal{P}^{\perp}V,\overline{W}_{(\afrak)})\big) = D\mathscr{B}^{\Sigma}(t,\bar{t}_{(\afrak)},\frg_{(\afrak)}, \mathcal{P}^{\perp}V,\overline{W}_{(\afrak)}) \cdot \begin{pmatrix}0 \\ 0 \\ 0 \\ \mathcal{P}^{\perp}\del_{\Sigma}V \\ \del_{\Sigma}\overline{W}_{(\afrak)} \end{pmatrix}.
\end{align}

Now, the matrix divergence $\text{Div}\mathscr{B}$ is,  cf.~ \eqref{eqn:Bscr0_time_derivative}, \eqref{eqn:time_deriv_Wbar_subsystem2_analysis} and \eqref{eqn:BscrSigma_spatial_derivative}, the matrix-valued map defined by
\begin{align*}
\text{Div}\mathscr{B}\bigl(t,V,\overline{W}_{(\afrak)},V_\Sigma,\overline{W}^{(\afrak)}_{\Sigma})\bigr)=&
D\mathscr{B}^{0}\bigl(t^{\mu_{(\afrak)}},\frg_{(\afrak)}(t),\overline{W}_{(\afrak)}\bigr)\cdot\begin{pmatrix} \mu_{(\afrak)} t^{\mu_{(\afrak)}-1} \\ \frg^{\prime}_{(\afrak)}(t) \\ \mathbf{Q}_{(\afrak)}\bigl(t,t^{\mu_{(\afrak)}},g_{(\mathfrak{a})}(t),V,\overline{W}_{(\afrak)}\bigr) \end{pmatrix} \notag \\
&
+ D\mathscr{B}^{\Sigma}(t,t^{\mu_{(\afrak)}},\frg_{(\afrak)}(t), \mathcal{P}^{\perp}V,\overline{W}_{(\afrak)}) \cdot \begin{pmatrix}0 \\ 0 \\ 0 \\ \mathcal{P}^{\perp}V_{\Sigma} \\ \overline{W}^{(\afrak)}_\Sigma \end{pmatrix}.
\end{align*}
From \eqref{eqn:Mu1<Mu2}, the smoothness of the maps  $\mathscr{B}^{0}$,  $\mathscr{B}^{\Sigma}$, and $\mathbf{Q}_{(\afrak)}$, the bound on the map $\mathbf{Q}_{(\afrak)}$ from \eqref{eqn:Qbf_H_defn} and $\eqref{eqn:Qbf_Bound}$, and the bounds on $\frg_{(\afrak)}(t)$ and $\frg_{(\afrak)}'(t)$ from Proposition \ref{prop:Asymptotic_ODE}, we deduce that $\text{Div}\mathscr{B}$ is bounded by 
\begin{align}
\label{eqn:DivBscr_Bound}
    \bigl|\text{Div}\mathscr{B}\bigl(t, V,\overline{W}_{(\afrak)},V_\Sigma,\overline{W}^{(\afrak)}_{\Sigma}\bigr)\bigr| \lesssim \bigl(t^{\mu_{(1)}-1}+1\bigr)\bigl|\bigl(V,\overline{W}_{(\afrak)},V_{\Omega},\overline{W}^{(\afrak)}_{\Omega}\bigr)\bigr|
\end{align}
for $\bigl(t,V,\overline{W}_{(\afrak)},V_\Sigma,\overline{W}^{(\afrak)}_{\Sigma}\bigr)\in (0,1]\times B_R(\mathbb{R}^{35})\times B_R(\mathbb{R}^{10})\times B_R(\mathbb{R}^{35\times 3})\times B_R(\mathbb{R}^{10\times 3})$. \newline \par

Now, we consider the source term $\mathbf{R}$ from the evolution equation \eqref{eqn:EinsteinEuler_1stDeriv}. First, by \eqref{eqn:Hct_Alphat_defs}, \eqref{eqn:Hct_identities}-\eqref{eqn:lapse_Hct_identity}, \eqref{eqn:pi_projection_tracefree}, \eqref{eqn:f_modified_frame_defn}, \eqref{eqn:Scal_operator_defn}, \eqref{eqn:B0_Einstein}-\eqref{eqn:BSigma_Einstein}, \eqref{eqn:Symmetrisation_parameter_values}, \eqref{eqn:Fcal_components_Final},  \eqref{eqn:Grav_bf_Variables_Defn}, \eqref{eqn:Fuchsian_Solution_Vector_Defns}, \eqref{eqn:tbar_ttilde_defns}, and \eqref{eqn:Tab_Decomp_Wbar_Vars_app}, it is straightforward to verify that the first component \eqref{eqn:Rbf_0_defn} of $\mathbf{R}$ can be viewed as a map
\begin{align*}
\mathbf{R}_{0} = \mathbf{R}_{0}\bigl(t,\bar{t}_{(\afrak)},\tilde{t}_{(\afrak)},g_{(\mathfrak{a})},V,\overline{W}_{(\afrak)}, V_{\Omega}, \overline{W}^{(\afrak)}_{\Omega}\bigr)
\end{align*}
that is smooth on the domain
\begin{align*}
 (-r,2)& \times \prod_{\afrak = 1}^{2}\bigl(-r,2^{\mu_{(\afrak)}}\bigr) \times \prod_{\afrak = 1}^{2}\biggl(-r,2^{\frac{4K_{(\afrak)}}{1-K_{(\afrak)}}}\biggr) \times B_{R}(\mathbb{R}^{2}) \nonumber \\
&\times B_{R}(\mathbb{R}^{35}) 
\times B_{R}(\mathbb{R}^{10}) \times B_{R}(\mathbb{R}^{35\times 3}) \times B_{R}(\mathbb{R}^{10\times 3}),
\end{align*}
and vanishes for $ \bigl(V, \overline{W}_{(\afrak)}, V_{\Sigma},\overline{W}_{\Sigma}^{(\afrak)}\bigr)=0$.
Similarly, from \eqref{eqn:Hct_Alphat_defs},\eqref{eqn:Hct_identities}-\eqref{eqn:lapse_Hct_identity}, \eqref{eqn:Source_ttt_defns}, \eqref{eqn:A0_final}, \eqref{eqn:Pi_def}, \eqref{eqn:Pi_perp_defn},  \eqref{eqn:Abf_def}, \eqref{eqn:Hbf_def}, \eqref{eqn:Sbf_def}, \eqref{eqn:Fuchsian_Solution_Vector_Defns}, \eqref{eqn:Rbf_a_defn} and \eqref{eqn:uh_w_defs_app}-\eqref{eqn:nuh_w_identity_app}, a careful inspection shows that the maps
\begin{equation*}
\Pi^{\perp}C_{(\afrak)}^{0}\Pi^{\perp}\del_{\Omega}\mathbf{S}_{(\afrak)}
\end{equation*}
and
\begin{equation*}
\mathbf{R}_{(\afrak)}-t^{2\mu_{(\afrak)}-1}\Pi^{\perp}C_{(\afrak)}^{0}\Pi^{\perp}\del_{\Omega}\mathbf{S}_{(\afrak)}
\end{equation*}
are smooth in the variables $\bigl(t,\bar{t}_{(\afrak)},g_{(\mathfrak{a})},V,\overline{W}_{(\afrak)}, V_{\Omega}, \overline{W}^{(\afrak)}_{\Omega}\bigr)$ and $\bigl(t,\bar{t}_{(\afrak)},g_{(\mathfrak{a})},V,\overline{W}_{(\afrak)}, V_{\Omega}, \overline{W}^{(\afrak)}_{\Omega}, \overline{W}^{(\afrak)}_{\Omega\Sigma}\bigr)$, respectively, 
on the domains
\begin{align*}
(-r,2) \times \prod_{\afrak = 1}^{2}\bigl(-r,2^{\mu_{(\afrak)}}\bigr) \times B_{R}(\mathbb{R}^{2}) \times B_{R}(\mathbb{R}^{35}) 
\times B_{R}(\mathbb{R}^{10}) \times B_{R}(\mathbb{R}^{35\times 3}) \times B_{R}(\mathbb{R}^{10\times 3}),
\end{align*}
and
\begin{align*}
(-r,2)& \times \prod_{\afrak = 1}^{2}\bigl(-r,2^{\mu_{(\afrak)}}\bigr) \times B_{R}(\mathbb{R}^{2}) \times B_{R}(\mathbb{R}^{35}) \nonumber \\
&\times B_{R}(\mathbb{R}^{10}) \times B_{R}(\mathbb{R}^{35\times 3}) \times B_{R}(\mathbb{R}^{10\times 3})\times B_{R}(\mathbb{R}^{10 \times 9}),
\end{align*}
and both maps vanish for $ \bigl(V, \overline{W}_{(\afrak)}, V_{\Sigma},\overline{W}_{\Sigma}^{(\afrak)}, \overline{W}_{\Omega\Sigma}^{(\afrak)}\bigr)=0$.
Thus, from  \eqref{eqn:Mu1<Mu2} and the bounds on $\frg_{(\afrak)}$ established in Proposition \ref{prop:Asymptotic_ODE}, we deduce that
\begin{align}
\label{eqn:Rbf_Bound}
\bigr|\mathbf{R}\bigr(t,\overline{W}_{(\afrak)},V_{\Omega}, \overline{W}^{(\afrak)}_{\Omega},\overline{W}^{(\afrak)}_{\Omega\Sigma}\bigr)\bigr|
    \lesssim \bigl(t^{2\mu_{(1)}-1} + 1\bigr)\bigr|\bigl(V,\overline{W}_{(\afrak)},V_{\Omega},\overline{W}^{(\afrak)}_{\Omega},\overline{W}^{(\afrak)}_{\Omega\Sigma}\bigr)\bigr|
\end{align}
for 
\begin{align*}
\bigl(t,V,\overline{W}_{(\afrak)},V_{\Omega},\overline{W}^{(\afrak)}_{\Omega},\overline{W}^{(\afrak)}_{\Omega\Sigma}\bigr)\in &(0,1]\times B_{R}(\mathbb{R}^{35}) \times B_{R}(\mathbb{R}^{10})\nonumber \\
&\times B_{R}(\mathbb{R}^{35 \times 3}) \times B_{R}(\mathbb{R}^{10 \times 3}) \times B_{R}(\mathbb{R}^{10 \times 9}).    
\end{align*}

\subsubsection{Second Derivative Sub-System}
\label{sec:CoefficientAnalysis_2ndDeriv}
Finally, we consider the second differentiated sub-system \eqref{eqn:EinsteinEuler_2ndDeriv}. Here, we must be careful with the analysis due to the renormalisation of the second derivatives of the gravitational and fluid variables by $t^{-\psi}$ and $t^{-\mu_{(\afrak)}}$, respectively, i.e. 
\begin{align*}
    \del_{\Gamma}\del_{\Omega}V = t^{-\psi}V_{\Gamma\Omega}, \;\; \del_{\Gamma}\del_{\Omega}W^{(\afrak)} = t^{-\mu_{(\afrak)}}\overline{W}^{(\afrak)}_{\Gamma\Omega}.
\end{align*}

To start the analysis, we observe, with the help of 
\eqref{eqn:A0_final}, \eqref{eqn:Pi_def}, \eqref{eqn:Pi_perp_defn} and
\eqref{eqn:Pbb2_Def}, that the matrix $\mathscr{E}$, defined above by \eqref{eqn:E_scr}, satisfies
\begin{align*}
    [\mathscr{E},\mathbb{P}_{2}] = 0.
\end{align*}
Using
\eqref{eqn:lapse_Hct_identity}, \eqref{eqn:B0_Einstein}, 
\eqref{eqn:Pcal_perp_defn}, \eqref{eqn:A0_final} , \eqref{eqn:Mu1<Mu2}, \eqref{eqn:Fuchsian_Solution_Vector_Defns},  \eqref{eqn:tbar_ttilde_defns} and \eqref{eqn:uh_w_defs_app}-\eqref{eqn:nuh_w_identity_app},
it is not difficult to verify that
$\mathscr{E}$ defines a matrix-valued map 
\begin{align*}
  \mathscr{E} = \mathscr{E}\bigl(t,\bar{t}_{(\afrak)},\frg_{(\afrak)},V,\overline{W}_{(\afrak)},V_{\Omega},\overline{W}^{(\afrak)}_{\Omega},V_{\Omega\Sigma},\overline{W}^{(\afrak)}_{\Omega\Sigma}\bigr)
\end{align*} 
that is smooth on the domain
\begin{align*}
(-r,2)& \times \prod_{\afrak = 1}^{2}\bigl(-r,2^{\mu_{(\afrak)}}\bigr) \times B_{R}(\mathbb{R}^{2}) \times B_{R}(\mathbb{R}^{35})  \times B_{R}(\mathbb{R}^{10})  \nonumber \\ 
& \times B_{R}(\mathbb{R}^{35\times 3})  \times B_{R}(\mathbb{R}^{10\times 3}) \times B_{R}(\mathbb{R}^{35\times 9}) \times B_{R}(\mathbb{R}^{10\times 9})
\end{align*}
and satisfies
\begin{align*}
    \mathscr{E} \gtrsim \mathbbm{1}.
\end{align*}
From similar considerations, it is straightforward to verify that $\mathscr{C}^{0}$, defined in \eqref{eqn:C_scr_0}, 
defines a matrix-valued map
\begin{align*}
\mathscr{C}^{0} = \mathscr{C}^{0}\bigl(t,\bar{t}_{(\afrak)},\frg_{(\afrak)},\mathcal{P}^{\perp}V,\overline{W}_{{\afrak}}\bigr)
\end{align*}
that is smooth and symmetric on the domain 
\begin{align*}
  (-r,2)\times\prod_{\afrak = 1}^{2}\bigl(-r,2^{\mu_{(\afrak)}}\bigr)\times B_{R}(\mathbb{R}^{2})\times B_{R}(\mathbb{R}^{26})\times B_{R}(\mathbb{R}^{10})
\end{align*} 
and satisfies
\begin{align*}
    \mathscr{C}^{0} \gtrsim \mathbbm{1}.
\end{align*}
We note also from \eqref{eqn:A0_final}, \eqref{eqn:Pi_def}, \eqref{eqn:Pi_perp_defn} and \eqref{eqn:Pbb2_Def} that
\begin{align*}
\mathbb{P}_{2}^{\perp}\mathscr{C}^{0}\mathbb{P}_{2} = \mathbb{P}_{2}\mathscr{C}^{0}\mathbb{P}_{2}^{\perp} = 0.
\end{align*}
Furthermore, taking the formal time derivative of $\mathscr{C}^{0}$ yields 
\begin{align}
\label{eqn:C_scr_0_timederivative}
    \del_{t}\big(\mathscr{C}^{0}(t,\bar{t}_{(\afrak)},\frg_{(\afrak)}(t),\mathcal{P}^{\perp}V,\overline{W}_{{\afrak}})\big) = D\mathscr{C}^{0}(t,\bar{t}_{(\afrak)},\frg_{(\afrak)}(t),\mathcal{P}^{\perp}V,\overline{W}_{{\afrak}})\cdot\begin{pmatrix}1 \\ \mu_{(\afrak)}t^{\mu_{(\afrak)}-1} \\ \frg^{\prime}_{(\afrak)}(t) \\ \del_{t}\mathcal{P}^{\perp}V \\ \del_{t}\overline{W}_{(\afrak)}  \end{pmatrix},
\end{align}
where $\del_{t}\mathcal{P}^{\perp}V$ and $\del_{t}\overline{W}_{(\afrak)}$ are determined using \eqref{eqn:EinsteinEuler_NoDeriv} and \eqref{eqn:Qbfsource_Background_Error_Split}-\eqref{eqn:Qbf_bar_defn}, that is,
\begin{equation*}
\begin{aligned}
   \del_{t}\mathcal{P}^{\perp}V &= (B^{0})^{-1}\Big(\alpha \mathcal{P}^{\perp}B^{\Omega}\del_{\Omega}V + \mathcal{P}^{\perp}\bar{\mathbf{Q}} + \mathcal{P}^{\perp}\hat{\mathbf{Q}} \Big), \\
   \del_{t}\overline{W}_{(\afrak)} &= \mathbf{Q}_{(\afrak)}.
\end{aligned}
\end{equation*}

Next, we consider the matrices $\mathscr{C}^{\Sigma}$ defined by \eqref{eqn:C_scr_Sigma}. By
\eqref{eqn:Hct_Alphat_defs}, \eqref{eqn:Hct_identities}-\eqref{eqn:lapse_Hct_identity}, \eqref{eqn:BSigma_Einstein}, \eqref{eqn:BSigma_Components}, \eqref{eqn:Symmetrisation_parameter_values}, \eqref{eqn:AF_final}, and \eqref{eqn:uh_w_defs_app}-\eqref{eqn:nuh_w_identity_app}, we can view the $\mathscr{C}^{\Sigma}$ as matrix-valued maps  
\begin{align*}
\mathscr{C}^{\Sigma} = \mathscr{C}^{\Sigma}(t,\bar{t}_{(\afrak)},\frg_{(\afrak)},\mathcal{P}^{\perp}V,\overline{W}_{(\afrak)})
\end{align*}
that are smooth and symmetric on the domain 
\begin{align*}
    (-r,2)\times\prod_{\afrak = 1}^{2}(-r,2^{\mu_{(\afrak)}})\times B_{R}(\mathbb{R}^{2})\times B_{R}(\mathbb{R}^{26})\times B_{R}(\mathbb{R}^{10}).
\end{align*}
From this we see that differentiating $\mathscr{C}^{\Sigma}$ spatially gives
\begin{align}
\label{eqn:C_scr_Sigma_spatialderivative}
    \del_{\Sigma}\big(\mathscr{C}^{\Sigma}(t,\bar{t}_{(\afrak)},\frg_{(\afrak)}(t),\mathcal{P}^{\perp}V,\overline{W}_{(\afrak)})\big) = D\mathscr{C}^{\Sigma}(t,\bar{t}_{(\afrak)},\frg_{(\afrak)}(t),\mathcal{P}^{\perp}V,\overline{W}_{(\afrak)})\cdot\begin{pmatrix}0 \\ 0 \\ 0\\ \del_{\Sigma}\mathcal{P}^{\perp}V \\ \del_{\Sigma}\overline{W}^{(\afrak)} \end{pmatrix}.
\end{align}

Now, the matrix divergence $\text{Div}\mathscr{C}$ is,  cf.~ \eqref{eqn:C_scr_0_timederivative}, \eqref{eqn:C_scr_0_timederivative} and \eqref{eqn:C_scr_Sigma_spatialderivative}, the matrix-valued map defined by
\begin{align*}
\text{Div}\mathscr{C}\bigl(t,V,\overline{W}_{(\afrak)},V_\Sigma,\overline{W}^{(\afrak)}_{\Sigma}\bigr) &=D\mathscr{C}^{\Sigma}\bigl(t,t^{\mu_{(\afrak)}},\frg_{(\afrak)}(t),\mathcal{P}^{\perp}V,\overline{W}_{(\afrak)}\bigr)\cdot\begin{pmatrix}0 \\ 0 \\ 0\\ \mathcal{P}^{\perp}V_{\Sigma} \\ \overline{W}^{(\afrak)}_\Sigma \end{pmatrix}\nonumber \\
 +D\mathscr{C}^{0}\bigl(t,t^{\mu_{(\afrak)}},\frg_{(\afrak)}(t),&\mathcal{P}^{\perp}V,\overline{W}_{{(\afrak)}}\bigr)\cdot\begin{pmatrix}1 \\ \mu_{(\afrak)}t^{\mu_{(\afrak)}-1} \\ \frg^{\prime}_{(\afrak)}(t) \\ (B^{0})^{-1}\bigl(\alpha \mathcal{P}^{\perp}B^{\Omega}V_\Omega + \mathcal{P}^{\perp}\bar{\mathbf{Q}} + \mathcal{P}^{\perp}\hat{\mathbf{Q}} \bigr) \\ \mathbf{Q}_{(\afrak)} \end{pmatrix}.
\end{align*}
From \eqref{eqn:Mu1<Mu2}, the bounds on $\frg_{(\afrak)}(t)$ and $\frg^{\prime}_{(\afrak)}(t)$ in Proposition \ref{prop:Asymptotic_ODE}, the bounds on $\mathbf{Q}_{E}$ and $\mathbf{Q}_{H}$  from \eqref{eqn:Qbf_E_Bound} and \eqref{eqn:Qbf_Bound} respectively, and the smoothness properties of $\mathbf{Q}_{E}$, $\mathbf{Q}_{H}$, $\mathscr{C}^{0}$, and $\mathscr{C}^{\Sigma}$, we deduce that the matrix divergence $\text{Div}\mathscr{C}$ is bounded by
\begin{align}
\label{eqn:DivCscr_Bound}
    \bigl|\text{Div}\mathscr{C}\bigl(t, V,\overline{W}_{(\afrak)},V_\Sigma,\overline{W}^{(\afrak)}_{\Sigma}\bigr)\bigr| \lesssim \bigl(t^{\mu_{(1)}-1}+1\bigr)\bigr|(V,\overline{W}_{(\afrak)},V_{\Omega},\overline{W}^{(\afrak)}_{\Omega})\bigr|
\end{align}
for  $\bigl(t,V,\overline{W}_{(\afrak)},V_{\Omega},\overline{W}^{(\afrak)}_{\Omega}\bigr)\in (0,1]\times B_{R}(\mathbb{R}^{35}) \times B_{R}(\mathbb{R}^{10})\times B_{R}(\mathbb{R}^{35 \times 3}) \times B_{R}(\mathbb{R}^{10 \times 3}).$ \newline \par

Finally, we must analyse the source term $\mathbf{G}$ defined above by \eqref{eqn:Gbf_Source_Def}. We begin by considering the term $\mathbf{G}_{0}$ defined by \eqref{eqn:Gbf_0_defn}. Setting 
\begin{align*}
    \mathring{t}_{(\afrak)} = t^{\psi-\mu_{(\afrak)}+2}, \quad \hat{t} = t^{\psi}, \quad \text{and} \quad \check{t} = t_{(\afrak)}^{\psi-\mu_{(\afrak)}+\frac{4K_{(\afrak)}}{1-K_{(\afrak)}}}
\end{align*}
we observe, with the help of \eqref{eqn:Hct_Alphat_defs}, \eqref{eqn:Hct_identities}-\eqref{eqn:lapse_Hct_identity}, \eqref{eqn:B0_Einstein},  \eqref{eqn:Fcal_components_Final}, \eqref{eqn:Grav_bf_Variables_Defn}, \eqref{eqn:Fuchsian_Solution_Vector_Defns}, and 
\eqref{eqn:uh_w_defs_app}-\eqref{eqn:Tab_Decomp_Wbar_Vars_app}, that $\mathbf{G}_{0}$ can be viewed as a matrix-valued map
\begin{align*}
   \mathbf{G}_{0} = \mathbf{G}_{0}\bigl(t,\bar{t}_{(\afrak)},\tilde{t}_{(\afrak)},\mathring{t}_{(\afrak)},\hat{t},\check{t}_{(\afrak)},\frg_{(\afrak)}, V, \overline{W}_{(\afrak)}, V_{\Sigma},\overline{W}_{\Sigma}^{(\afrak)}, V_{\Omega\Sigma}, \overline{W}_{\Omega\Sigma}^{(\afrak)}\bigr)
\end{align*}
that is smooth on the domain 
\begin{align*}
     (-r&,2) \times \prod_{\afrak = 1}^{2}\bigl(-r,2^{\mu_{(\afrak)}}\bigr)\times \prod_{\afrak = 1}^{2}\Bigl(-r,2^{\frac{4K_{(\afrak)}}{1-K_{(\afrak)}}}\Bigr) \times \prod_{\afrak = 1}^{2}\bigl(-r,2^{\psi-\mu_{(\afrak)}+2}\bigr) \times \bigl(-r,2^{\psi}\bigr)  \nonumber \\
     &\times \prod_{\afrak = 1}^{2}\bigl(-r,2^{\psi-\mu_{(\afrak)}+\frac{4K_{(\afrak)}}{1-K_{(\afrak)}}}\bigr) \times B_{R}(\mathbb{R}^{2})  \times B_{R}(\mathbb{R}^{35}) 
\times B_{R}(\mathbb{R}^{10}) \times B_{R}(\mathbb{R}^{35\times 3}) \nonumber \\
&\times B_{R}(\mathbb{R}^{10\times 3}) \times B_{R}(\mathbb{R}^{35\times 9}) \times B_{R}(\mathbb{R}^{10\times 9})
\end{align*}
and vanishes for $ \bigl(V, \overline{W}_{(\afrak)}, V_{\Sigma},\overline{W}_{\Sigma}^{(\afrak)}, V_{\Omega\Sigma}, \overline{W}_{\Omega\Sigma}^{(\afrak)}\bigr)=0$. 
It is worth noting that the $\mathring{t}_{(\afrak)}$ and $\check{t}_{(\afrak)}$ dependence is due to the appearance of $t^{\psi-\mu_{(\afrak)}+2}$ and $t_{(\afrak)}^{\psi-\mu_{(\afrak)}+\frac{4K_{(\afrak)}}{1-K_{(\afrak)}}}$ terms that arise from expressing the stress-energy terms in $t^{\psi}B^{0}\del_{\Gamma}\del_{\Omega}\mathbf{F}$ in terms of the renormalised second derivative variables $\overline{W}_{\Omega\Sigma}^{(\afrak)}$, see \eqref{eqn:Fuchsian_Solution_Vector_Defns}.
Due to \eqref{eqn:Mu1<Mu2},  $\psi>0$, and the bound on $\frg_{(\afrak)}(t)$ from Proposition \ref{prop:Asymptotic_ODE},  we conclude that
\begin{align}
\label{eqn:Gbf_0_bound}
    \bigl|\mathbf{G}_{0}\bigl(t,V,\overline{W}_{(\afrak)},V_{\Omega},\overline{W}^{(\afrak)}_{\Omega},V_{\Sigma\Omega},\overline{W}^{(\afrak)}_{\Sigma\Omega}\bigr)\bigr| \lesssim \bigl(t^{\psi-\mu_{(2)}+2} +1\bigr)\bigr|\bigl(V,\overline{W}_{(\afrak)},V_{\Omega},\overline{W}^{(\afrak)}_{\Omega},V_{\Sigma\Omega},\overline{W}^{(\afrak)}_{\Sigma\Omega}\bigr)\bigr|
\end{align}
for 
\begin{align*}
\bigl(t,V,\overline{W}_{(\afrak)},V_{\Omega},\overline{W}^{(\afrak)}_{\Omega},V_{\Sigma\Omega},\overline{W}^{(\afrak)}_{\Sigma\Omega}\bigr)&\in  (0,1]\times B_{R}(\mathbb{R}^{35}) \times B_{R}(\mathbb{R}^{10}) \nonumber \\
&\quad \times B_{R}(\mathbb{R}^{35 \times 3}) \times B_{R}(\mathbb{R}^{10 \times 3})\times B_{R}(\mathbb{R}^{10 \times 9}).
\end{align*}
It remains to analyse the source term $\mathbf{G}_{(\afrak)}$ defined in \eqref{eqn:Gbf_a_defn}. Similar to $\mathbf{G}_{0}$, the most singular terms arise from expressing second derivatives in terms of the renomalised variables $V_{\Gamma\Omega}$ and $\overline{W}_{\Gamma\Omega}$. Setting 
\begin{align*}
    t_{(\afrak)}^{\#} = t^{\mu_{(\afrak)}-\psi} \quad \text{and} \quad \breve{t} = t^{-\psi},
\end{align*}
it is not difficult using \eqref{eqn:Hct_Alphat_defs}, \eqref{eqn:Hct_identities}-\eqref{eqn:lapse_Hct_identity}, \eqref{eqn:A0_final}-\eqref{eqn:Pi_def}, \eqref{eqn:Pi_perp_defn}, \eqref{eqn:Abf_def},  \eqref{eqn:Hbf_def}, \eqref{eqn:Sbf_def}, \eqref{eqn:Fuchsian_Solution_Vector_Defns}, \eqref{eqn:Gbf_0_bound}, \eqref{eqn:Gbf_a_bound}, and \eqref{eqn:uh_w_defs_app}-\eqref{eqn:nuh_w_identity_app} to verify that the $\mathbf{G}_{(\afrak)}$ can be viewed as matrix-valued maps
\begin{align*}
   \mathbf{G}_{(\afrak)} = \mathbf{G}_{(\afrak)}\bigl(t,\bar{t}_{(\afrak)},t^{\#}_{(\afrak)},\breve{t},\frg_{(\afrak)}, V, \overline{W}_{(\afrak)}, V_{\Sigma},\overline{W}_{\Sigma}^{(\afrak)}, V_{\Omega\Sigma}, \overline{W}_{\Omega\Sigma}^{(\afrak)}\bigr)
\end{align*}
that are smooth on the domain 
\begin{align*}
& (-r,2) \times \prod_{\afrak = 1}^{2}\big(-r,2^{\mu_{(\afrak)}}\big) \times \prod_{\afrak = 1}^{2}\big(-r,2^{\mu_{(\afrak)}-\psi}\big) \times \prod_{\afrak = 1}^{2}\big(-r,2^{-\psi}\big) \times B_{R}(\mathbb{R}^{2}) \nonumber \\ &\quad \times B_{R}(\mathbb{R}^{35}) \times B_{R}(\mathbb{R}^{10}) \times B_{R}(\mathbb{R}^{35\times 3}) \times B_{R}(\mathbb{R}^{10\times 3}) \times B_{R}(\mathbb{R}^{35\times 9}) \times B_{R}(\mathbb{R}^{10\times 9})
\end{align*}
and vanish for $\bigl( \overline{W}_{(\afrak)}, V_{\Sigma},\overline{W}_{\Sigma}^{(\afrak)}, V_{\Omega\Sigma}, \overline{W}_{\Omega\Sigma}^{(\afrak)}\bigr)=0$.
Due to \eqref{eqn:Mu1<Mu2}, $\psi>0$, and the bound on $\frg_{(\afrak)}(t)$ from Proposition \ref{prop:Asymptotic_ODE}, we conclude that
\begin{align}
\label{eqn:Gbf_a_bound}  
\bigl|\mathbf{G}_{(\afrak)}\bigl(t,V,\overline{W}_{(\afrak)},V_{\Omega},\overline{W}^{(\afrak)}_{\Omega},V_{\Sigma\Omega},\overline{W}^{(\afrak)}_{\Sigma\Omega}\bigr)\bigr| \lesssim (t^{\mu_{(1)}-\psi}+t^{-\psi} +1)\bigr|\bigl(V,\overline{W}_{(\afrak)},V_{\Omega},\overline{W}^{(\afrak)}_{\Omega},V_{\Sigma\Omega},\overline{W}^{(\afrak)}_{\Sigma\Omega}\bigr)\bigr|
\end{align}
for 
\begin{align*}
\bigl(t,V,\overline{W}_{(\afrak)},V_{\Omega},\overline{W}^{(\afrak)}_{\Omega},V_{\Sigma\Omega},\overline{W}^{(\afrak)}_{\Sigma\Omega}\bigr)&\in  (0,1]\times B_{R}(\mathbb{R}^{35}) \times B_{R}(\mathbb{R}^{10}) \nonumber \\
&\quad \times B_{R}(\mathbb{R}^{35 \times 3}) \times B_{R}(\mathbb{R}^{10 \times 3})\times B_{R}(\mathbb{R}^{10 \times 9}).
\end{align*}
Together, \eqref{eqn:Gbf_0_bound} and \eqref{eqn:Gbf_a_bound} imply the source term $\mathbf{G}$ satisfies
\begin{align*}
&\bigl|\mathbf{G}\bigl(t,V,\overline{W}_{(\afrak)},V_{\Omega}, \overline{W}^{(\afrak)}_{\Omega},V_{\Sigma\Omega},\overline{W}^{(\afrak)}_{\Sigma\Omega}\bigr)\bigr| \nonumber \\
&\qquad \lesssim \bigl(t^{-\psi} + t^{\psi-\mu_{(2)}+2} + t^{\mu_{(1)}-\psi} + 1\bigr)\bigr|\bigl(V,\overline{W}_{(\afrak)},V_{\Omega},\overline{W}^{(\afrak)}_{\Omega},V_{\Sigma\Omega},\overline{W}^{(\afrak)}_{\Sigma\Omega}\bigr)\bigr|
\end{align*}
for 
\begin{align*}
\bigl(t,V,\overline{W}_{(\afrak)},V_{\Omega},\overline{W}^{(\afrak)}_{\Omega},V_{\Sigma\Omega},\overline{W}^{(\afrak)}_{\Sigma\Omega}\bigr)&\in  (0,1]\times B_{R}(\mathbb{R}^{35}) \times B_{R}(\mathbb{R}^{10}) \nonumber \\
&\quad \times B_{R}(\mathbb{R}^{35 \times 3}) \times B_{R}(\mathbb{R}^{10 \times 3})\times B_{R}(\mathbb{R}^{10 \times 9}).
\end{align*}

\subsubsection{Fuchsian Constants}
\label{sec:Fuchsian_Constants}
For the coefficients, e.g. $\mathscr{A}^{0}$, $\mathscr{A}^{\Sigma}$, $\mathscr{D}$, $\mathbf{Q}$, etc., of the Fuchsian system \eqref{eqn:EinsteinEuler_NoDeriv}-\eqref{eqn:EinsteinEuler_2ndDeriv} to
satisfy the conditions from \cite{BOOS:2021}*{\S 3.4}, it will be enough if we can arrange that the only non-integrable singular terms that appear in the system are in the terms $t^{-1}\mathscr{D}\mathbb{P}_1$ and
$t^{-1}\mathscr{E}\mathbb{P}_2$.
To accomplish this, we will ensure that the other singular terms appearing in the Einstein-Euler system are integrable by demanding that the power of $t$ appearing in the bounds 
\eqref{eqn:Qbf_Bound}, \eqref{eqn:DivBscr_Bound}, \eqref{eqn:Rbf_Bound}, \eqref{eqn:DivCscr_Bound}, \eqref{eqn:Gbf_0_bound}, and \eqref{eqn:Gbf_a_bound} are of the form
\begin{align}
\label{eqn:integrable_t_example}
    t^{\varphi} \quad\text{with} \quad \varphi > -1.
\end{align}
This demand leads to the system of inequalities
\begin{align}
\label{eqn:Mu_integrable_ineqs}
    2\mu_{(\afrak)}-1>-1,& \;\; \mu_{(\afrak)}>-1, \\
\label{eqn:Psi_integrable_ineqs}
    \psi-\mu_{(\afrak)}+2>-1, \;\; -\psi&>-1, \;\; \mu_{(\afrak)}-\psi > -1.
\end{align}
Clearly, the conditions \eqref{eqn:Mu_integrable_ineqs} are satisfied since $\mu_{(\afrak)}>0$ for $K \in (\frac{1}{3},1)$. On the other hand, since $\psi>0$, it is straightforward to see that \eqref{eqn:Psi_integrable_ineqs} implies that
\begin{equation*}
0<\mu_1 <\mu_2 < 3+\psi \quad \text{and} \quad 0<\psi < 1.
\end{equation*}
Recalling that $\mu_{(\afrak)}=\frac{3K_{(\afrak)}-1}{1-K_{(\afrak)}}$, these inequalities can be expressed in terms the parameters $K_{(\afrak)}$ as follows:
\begin{enumerate}[(i)]
\label{eqn:Psi_inequalities}
\item $\frac{1}{3} < K_{(1)} \leq K_{(2)} < \frac{2}{3} \quad \text{and} \quad   0 < \psi < 1,$
\item $\frac{1}{3} < K_{(1)} < \frac{2}{3}, \quad \frac{2}{3}\leq K_{(2)} < \frac{5}{7}, \quad \text{and} \quad \frac{4-6K_{(2)}}{K_{(2)}-1} < \psi < 1,$ 
\item $\frac{2}{3} \leq K_{(1)} \leq K_{(2)} < \frac{5}{7} \quad \text{and} \quad \frac{4-6K_{(2)}}{K_{(2)}-1} < \psi < 1. $
\end{enumerate}
This shows that we can satisfy the conditions \eqref{eqn:Mu_integrable_ineqs}-\eqref{eqn:Psi_integrable_ineqs} for any choice of $K_{(1)}$ and $K_{(2)}$ satisfying $\frac{1}{3}< K_{(1)} \leq K_{(2)} <\frac{5}{7}$. \newline \par

As discussed in \cite{BOOS:2021}*{\S3.4}, singular terms of the form \eqref{eqn:integrable_t_example} are handled by employing a time transformation of the form $t\mapsto t^{p}$ for a suitable choice of $p\in (0,1)$. From \eqref{K-range}, \eqref{mu-def}, \eqref{eqn:Mu1<Mu2} and $\psi>0$, we find that $p$ must be chosen as
\begin{align}
\label{eqn:p_time_transformation_defn}
p &= \begin{cases}
        \mu_{(1)} &\text{if $\mu_{(1)}\leq 1$ \& $\mu_{(2)}< 2+\psi$} \\
        1 + (\psi-\mu_{(2)}+2) &\text{if $\mu_{(1)}> 1$ \& $\mu_{(2)}\geq 2+\psi$} \\
        \min\bigl\{\mu_{(1)},1 + (\psi-\mu_{(2)}+2)\bigr\} &\text{if $\mu_{(1)}\leq 1$ \& $\mu_{(2)}\geq 2+\psi$}
    \end{cases}.
\end{align}
Thus, the bounds established in the previous sections imply that for $K_{(\afrak)} \in (\frac{1}{3},\frac{5}{7})$, $\psi$ chosen according to \eqref{eqn:Psi_inequalities}, and $R_{0}>0$ chosen sufficiently small there exist positive constants $\theta$, $\gamma_{1}=\tilde{\gamma}_{1}$, $\gamma_{2}=\tilde{\gamma}_{2}$, and $\kappa = \tilde{\kappa}$ such that the Fuchsian system \eqref{eqn:EinsteinEuler_NoDeriv}-\eqref{eqn:EinsteinEuler_2ndDeriv} satisfies the coefficient assumptions from \cite{BOOS:2021}*{\S3.4} for $p$ defined by \eqref{eqn:p_time_transformation_defn} and the following choice of constants: 
$\lambda_{i} = \beta_{a} = \alpha = 0$ where  $1\leq i\leq3$ and $1\leq a \leq7$.
Furthermore, since the matrices $\mathscr{A}^{\Sigma}$, $\mathscr{B}^{\Sigma}$, and $\mathscr{C}^{\Sigma}$ all have regular limits as $t\searrow0$, the constants $\texttt{b}$ and $\tilde{\texttt{b}}$, defined in \cite{BOOS:2021}*{Thm 3.8}, both vanish.

\section{Future Stability of Tilted Bianchi I Spacetimes}
\label{sec:MainTheoremProof}

In this section, we establish the nonlinear stability of the zero solution to the Fuchsian system \eqref{eqn:EinsteinEuler_NoDeriv}-\eqref{eqn:EinsteinEuler_2ndDeriv} by applying the global existence theory from \cite{BOOS:2021}. As a consequence, using the continuation principle established in Section \ref{sec:Local_Existence_ConstraintProp}, we prove the future stability of \textit{tilted Bianchi I} solutions to the Einstein-Euler equations. The precise statement of our stability result is given in the following theorem.

\begin{thm}[Future Stability of Tilted Two-Fluid Bianchi I Spacetimes] \label{mainthm}
    Suppose $0<T_1\leq 1$, $\Lambda>0$, $K_{(\afrak)} \in (\frac{1}{3},\frac{5}{7})$, $k \in \mathbb{Z}_{>\frac{3}{2}+1}$,  $\frg^{(\afrak)}_{0} \in \mathbb{R}$, $\frg_{(\afrak)}(t)$ are the unique solutions to the IVP \eqref{eqn:Homogeneous_ODE_1}-\eqref{eqn:Homogeneous_ODE_2}, and $\mathcal{U}_{0} \in H^{k+2}(\mathbb{T}^{3},\mathbb{R}^{45})$. Then, there exists a $\delta >0$ independent of $T_1 \in (0,1]$ such that if
    \begin{align*}
        \|\mathcal{U}_{0}\|^{2}_{H^{k+2}} + \|\mathbf{Q}_{E}\|^{2}_{L^{\infty}(0,T_{1}]} \leq \delta,
    \end{align*}
    there exists a unique solution
    \begin{align*}
        \mathcal{U}  \in C^{0}\big((0,T_{1}],H^{k+2}(\mathbb{T}^{3},\mathbb{R}^{45})\big)\cap C^{1}\big((0,T_{1}],H^{k+1})(\mathbb{T}^{3},\mathbb{R}^{45})\big)
    \end{align*}
    to \eqref{eqn:EinsteinEuler_NotFuchsian} satisfying $\mathcal{U}_{t=T_1}=\mathcal{U}_0$.
    Moreover, the following hold:
    \begin{enumerate}[(a)]
    \item $\mathcal{U}$ satisfies the energy estimate 
        \begin{align*}
    \mathcal{E}(t) + \|\mathbf{Q}_{E}\|^{2}_{L^{\infty}(0,T_{1}]} + \int^{T_{1}}_{t}\frac{1}{\tau}\mathcal{R}(\tau) d\tau \lesssim \|\mathcal{U}_{0}\|^{2}_{H^{k}}+ \|\mathbf{Q}_{E}\|^{2}_{L^{\infty}(0,T_{1}]}
\end{align*}
for all $t \in (0,T_{1}]$, where $\mathcal{U}$ is defined by \eqref{eqn:ufrak_defn},
\begin{align*}
    \mathcal{E}(t) &= \|f^{\Sigma}_{P}(t)\|_{H^{k+1}}^{2} +\|\tilde{\mathcal{H}}(t)\|_{H^{k+1}}^{2}+\|\tilde{\alpha}(t)\|_{H^{k+1}}^{2}+\|\dot{U}_{P}(t)\|_{H^{k+1}}^{2}+\|A_{P}(t)\|_{H^{k+1}}^{2} \nonumber \\
    &+\|N_{IP}(t)\|_{H^{k+1}}^{2}+\|\Sigma_{IP}(t)\|_{H^{k+1}}^{2}   + \|\zeta_{(2)}(t)\|_{H^{k+1}}^{2} +\|\uh_{(2)}(t)\|_{H^{k+1}}^{2} + \|\zeta_{(1)}(t)\|_{H^{k+1}}^{2} \nonumber \\
    &+\|\uh_{(1)}(t)\|_{H^{k+1}}^{2} +t^{2\psi}\Big(\|D^{2}f^{\Sigma}_{P}(t)\|_{H^{k}}^{2} +\|D^{2}\tilde{\mathcal{H}}(t)\|_{H^{k}}^{2}+\|D^{2}\tilde{\alpha}(t)\|_{H^{k}}^{2}\nonumber \\
    &+\|D^{2}\dot{U}_{P}(t)\|_{H^{k}}^{2}+\|D^{2}A_{P}(t)\|_{H^{k}}^{2} +\|D^{2}N_{IP}(t)\|_{H^{k}}^{2}+\|D^{2}\Sigma_{IP}(t)\|_{H^{k}}^{2}\Big) \nonumber \\
    &+t^{2\mu_{(2)}}\Big(\|D^{2}\zeta_{(2)}(t)\|_{H^{k}}^{2} +\|D^{2}\uh_{(2)}(t)\|_{H^{k}}^{2}+ \|w^{(2)}_{P}(t)\|_{H^{k+2}}^{2}\Big) \nonumber \\
    &+ t^{2\mu_{(1)}}\Big(\|D^{2}\zeta_{(1)}(t)\|_{H^{k}}^{2} +\|D^{2}\uh_{(1)}(t)\|_{H^{k}}^{2} + \|w^{(1)}_{P}(t)\|_{H^{k+2}}^{2}\Big), \\
    \intertext{and}
    \mathcal{R}(t) &= \tau^{2\mu_{(1)}}\Big(\|D^{2}\uh_{(1)}(t)\|_{H^{k}}^{2}+\|D^{2}\zeta_{(1)}(t)\|_{H^{k}}^{2}\Big) + \tau^{2\mu_{(2)}}\Big(\|D^{2}\uh_{(2)}(t)\|_{H^{k}}^{2}+\|D^{2}\zeta_{(2)}(t)\|_{H^{k}}^{2}\Big)  \nonumber \\
    &+\|\Sigma_{IP}(t)\|_{H^{k+1}}^{2} +\tau^{2\psi}\Big(\|D^{2}f^{\Sigma}_{P}(t)\|_{H^{k}}^{2} +\|D^{2}\tilde{\mathcal{H}}(t)\|_{H^{k}}^{2}+\|D^{2}\tilde{\alpha}(t)\|_{H^{k}}^{2}+\|D^{2}\dot{U}_{P}(t)\|_{H^{k}}^{2} \nonumber \\
    &+\|D^{2}A_{P}(t)\|_{H^{k}}^{2} +\|D^{2}N_{IP}(t)\|_{H^{k}}^{2} +\|D^{2}\Sigma_{IP}(t)\|_{H^{k}}^{2}\Big).
\end{align*}
    \item There exists a constant $\mathfrak{z}>0$ and functions
${}^{*}\bar{e}^{\Sigma}_{P}$, $\bar{\mathcal{H}}^{*}$, $\bar{\alpha}^{*}$,  $\bar{A}_{P}^{*}$, $\bar{N}_{IP}^{*}\in H^{k-1}(\mathbb{T}^{3})$,
    ${}^{*}f^{\Sigma}_{P}$, $\tilde{\mathcal{H}}^{*}$, $\tilde{\alpha}^{*}$, $A_{P}^{*}$, $N_{IP}^{*}$, $\zeta_{(\afrak)}^{*}$,  $\uh_{(\afrak)}^{*}\in H^{k}(\mathbb{T}^{3})$, and
    ${}^{*}\bar{w}^{(\afrak)}_{P}\in H^{k+1}(\mathbb{T}^{3})$
such that 
\begin{align*}
    \bar{\mathcal{E}}(t) \lesssim t^{\mathfrak{z}-\sigma}
\end{align*}
for all $t \in(0,T_{1}]$ and $\sigma\in (0,\mathfrak{z})$ where 
\begin{align*}
\bar{\mathcal{E}}(t) &= \|f^{\Sigma}_{P}(t)-{}^{*}f^{\Sigma}_{P}\|_{H^{k}}+\|\tilde{\mathcal{H}}(t)-\tilde{\mathcal{H}}^{*}\|_{H^{k}}
+\|\tilde{\alpha}(t)-\tilde{\alpha}^{*}\|_{H^{k}}\nonumber \\
&+\|A_{P}(t)-A_{P}^{*}\|_{H^{k}} +\|N_{IP}(t)-N_{IP}^{*}\|_{H^{k}}+\|\zeta_{(1)}(t)-\zeta_{(1)}^{*}\|_{H^{k}}
\nonumber \\
&+\|\uh_{(1)}(t)-\uh_{(1)}^{*}\|_{H^{k}} +\|\zeta_{(2)}(t)-\zeta_{(2)}^{*}\|_{H^{k}}
+\|\uh_{(2)}(t)-\uh_{(2)}^{*}\|_{H^{k}} \nonumber \\
&+\|t^{\psi}D^{2}f^{\Sigma}_{P}(t)-{}^{*}\bar{e}^{\Sigma}_{P}\|_{H^{k-1}}
+\|t^{\psi}D^{2}\tilde{\mathcal{H}}(t)-\bar{\mathcal{H}}^{*}\|_{H^{k-1}}+\|t^{\psi}D^{2}\alpha(t)-\bar{\alpha}^{*}\|_{H^{k-1}} \nonumber \\
&+\|t^{\psi}D^{2}A_{P}(t)-\bar{A}_{P}^{*}\|_{H^{k-1}}+\|t^{\psi}D^{2}N_{IP}(t)-\bar{N}_{IP}^{*}\|_{H^{k-1}}
+\|t^{\mu_{(1)}}w^{(1)}_{P}(t)-{}^{*}\bar{w}^{(1)}_{P}\|_{H^{k+1}}\nonumber \\
&+\|t^{\mu_{(2)}}w^{(2)}_{P}(t)-{}^{*}\bar{w}^{(2)}_{P}\|_{H^{k+1}}.
\end{align*}
\item $\mathcal{U}$, from the calculations in Sections \ref{sec:Sym_Hyp_Einstein}-\ref{sec:Sym_Hyp_Euler}, determines a unique solution of the Einstein-Euler equations \eqref{eqn:Einstein_Physical}-\eqref{eqn:Euler_Physical} on the spacetime region $(0,T_{1}]\times \mathbb{T}^{3}$.
\end{enumerate}
\noindent The implicit constants in the estimates above and $\mathfrak{z}$ are independent of the choice of $T_1\in (0,1]$.
\end{thm}

\begin{rem}
Since the constant $\delta>0$ from Theorem \ref{mainthm} is independent of the the choice of $T_1\in (0,1]$, it follows from \eqref{eqn:Qbf_E_Bound} that we can satisfy
$\|\mathbf{Q}_{E}\|^{2}_{L^{\infty}(0,T_{1}]} \leq \frac{\delta}{2}$
 by choosing $T_1>0$ sufficiently small. For such a choice of $T_1$, we can then satisfy the smallness condition
$\|\mathcal{U}_{0}\|^{2}_{H^{k+2}} + \|\mathbf{Q}_{E}\|^{2}_{L^{\infty}(0,T_{1}]} \leq \delta$
from Theorem \ref{mainthm}
by choosing initial data $\mathcal{U}_0$ satisfying 
$\|\mathcal{U}_{0}\|^{2}_{H^{k+2}} \leq \frac{\delta}{2}$.
\end{rem}

\begin{rem} As discussed in Appendix \ref{Appendix:two-fluid}, two-fluid Bianchi I spacetimes with a positive cosmological asymptotically approach tilted solutions of the relativistic Euler equations on a fixed de Sitter spacetime. In particular, this means tilted Bianchi I initial data, in terms of our frame variables, can be made as small as we like by taking the initial time $T_{1}>0$ to be sufficiently small. In this sense, tilted Bianchi I models (and any small deviations from them) can be regarded as small perturbations of the trivial solution to \eqref{eqn:EinsteinEuler_NoDeriv}-\eqref{eqn:EinsteinEuler_2ndDeriv}. Moreover, there is no loss of generality by taking $T_{1}$ small. This is because Bianchi I models are global solutions of the Einstein-Euler equations. Thus, by Cauchy stability \cite{Kroon:2017}, we can evolve a perturbation of a reference Bianchi I solution at any time $t_{1}>0$ to obtain a perturbation of this reference solution at time $T_{1}\in(0,t_{1})$. Moreover, the perturbation at $T_{1}$ can be made arbitrarily small by taking the initial perturbation at $t_{1}$ suitably small.
\end{rem}

\begin{rem}[Extreme Tilt]
Using the 3+1 decomposition \eqref{eqn:Tab_3+1_lorentz} and the identities in Appendix \ref{sec:Appendix_Stress-Energy-Expansions}, it is straightforward to verify that the conformal fluid four-velocity is defined by
\begin{align*}
    v^{(\afrak)}_{a} = \Gamma^{(\afrak)} (1, \, \nu_{1}^{(\afrak)}, \, \nu_{2}^{(\afrak)}, \, \nu_{3}^{(\afrak)}),
\end{align*}
where 
\begin{align*}
    \Gamma^{(\afrak)} &= t^{-\mu_{(\afrak)}}\sqrt{t^{2\mu_{(\afrak)}} + e^{2(\uh_{(\afrak)} + \fru_{(\afrak)})}}, \\
    \nu_{A}^{(\afrak)} &= \frac{e^{\uh_{(\afrak)}+\fru_{(\afrak)}}}{\sqrt{t^{2\mu} + e^{2(\uh_{(\afrak)}+\fru_{(\afrak)})}}}\big(t^{\mu_{(\afrak)}}w^{(\afrak)}_{A} + \xi^{(\afrak)}_{A}\big).
\end{align*}
The norm $\nu_{A}^{(\afrak)}\nu^{A}_{(\afrak)}$ is therefore,
\begin{align*}
   \nu_{A}^{(\afrak)}\nu^{A}_{(\afrak)} &=  \frac{e^{2(\uh_{(\afrak)}+\fru_{(\afrak)})}}{t^{2\mu_{(\afrak)}} + e^{2(\uh_{(\afrak)}+\fru_{(\afrak)})}}\big(t^{2\mu_{(\afrak)}}w^{(\afrak)}_{A}w^{A}_{(\afrak)} + 2t^{\mu_{(\afrak)}}w_{A}^{(\afrak)}\xi_{(\afrak)}^{A} + \xi^{A}\xi_{A}^{(\afrak)}\big), \\
   &=  \frac{e^{2(\uh_{(\afrak)}+\fru_{(\afrak)})}}{t^{2\mu_{(\afrak)}} + e^{2(\uh_{(\afrak)}+\fru_{(\afrak)})}}
\end{align*}
where we have used the normalisation constraint for the fluid velocity $\mathcal{C}_{5}$ (cf. Part (c) of Proposition \ref{prop:Local_ExistenceUniqueness_Prop}) to obtain the second equality. Taking the limit as $t \searrow 0$ and using the decay estimates in the statement of Theorem \ref{mainthm}, we find
\begin{align*}
    \lim_{t \searrow 0} \nu_{A}^{(\afrak)}\nu^{A}_{(\afrak)} = 1.
\end{align*}
This, in combination with the definition of the conformal fluid four-velocity \eqref{eqn:Conformal_Tmunu_def}, then implies 
\begin{align*}
   \lim_{t \searrow 0} \tilde{v}^{(\afrak)}_{a}\tilde{v}^{a}_{(\afrak)} = \lim_{t \searrow 0} v^{(\afrak)}_{a}v^{a}_{(\afrak)} = 0.
\end{align*}
Hence, the physical fluid four-velocities develop extreme tilts.
\end{rem}

\begin{proof}[Proof of Theorem \ref{mainthm}.]
Suppose $\delta>0$ and we have initial data $\mathcal{U}_{0}\in H^{k+2}(\mathbb{T}^{3},\mathbb{R}^{45})$, $k \in \mathbb{Z}_{>\frac{3}{2}+1}$ which satisfies the constraint equations \eqref{eqn:C1_constraint}-\eqref{eqn:Hamiltonian_Constraint} and \eqref{eqn:C5_nu_normalisation_constraint} and the bound
 \begin{equation*}
        \|\mathcal{U}_{0}\|^{2}_{H^{k+2}} + \|\mathbf{Q}_{E}\|^{2}_{L^{\infty}(0,T_{1}]} \leq \delta.
\end{equation*}
By Proposition \ref{prop:Local_ExistenceUniqueness_Prop}, there exists a minimal time $T_0\in [0,T_1]$ and a unique solution 
\begin{equation*}
        \mathcal{U}\in C^{0}\big((T_{0},T_{1}],H^{k+2}(\mathbb{T}^{3},\mathbb{R}^{45})\big)\cap C^{1}\big((T_{0},T_{1}],H^{k+1})(\mathbb{T}^{3},\mathbb{R}^{45})\big)
\end{equation*}
of the reduced Einstein-Euler equations \eqref{eqn:EinsteinEuler_NotFuchsian}
satisfying $\mathcal{U}|_{t=T_1}=\mathcal{U}_0$. Moreover, the constraints  \eqref{eqn:C1_constraint}-\eqref{eqn:Hamiltonian_Constraint} and \eqref{eqn:C5_nu_normalisation_constraint} vanish everywhere in $(T_0,T_1]\times \mathbb{T}^3$ and $\mathcal{U}$ determines a solution of 
the Einstein-Euler equations \eqref{eqn:Einstein_Physical}-\eqref{eqn:Euler_Physical} on $(T_{0},T_{1}] \times \mathbb{T}^{3}$. \newline \par  

From the computations carried out in Section \ref{sec:Differentiated_Einstein_Euler_System}, we know that $\mathcal{U}$ determines a solution 
\begin{align}
\label{eqn:Wscr_def}
    \mathscr{W} = (V,\overline{W}_{(\afrak)},V_{\Omega},\overline{W}^{(\afrak)}_{\Omega}, &V_{\Gamma\Omega},\overline{W}^{(\afrak)}_{\Gamma\Omega}) \nonumber \\
    & \in C^{0}\big((T_{0},T_{1}],H^{k+2}(\mathbb{T}^{3},\mathbb{R}^{585})\big)\cap C^{1}\big((T_{0},T_{1}],H^{k}(\mathbb{T}^{3},\mathbb{R}^{585})\big)
\end{align}
of \eqref{eqn:EinsteinEuler_NoDeriv}-\eqref{eqn:EinsteinEuler_2ndDeriv} on $(T_0,T_1]\times \mathbb{T}^3$ satisfying
$\mathscr{W}|_{t=T_1}= \mathscr{W}_{0}$
where
\begin{align*}
    \mathscr{W}_{0} &:= (V,\overline{W}_{(\afrak)},V_{\Omega},\overline{W}^{(\afrak)}_{\Omega}, V_{\Gamma\Omega},\overline{W}^{(\afrak)}_{\Gamma\Omega})|_{t=T_{1}} \in H^{k}(\mathbb{T}^{3},\mathbb{R}^{585})
\end{align*}
and
\begin{align}
\label{eqn:Initial_Data_Bound1}
    \|\mathscr{W}_{0}\|_{H^{k}} + \|\mathbf{Q}_{E}\|^{2}_{L^{\infty}(0,T_{1}]} \lesssim \|\mathcal{U}_{0}\|_{H^{k+2}} + \|\mathbf{Q}_{E}\|^{2}_{L^{\infty}(0,T_{1}]} \leq \delta.
\end{align}

On the other hand, from Section \ref{sec:Coefficient_Assumptions}, we know that the Fuchsian system \eqref{eqn:EinsteinEuler_NoDeriv}-\eqref{eqn:EinsteinEuler_2ndDeriv} satisfies the coefficient assumptions in \cite{BOOS:2021}*{\S 3.4}. Thus, due to the bound \eqref{eqn:Initial_Data_Bound1}, we can apply a time transformed version, see \cite{BOOS:2021}*{\S 3.4}, of Theorem 3.8 from \cite{BOOS:2021} to conclude the existence of a $\delta>0$ such that there exists a unique solution
\begin{align*}
     \mathscr{W}^{*} \in C^{0}_b\big((0,T_{1}],H^{k}(\mathbb{T}^{3},\mathbb{R}^{585})\big) \cap C^{1}\big((0,T_{1}],H^{k-1})(\mathbb{T}^{3},\mathbb{R}^{585})\big)
\end{align*}
of \eqref{eqn:EinsteinEuler_NoDeriv}-\eqref{eqn:EinsteinEuler_2ndDeriv} satisfying $\mathscr{W}^{*}|_{t=T_1}=\mathscr{W}_0$ and the following properties:
\begin{enumerate}
    \item The limit of $\mathbb{P}^{\perp}\mathscr{W}^{*}(t)$ as $t \searrow 0$, denoted $\mathbb{P}^{\perp}\mathscr{W}^{*}(0)$, exists in $H^{k-1}(\mathbb{T}^{3},\mathbb{R}^{585})$.
    \item The solution $\mathscr{W}^{*}$ is bounded by the energy estimate 
    \begin{align}
    \label{eqn:Fuchsian_Thm_EnergyEstimate}
        \|\mathscr{W}^{*}\|^{2}_{H^{k}} + \|\mathbf{Q}_{E}\|^{2}_{L^{\infty}(0,T_{1}]} + \int^{T_{1}}_{t}\frac{1}{\tau}\|\mathbb{P}\mathscr{W}^{*}\|_{H^{k}}^{2}\dd\tau \lesssim \|\mathscr{W}_{0}\|^{2}_{H^{k}} + \|\mathbf{Q}_{E}\|^{2}_{L^{\infty}(0,T_{1}]}.
    \end{align}
    \item There exists a constant $\mathfrak{z}>0$ such that the  solution $\mathscr{W}^{*}$ satisfies the decay estimate
    \begin{align}
    \label{eqn:Fuchsian_Thm_DecayRates}
        \|\mathbb{P}\mathscr{W}^{*}(t)\|^{2}_{H^{k}} \lesssim t^{\mathfrak{z}-\sigma} \;\; \text{and} \;\; \|\mathbb{P}^{\perp}\mathscr{W}^{*}(t)-\mathbb{P}^{\perp}\mathscr{W}^{*}(0)\|_{H^{k-1}}\lesssim t^{\mathfrak{z}-\sigma}.
    \end{align}
    for all $t \in(0,T_{1}]$ and $0<\sigma<\mathfrak{z}$. \newline \par
\end{enumerate}

Recalling that the system \eqref{eqn:EinsteinEuler_NoDeriv}-\eqref{eqn:EinsteinEuler_2ndDeriv} is symmetric hyperbolic, we know from the uniqueness property of solutions to symmetric hyperbolic equations that solutions to \eqref{eqn:EinsteinEuler_NoDeriv}-\eqref{eqn:EinsteinEuler_2ndDeriv} that are generated from the same initial data must be the same. Thus $\mathscr{W}$ and $\mathscr{W}^*$ must agree on their common domain of definition, that is, 
\begin{equation*}
  \mathscr{W}(t,x)=\mathscr{W}^{*}(t,x), \quad \forall\; (t,x) \in (T_{0},T_{1}]\times \mathbb{T}^3.  
\end{equation*} 
This, in turn, implies by the energy estimate \eqref{eqn:Fuchsian_Thm_EnergyEstimate}, and Sobolev's inequality \cite{Rauch:2012}*{Thm. 6.2.1} that 
\begin{align*}
    \|\mathcal{U}(t)\|_{W^{1,\infty}} \lesssim \|\mathcal{U}(t)\|_{H^{k+2}} \leq \|\mathscr{W}^*(t)\|_{H^{k}} \lesssim \|\mathscr{W}_{0}\|_{H^{k}}, \;\; T_{0}< t \leq T_{1}.
\end{align*}
Thus, from the continuation principle from Proposition \ref{prop:Local_ExistenceUniqueness_Prop} and the minimality of $T_{0}$, we conclude that $T_{0}=0$ and $\mathscr{W}(t)= \mathscr{W}^{*}(t)$ for all $t \in (0,T_{1}]$. In particular, it follows that $\mathcal{U}$ determines a solution of 
the Einstein-Euler equations \eqref{eqn:Einstein_Physical}-\eqref{eqn:Euler_Physical} on $(0,T_{1}] \times \mathbb{T}^{3}$. \newline \par

Now, from the energy estimate \eqref{eqn:Fuchsian_Thm_EnergyEstimate} and \eqref{eqn:Fuchsian_Solution_Vector_Defns}, \eqref{eqn:ufrak_defn}, and \eqref{eqn:Wscr_def},  we observe that
\begin{align*}
    \mathcal{E}(t) + \|\mathbf{Q}_{E}\|^{2}_{L^{\infty}(0,T_{1}]} + \int^{T_{1}}_{t}\frac{1}{\tau}\mathcal{R}(\tau) d\tau \lesssim \|\mathcal{U}_{0}\|^{2}_{H^{k}} + \|\mathbf{Q}_{E}\|^{2}_{L^{\infty}(0,T_{1}]} , \;\; 0<t\leq T_{1}.
\end{align*}
where
\begin{align*}
    \mathcal{E}(t) &= \|f^{\Sigma}_{P}(t)\|_{H^{k+1}}^{2} +\|\tilde{\mathcal{H}}(t)\|_{H^{k+1}}^{2}+\|\tilde{\alpha}(t)\|_{H^{k+1}}^{2}+\|\dot{U}_{P}(t)\|_{H^{k+1}}^{2}+\|A_{P}(t)\|_{H^{k+1}}^{2} \nonumber \\
    &+\|N_{IP}(t)\|_{H^{k+1}}^{2}+\|\Sigma_{IP}(t)\|_{H^{k+1}}^{2}   + \|\zeta_{(2)}(t)\|_{H^{k+1}}^{2} +\|\uh_{(2)}(t)\|_{H^{k+1}}^{2} + \|\zeta_{(1)}(t)\|_{H^{k+1}}^{2} \nonumber \\
    &+\|\uh_{(1)}(t)\|_{H^{k+1}}^{2} +t^{2\psi}\Big(\|D^{2}f^{\Sigma}_{P}(t)\|_{H^{k}}^{2} +\|D^{2}\tilde{\mathcal{H}}(t)\|_{H^{k}}^{2}+\|D^{2}\tilde{\alpha}(t)\|_{H^{k}}^{2}\nonumber \\
    &+\|D^{2}\dot{U}_{P}(t)\|_{H^{k}}^{2}+\|D^{2}A_{P}(t)\|_{H^{k}}^{2} +\|D^{2}N_{IP}(t)\|_{H^{k}}^{2}+\|D^{2}\Sigma_{IP}(t)\|_{H^{k}}^{2}\Big) \nonumber \\
    &+t^{2\mu_{(2)}}\Big(\|D^{2}\zeta_{(2)}(t)\|_{H^{k}}^{2} +\|D^{2}\uh_{(2)}(t)\|_{H^{k}}^{2}+ \|w^{(2)}_{P}(t)\|_{H^{k+2}}^{2}\Big) \nonumber \\
    &+ t^{2\mu_{(1)}}\Big(\|D^{2}\zeta_{(1)}(t)\|_{H^{k}}^{2} +\|D^{2}\uh_{(1)}(t)\|_{H^{k}}^{2} + \|w^{(1)}_{P}(t)\|_{H^{k+2}}^{2}\Big), \\
    \intertext{and}
    \mathcal{R}(t) &= \tau^{2\mu_{(1)}}\Big(\|D^{2}\uh_{(1)}(t)\|_{H^{k}}^{2}+\|D^{2}\zeta_{(1)}(t)\|_{H^{k}}^{2}\Big) + \tau^{2\mu_{(2)}}\Big(\|D^{2}\uh_{(2)}(t)\|_{H^{k}}^{2}+\|D^{2}\zeta_{(2)}(t)\|_{H^{k}}^{2}\Big)  \nonumber \\
    &+\|\Sigma_{IP}(t)\|_{H^{k+1}}^{2} +\tau^{2\psi}\Big(\|D^{2}f^{\Sigma}_{P}(t)\|_{H^{k}}^{2} +\|D^{2}\tilde{\mathcal{H}}(t)\|_{H^{k}}^{2}+\|D^{2}\tilde{\alpha}(t)\|_{H^{k}}^{2}+\|D^{2}\dot{U}_{P}(t)\|_{H^{k}}^{2} \nonumber \\
    &+\|D^{2}A_{P}(t)\|_{H^{k}}^{2} +\|D^{2}N_{IP}(t)\|_{H^{k}}^{2} +\|D^{2}\Sigma_{IP}(t)\|_{H^{k}}^{2}\Big).
\end{align*}
Additionally, from the decay estimates \eqref{eqn:Fuchsian_Thm_DecayRates}, we obtain the existence of 
functions
${}^{*}\bar{e}^{\Sigma}_{P}$, $\bar{\mathcal{H}}^{*}$, $\bar{\alpha}^{*}$,  $\bar{A}_{P}^{*}$, $\bar{N}_{IP}^{*}\in H^{k-1}(\mathbb{T}^{3})$,
    ${}^{*}f^{\Sigma}_{P}$, $\tilde{\mathcal{H}}^{*}$, $\tilde{\alpha}^{*}$, $A_{P}^{*}$, $N_{IP}^{*}$, $\zeta_{(\afrak)}^{*}$,  $\uh_{(\afrak)}^{*}\in H^{k}(\mathbb{T}^{3})$, and
    ${}^{*}\bar{w}^{(\afrak)}_{P}\in H^{k+1}(\mathbb{T}^{3})$
such that 
\begin{align*}
    \bar{\mathcal{E}}(t) \lesssim t^{\mathfrak{z}-\sigma}
\end{align*}
for all $t \in(0,T_0]$, where
\begin{align*}
\bar{\mathcal{E}}(t) &:= \|f^{\Sigma}_{P}(t)-{}^{*}f^{\Sigma}_{P}\|_{H^{k}}+\|\tilde{\mathcal{H}}(t)-\tilde{\mathcal{H}}^{*}\|_{H^{k}}
+\|\tilde{\alpha}(t)-\tilde{\alpha}^{*}\|_{H^{k}}\nonumber \\
&+\|A_{P}(t)-A_{P}^{*}\|_{H^{k}} +\|N_{IP}(t)-N_{IP}^{*}\|_{H^{k}}+\|\zeta_{(1)}(t)-\zeta_{(1)}^{*}\|_{H^{k}}
\nonumber \\
&+\|\uh_{(1)}(t)-\uh_{(1)}^{*}\|_{H^{k}} +\|\zeta_{(2)}(t)-\zeta_{(2)}^{*}\|_{H^{k}}
+\|\uh_{(2)}(t)-\uh_{(2)}^{*}\|_{H^{k}} \nonumber \\
&+\|t^{\psi}D^{2}f^{\Sigma}_{P}(t)-{}^{*}\bar{e}^{\Sigma}_{P}\|_{H^{k-1}}
+\|t^{\psi}D^{2}\tilde{\mathcal{H}}(t)-\bar{\mathcal{H}}^{*}\|_{H^{k-1}}+\|t^{\psi}D^{2}\alpha(t)-\bar{\alpha}^{*}\|_{H^{k-1}} \nonumber \\
&+\|t^{\psi}D^{2}A_{P}(t)-\bar{A}_{P}^{*}\|_{H^{k-1}}+\|t^{\psi}D^{2}N_{IP}(t)-\bar{N}_{IP}^{*}\|_{H^{k-1}}
+\|t^{\mu_{(1)}}w^{(1)}_{P}(t)-{}^{*}\bar{w}^{(1)}_{P}\|_{H^{k+1}}\nonumber \\
&+\|t^{\mu_{(2)}}w^{(2)}_{P}(t)-{}^{*}\bar{w}^{(2)}_{P}\|_{H^{k+1}},
\end{align*}
which completes the proof.
\end{proof}

\section{Existence of Perturbed Initial Data}
\label{sec:PerturbedInitialData}
In this section, we prove the existence of an open set of perturbed initial data around spatially homogeneous and tilted solutions of the constraint equations. The argument presented here follows a standard implicit function theorem argument, see for example \cites{BruhatDeser:1973,D'eath:1976,Brauer:1991}. In particular, our proof follows using a variation of the arguments from \cites{Oliynyk:CMP_2009,Oliynyk:CMP_2010,Oliynyk:JHDE_2010}. \newline \par 

To begin, we fix a time $T_1>0$ and consider the initial hypersurface 
\begin{equation*}
\Sigma=\{T_1\}\times \mathbb{T}^3\cong \mathbb{T}^3.
\end{equation*} 
On this hypersurface, we consider the $3+1$ decomposition of the constraints (see \cite{Gourgoulhon:2012}*{Ch.~9}) for the \textit{physical} Einstein equations given by
\eqref{eqn:Einstein_Physical}
\begin{align*}
    R[h] + K^{2} - K_{\Gamma\Omega}K^{\Gamma\Omega} = 2E,\\
    D_{\Omega}\tensor{K}{^\Omega_\Gamma} - D_{\Gamma}K = J_{\Gamma}, 
\end{align*}
where $R[h]$ is the Ricci scalar of the spatial metric $h_{\Omega\Gamma}$, $K_{\Omega\Gamma}$ is the extrinsic curvature, $D_{\Omega}$ is the covariant derivative of the spatial metric, and
\begin{align}
  E &= n^{\mu}n^{\nu}T_{\mu\nu} = \big(\Gamma_{(1)}^{2}(K_{(1)}+1)- K_{(1)}\big)\rho_{(1)} +\big(\Gamma_{(2)}^{2}(K_{(2)}+1)- K_{(2)}\big)\rho_{(2)} +\Lambda, \label{E-def} \\
  J_{\Gamma} &= n^{\mu}h^{\nu}_{\Gamma}T_{\mu\nu} = (K_{(1)}+1)\Gamma_{(1)}^{2}\rho_{(1)}w^{(1)}_{\Gamma} + (K_{(2)}+1)\Gamma_{(2)}^{2}\rho_{(2)}w^{(2)}_{\Gamma}\label{J-def},
\end{align} 
where 
\begin{align*}
    v_{(\afrak)}^{\mu} &= \Gamma(n^{\mu} + w_{(\afrak)}^{\mu}), \quad \Gamma_{(\afrak)}^{2} = \frac{1}{1-|w_{(\afrak)}|_h^2}, \nonumber \\ |w_{(\afrak)}|_h^2 &= h_{\Omega\Gamma}w_{(\afrak)}^{\Omega}w_{(\afrak)}^{\Gamma} < 1, \quad
    n_{\mu} = (-N,0,0,0),
\end{align*}
and $N$ is the lapse of the physical metric.
To obtain the existence of perturbed initial data, we decompose the constraints into an appropriate form. To start, we decompose the extrinsic curvature into its trace and trace-free parts
\begin{align*}
    K_{\Omega\Gamma} = \hat{K}_{\Omega\Gamma} + \frac{1}{3}Kh_{\Omega\Gamma}.
\end{align*}

Next, following \cite{D'eath:1976}, we introduce the conformal decomposition 
\begin{align}
\label{eqn:Constraints_conformal_decompositions}
    h^{\Omega\Gamma} &= \Phi^{-4}j^{\Omega\Gamma}, \nonumber \\
    \hat{K}_{\Omega\Gamma} &= \Phi^{4}H_{\Omega\Gamma}, \nonumber \\
    \nu^{\Omega} &= \Phi^{2}w^{\Omega}, \\
    H_{\Omega\Gamma} &= \frac{1}{2}(LX)_{\Omega\Gamma}  + A_{\Omega\Gamma}, \nonumber \\ 
    (LX)_{\Omega\Gamma} &= \Dt_{\Omega}X_{\Gamma} + \Dt_{\Gamma}X_{\Omega} - \frac{2}{3}j_{\Gamma\Omega}\Dt_{\Sigma}X^{\Sigma}, \nonumber
\end{align}
where $\Dt$ is the covariant derivative associated with the conformal metric $j_{\Omega\Gamma}$, $A_{\Gamma\Omega}$ is a trace-free symmetric two-tensor, and conformal objects have their indices raised and lowered with $j_{\Omega\Gamma}$.
We now express the constraints in terms of $\Phi$, $X_{\Omega}$, $j_{\Omega\Gamma}$, $H_{\Omega\Gamma}$, and $A_{\Gamma\Omega}$ and $K$ by using the identities
\begin{align*}
    D_{\Omega}\tensor{\hat{K}}{^{\Gamma}_{\Sigma}} &= \Dt\tensor{\hat{K}}{^{\Gamma}_{\Sigma}} + \tensor{C}{^\Gamma_{\Omega\Theta}}\tensor{\hat{K}}{^{\Theta}_{\Sigma}} - \tensor{C}{^\Theta_{\Omega\Sigma}}\tensor{\hat{K}}{^{\Gamma}_{\Theta}}, \\
    \tensor{C}{^\Sigma_{\Omega\Gamma}} &= \delta^{\Sigma}_{\Omega}\Dt_{\Gamma}\log(\Phi^{2}) + \delta^{\Sigma}_{\Gamma}\Dt_{\Omega}\log(\Phi^{2}) - j_{\Omega\Gamma}j^{\Sigma\Theta}\Dt_{\Theta}\log(\Phi^{2}),
\end{align*}
to write the constraints as 
\begin{align}
\label{eqn:Hamiltonian_Elliptic}
&8\Dt_{\Gamma}\Dt^{\Gamma}\Phi  - \Phi R[j] + (H_{\Omega\Gamma}H^{\Omega\Gamma} - \frac{2}{3}K^{2} +2E)\Phi^{5} = 0, \\
\label{eqn:Momentum_Elliptic}
&\frac{1}{2}\Dt_{\Omega}\Dt^{\Omega}X_{\Gamma} + \frac{1}{6}\Dt_{\Gamma}\Dt^{\Omega}X_{\Omega} +\frac{1}{2}R[j]_{\Omega\Gamma}X^{\Omega} + \Dt_{\Omega}\tensor{A}{^\Omega_{\Gamma}} \nonumber \\ &-\frac{2}{3}\Dt_{\Gamma}K 
+ 6\Phi^{-1}\tensor{H}{^\Omega_\Gamma}\Dt_{\Omega}\Phi -\Phi^{-2}J_{\Gamma} = 0.
\end{align}
Observe that the operators $\Dt_{\Gamma}\Dt^{\Gamma}\Phi$ and $\Dt_{\Omega}\Dt^{\Omega}X_{\Gamma} + \frac{1}{3}\Dt_{\Gamma}\Dt^{\Omega}X_{\Omega}$ reduce to the flat Laplacian and conformal vector Laplacian, respectively, when $j_{\Omega\Gamma}=\delta_{\Omega\Gamma}$.

\subsection{Spatially Homogeneous Solution}
\label{sec:SH_Background_Constraint_Soln}
We now consider the spatially homogeneous constraint equations by setting
\begin{align} \label{hom-fix}
X_{\Omega} = 0, \;\; \Phi = 1, \;\;  j_{\Omega\Gamma} = \delta_{\Omega\Gamma}, \;\;  A_{\Omega\Gamma} = 0.
\end{align}
With these choices, \eqref{eqn:Hamiltonian_Elliptic}-\eqref{eqn:Momentum_Elliptic} reduce to
\begin{align*}
    J_{\Gamma} &= 0, \\
    2E - \frac{2}{3}K^{2} &= 0.
\end{align*}
Using \eqref{E-def}-\eqref{J-def} and \eqref{eqn:Constraints_conformal_decompositions}, these become
\begin{align}
\label{eqn:Homogeneous_ID_Condition1}
\Big(\Gamma^{2}_{(1)}(K_{(1)}+1)\rho_{(1)}\nu^{\Omega}_{(1)} + \Gamma^{2}_{(2)}(K_{(2)}+1)\rho_{(2)}\nu^{\Omega}_{(2)}\Big) &= 0, \\
\label{eqn:Homogeneous_ID_Condition2}
2\Big((\Gamma^{2}_{(1)}(K_{(1)}+1)-K_{(1)})\rho_{(1)} +(\Gamma^{2}_{(2)}(K_{(2)}+1)-K_{(2)})\rho_{(2)} +\Lambda \Big) - \frac{2}{3}K^{2} &= 0.
\end{align}
To solve these equations, we first choose a positive constant $\rho_{(2)}>0$ and a vector $\nu_{(2)}=(\nu^{\Omega}_{(2)})\in\Rbb^3$ satisfying 
$0<|\nu_{(2)}|<1$ where $|\nu_{(\afrak)}|^2 = \delta_{\Omega\Gamma}\nu_{(\afrak)}^\Omega\nu_{(\afrak)}^\Gamma$.
Then from \eqref{eqn:Homogeneous_ID_Condition1}-\eqref{eqn:Homogeneous_ID_Condition2}, we find that
\begin{align}
\label{eqn:nu1_background_value}
    \nu_{(1)}^{\Omega} = -c\nu_{(2)}^{\Omega}
\end{align}
where 
\begin{align}
    c = \frac{-(1+K_{(1)})\rho_{(1)}(1-|\nu_{(2)}|^2) + \sqrt{(1+K_{(1)})^{2}\rho_{(1)}^{2}(1-|\nu_{(2)}|^2)^{2} + 4(1+K_{(2)})^{2}\rho_{(2)}^{2}|\nu_{(2)}|^{2}}}{2(1+K_{(2)})\rho_{(2)}|\nu_{(2)}|^2} \label{c-def}
\end{align}
and 
\begin{align}
\label{eqn:rho1_background_value}
    \rho_{(1)} &= \frac{\frac{1}{3}K^{2} - (\Gamma^{2}_{(2)}(K_{(2)}+1)-K_{(2)})\rho_{(2)} - \Lambda}{(\Gamma^{2}_{(1)}(K_{(1)}+1)-K_{(1)})}.  
\end{align}
To ensure that $\rho_{(1)}>0$, we require that the constant $K$ is chosen so that
\begin{align}
\label{eqn:MeanCurvature_Background_Constraint}
     K &> 3\sqrt{(\Gamma^{2}_{(2)}(K_{(2)}+1)-K_{(2)})\rho_{(2)} + \Lambda}.
\end{align}}
We further observe from \eqref{c-def} that $0<c\leq1$, which by \eqref{eqn:nu1_background_value} implies $0<|\nu_{(1)}|<1$. \newline \par 

To summarise, for each choice of
\begin{equation*}
(\rho_{(2)},\nu_{(2)}^\Omega,K) \in \Rbb \times \Rbb^3\times \Rbb
\end{equation*}
satisfying $\rho_{(2)}>0$, $0<|\nu_{(2)}|<1$, and the inequality \eqref{eqn:MeanCurvature_Background_Constraint}, the formulas  \eqref{hom-fix}, \eqref{eqn:nu1_background_value}, \eqref{c-def} and
\eqref{eqn:rho1_background_value} determine a homogeneous solution of the constraint equations \eqref{eqn:Homogeneous_ID_Condition1}-\eqref{eqn:Homogeneous_ID_Condition2}.

\subsection{Constraint Map}
To find inhomogeneous solutions to the constraint equation, we split the initial data into their values on the homogeneous background solution and inhomogeneous perturbations as follows:
\begin{align}
\label{eqn:Phi_constraintmap_decomp}
\Phi &= 1+ \tilde{\Phi}, \\
j^{\Omega\Gamma} &= \delta^{\Omega\Gamma} + \gamma^{\Omega\Gamma}, \\
K &= \check{K} + \tilde{K}, \\
\rho_{(1)} &= \check{\rho}_{(1)} + \tilde{\rho}_{(1)}, \\
\rho_{(2)} &= \check{\rho}_{(2)} + \bar{\rho}_{(2)}, \\
\nu^{\Omega}_{(1)} &= \check{\nu}^{\Omega}_{(1)} 
 + \tilde{\nu}^{\Omega}_{(1)}, \\
 \label{eqn:nu_(2)_constraintmap_decomp}
\nu^{\Omega}_{(2)} &= \check{\nu}^{\Omega}_{(2)} + \bar{\nu}^{I}_{(2)},
\end{align}
where the background values $\check{\rho}_{(2)} >0$ and $\check{\nu}_{(2)}^{\Omega}$ are a constant and a non-zero constant vector, respectively, $\check{\nu}_{(2)}$, $\check{\rho}_{(1)}$ are defined by \eqref{eqn:nu1_background_value}-\eqref{eqn:rho1_background_value}, and $\check{K}$ is a positive constant which satisfies \eqref{eqn:MeanCurvature_Background_Constraint}. The remaining variables denote perturbations of the homogeneous background.  We can further decompose the inhomogeneous perturbations $\bar{\rho}_{(2)}$ and $\bar{\nu}_{(2)}^{\Omega}$ as follows
\begin{equation}
\begin{aligned}
 \label{eqn:bar_rho2_nu2_constraintmap_decomp}
\bar{\rho}_{(2)} &= \Pi^{\perp}\bar{\rho}_{(2)} + \Pi\bar{\rho}_{(2)} = \mathring{\rho}_{(2)} +  \tilde{\rho}_{(2)}, \\
\bar{\nu}^{I}_{(2)} &= \Pi^{\perp}\bar{\nu}^{I}_{(2)} + \Pi\bar{\nu}^{I}_{(2)} = \mathring{\nu}^{\Omega}_{(2)} + \tilde{\nu}^{\Omega}_{(2)},
\end{aligned}
\end{equation}
where, following \cite{Oliynyk:JHDE_2010}, we have introduced the projection operators\footnote{Here $\langle u,v\rangle = \frac{1}{(2\pi)^3}\int_{\Tbb^3} u(x)v(x)d^3 x$ where $\Tbb^3 = [0,2\pi]^3/ \sim$ with $\sim$ the identification of the opposite sides of the box $[0,2 \pi]^3$.}
\begin{align*}
\Pi u = u - \langle1,u\rangle_{L^{2}}1, \quad 
\Pi^{\perp} u  = (1-\Pi)u = \langle1, u\rangle_{L^{2}}1,
\end{align*}
and set 
\begin{equation}
\begin{aligned}
\label{eqn:constraintperturbation_decomposition}
    \mathring{\rho}_{(2)} &= \Pi^{\perp} \bar{\rho}_{(2)}, \;\; \tilde{\rho}_{(2)} = \Pi \bar{\rho}_{(2)}, \;\;
   \mathring{\nu}^{\Omega}_{(2)} = \Pi^{\perp} \bar{\nu}^{I}_{(2)}, \;\; \tilde{\nu}^{\Omega}_{(2)} = \Pi \bar{\nu}^{I}_{(2)}.
\end{aligned}
\end{equation}

The projection operator $\Pi$ can be used to orthogonally decompose the Sobolev space $H^s(\Tbb^3)$ as
\begin{equation*}
H^s(\Tbb^3) = \bar{H}^{s}(\Tbb^3)\oplus \Pi^\perp H^s(\Tbb^3) \cong \bar{H}^{s}(\Tbb^3)\times \Rbb
\end{equation*}
where
\begin{align*}
    \bar{H}^{s}(\Tbb^3) := \Pi H^{s}(\Tbb^3) = \{ \, u\in H^s(\Tbb^3)\, : \, \Pi^\perp u=0 \, \} 
\end{align*}
is the zero average subspace. Expanding $u\in H^s(\Tbb^3)$ in a Fourier series
\begin{equation} \label{Fourier-exp-u}
u = \sum_{k\in\Zbb^3} \widehat{u}(k) e^{i k_\Omega x^\Omega},
\end{equation}
it is clear that we can express the zero average space as
\begin{align*}
    \bar{H}^{s}(\Tbb^3) = \bigl\{\, u\in H^{s}(\Tbb^3) \, : \, \widehat{u}(0) = 0\,\bigr\}.
\end{align*}

After substituting the decompositions \eqref{eqn:Phi_constraintmap_decomp}-\eqref{eqn:constraintperturbation_decomposition} into \eqref{eqn:Hamiltonian_Elliptic}-\eqref{eqn:Momentum_Elliptic}, the constraint equations become
\begin{align}
\label{eqn:perturbed_hamiltonian_constraint}
\mathcal{C}_{H}:=& 8\Dt_{\Gamma}\Dt^{\Gamma}\tilde{\Phi}  - (1+\tilde{\Phi})R[\delta + \gamma] + (H_{\Omega\Gamma}H^{\Omega\Gamma} - \frac{2}{3}(\check{K}+\tilde{K})^{2} +2E)(1+\tilde{\Phi})^{5} = 0, \\
\label{eqn:perturbed_momentum_constraint}
\mathcal{C}_{M}:=&\frac{1}{2}\Dt_{\Omega}\Dt^{\Omega}X_{\Gamma} + \frac{1}{6}\Dt_{\Gamma}\Dt^{\Omega}X_{\Omega} +\frac{1}{2}R[\delta+\gamma]_{\Omega\Gamma}X^{\Omega} + \Dt^{\Omega}\tensor{A}{_{\Omega\Gamma}} -\frac{2}{3}\Dt_{\Gamma}\tilde{K} \nonumber \\
&+ 6(1+\tilde{\Phi})^{-1}\tensor{H}{_\Omega_\Gamma}\Dt^{\Omega}\tilde{\Phi}  -(1+\tilde{\Phi})^{-2}J_{\Gamma} = 0.
\end{align}
We then define the \textit{constraint map} $\mathcal{F}$ by
\begin{equation*}
    \mathcal{F}(Y,Z) = ( \mathcal{C}_{H}, \mathcal{C}_{M} )
\end{equation*}
where 
\begin{equation*}
Y = (\tilde{\rho}_{(1)},\tilde{\rho}_{(2)},\tilde{\nu}^{\Omega}_{(1)},\tilde{\nu}^{\Omega}_{(2)}, A_{\Omega\Gamma}, \gamma_{\Omega\Gamma}, \tilde{K}) \quad \text{and} \quad 
Z = (\tilde{\Phi}, X^{\Omega},\mathring{\rho}_{(2)}, \mathring{\nu}^{\Omega}_{(2)}).
\end{equation*}
Here $Y$ and $Z$ denote the \textit{free} and \textit{determined} data, respectively. \newline \par 

Clearly when $\mathcal{F}(Y,Z)=0$ we have a solution to the gravitational constraints \eqref{eqn:Hamiltonian_Elliptic}-\eqref{eqn:Momentum_Elliptic}. In particular, the trivial solution
\begin{equation*}
\mathcal{F}(0,0) = 0 
\end{equation*}
corresponds to the homogeneous background solution.

\subsection{Existence of Solutions to the Constraint Equations}
We are now in a position to use an implicit function theorem argument to establish the existence of an open neighbourhood of the trivial free data $Y=0$ such that for every choice of free data $Y$ in this neighbourhood there exists corresponding determined data $Z=f(Y)$ such that $(Y,Z)=(Y,f(Y))$ solves the constraint equations \eqref{eqn:perturbed_hamiltonian_constraint}-\eqref{eqn:perturbed_momentum_constraint}, i.e.~$\mathcal{F}(Y,f(Y))=0$. The precise statement of this result is given in the following theorem. 

\begin{thm}[Existence of Perturbed Initial Data]
Suppose $s\in \Zbb_{\geq 1}$ and let
\begin{align*}
    E &= H^{s}(\mathbb{T}^{3})\times  \bar{H}^{s}(\mathbb{T}^{3}) \times  B_{R}\big(H^{s+1}(\mathbb{T}^{3},\mathbb{R}^{3})\big) \times  B_{R}\big(\bar{H}^{s+1}(\mathbb{T}^{3},\mathbb{R}^{3})\big) \nonumber \\
    &\quad \; \times H^{s+1}(\mathbb{T}^{3},\mathbb{S}^{\text{TF}}_{3}) \times B_{R}\big(H^{s+2}(\mathbb{T}^{3},\mathbb{S}_{3})\big) \times H^{s+1}(\mathbb{T}^{3}), \\
    F &= B_{R}\big(\bar{H}^{s+2}(\mathbb{T}^{3},\mathbb{R})\big) \times\bar{H}^{s+2}(\mathbb{T}^{3},\mathbb{R}^{3})\times\mathbb{R}\times B_{R}(\mathbb{R}^{3}), \\
    G &= H^{s}(\mathbb{T}^{3},\mathbb{R}^{3})\times H^{s}(\mathbb{T}^{3},\mathbb{R}^{3}), 
\end{align*}
where $\mathbb{S}_{3}$ and $\mathbb{S}^{TF}_{3}$ denote the sets of symmetric and symmetric trace-free $3\times3$ real matrices, respectively.
Then there exists a $R>0$ such that constraint map
\begin{align*}
    \mathcal{F}\: :\: E \times F \longrightarrow G \: :\: (Y,Z)\longmapsto \mathcal{F}(Y,Z) 
\end{align*}
is well-defined and analytic. Moreover,  
there exists open neighbourhoods $U \subset E$ and $V\subset F$ where $(0,0)\in U\times V$, and an analytic map $\mathcal{G}\, :\, U \longrightarrow V$ such that $\mathcal{F}(Y,\mathcal{G}(Y))= 0$ for all $Y \in U$. 
\end{thm}
\begin{proof}$\;$
To prove the theorem, we will show, for some $R>0$, that the map $\mathcal{F}\,:\, E\times F\longrightarrow G$ is analytic, and that that the linear operator $D_{2}\mathcal{F}(0,0)$ is an isomorphism. Once that is accomplished, the proof will directly follow from an application of the analytic version of the implicit function theorem, e.g. see \cite{Deimling:1985}*{Thm.~15.3}. 

\subsubsection*{Analyticity of the Constraint Map}
To prove that the constraint map $\mathcal{F}$ is analytic, it is enough to show each of the individual terms that define it are analytic. Since the proof of analyticity for each term is essentially the same, we only provide the details for the the representative term
$\tilde{D}^{\Omega}\tensor{A}{_\Omega_\Gamma}$
from \eqref{eqn:perturbed_momentum_constraint}.
Now, expanding the covariant derivative in terms of partial derivatives and Christoffel symbols allows us to express $\tilde{D}^{\Omega}\tensor{A}{_\Omega_\Gamma}$ as
\begin{align}
\label{eqn:cd_A_expanded}
\tilde{D}^{\Omega}\tensor{A}{_\Omega_\Gamma}=j^{\Omega\Theta}\Big(\del_{\Theta}A_{\Omega\Gamma} - \tensor{\Gamma}{_\Theta^\Sigma_\Omega}A_{\Sigma\Gamma} -  \tensor{\Gamma}{_\Theta^\Sigma_\Gamma}A_{\Theta\Sigma}\Big).
\end{align}

We begin the analysis of \eqref{eqn:cd_A_expanded} by observing, for suitably small $R>0$, that the matrix inversion map
\begin{align}\label{mat-inv-map-A}
B_{R}(\mathbb{S}^{3})   \ni \gamma_{\Omega\Theta} \longmapsto j^{\Omega\Theta}:=(\delta_{\Omega\Theta} + \gamma_{\Omega\Theta})^{-1} \in \mathbb{S}^{3}
\end{align}
is analytic. Recalling that pointwise multiplication 
\begin{align*}
    H^{s_{1}}(\Tbb^3)\times H^{s_{2}}(\Tbb^3)\ni (u,v) \longmapsto uv \in H^{s_{3}}(\Tbb^3)
\end{align*}
is bilinear and continuous, and hence analytic, provided $s_{1}, \; s_{2} \geq s_{3}$ and $s_{1}+s_{2}-s_{3}>\frac{3}{2}$, e.g.~see \cite{Choquet_et_al:2000}*{\S VI.3}, it follows, in particular, that matrix multiplication 
\begin{align*}
    H^{s+2}(\Tbb^3,\mathbb{S}^{3})\times H^{s+2}(\Tbb^3,\mathbb{S}^{3})\ni (u,v) \longmapsto uv \in H^{s+2}(\Tbb^3,\mathbb{S}^{3})
\end{align*}
is continuous since $s\in \Zbb_{\geq 1}$ by assumption. From this and the analyticity of the matrix inversion map \eqref{mat-inv-map-A}, we conclude from \cite{Heilig:1995}*{Prop.~3.6}, shrinking $R>0$ if necessary, that the inversion map
\begin{align}\label{mat-inv-map-B}
     B_{R}(H^{s+1}(\mathbb{T}^{3},\mathbb{S}^{3})) \ni \gamma_{\Omega\Theta} \longmapsto j^{\Omega\Theta} \in H^{s+1}(\mathbb{T}^{3},\mathbb{S}^{3})
\end{align}
is analytic. \newline \par 

Next, we note that partial differentiation
\begin{equation} \label{dA-map}
\partial_{\Omega}\: : \: H^{s_1+1}(\Tbb^3)\longrightarrow H^{s_1}(\Tbb^3)
\end{equation}
is linear and bounded, and hence analytic. Then expanding the Christoffel symbols $\Gamma_{\Theta}{}^\Sigma{}_\Omega$ in
\begin{equation*}
\Gamma_{\Theta}{}^\Sigma{}_\Omega =  \frac{1}{2}j^{\Sigma\Delta}(\del_{\Theta}\gamma_{\Delta\Omega} + \del_{\Omega}\gamma_{\Theta\Delta} - \del_{\Delta}\gamma_{\Theta\Omega}),
\end{equation*}
it follows easily from the analyticity of the differentiation map \eqref{dA-map} ($s_1=s+1$), the analyticity of the inversion map \eqref{mat-inv-map-B}, the analyticity of the multiplication map ($s_1=s+2$, $s_2=s+1$ $s_3=s+1$), and the analyticity of the addition map, that the map 
\begin{equation*}
B_R\bigl(H^{s+2}(\Tbb^3,\mathbb{S}^{3})\bigr)\ni  \gamma_{\Omega\Theta} \longmapsto \Gamma_{\Theta}{}^\Sigma{}_\Omega \in H^{s+1}(\Tbb^3)
\end{equation*}
is analytic since it can be expressed as a composition of these analytic maps. Using this and similar arguments, it then follows from \eqref{eqn:cd_A_expanded} that the map 
\begin{equation*}
B_R\bigl(H^{s+2}(\Tbb^3,\mathbb{S}^{3})\bigr)\times H^{s+1}(\Tbb^3,\mathbb{S}^{3})\ni  \gamma_{\Omega\Theta} \longmapsto \tilde{D}^{\Omega}\tensor{A}{_\Omega_\Gamma} \in H^{s}(\Tbb^3,\Rbb^3)
\end{equation*}
is analytic. \newline \par 

By analysing the rest of the terms in the constraint map, see \eqref{eqn:perturbed_hamiltonian_constraint}-\eqref{eqn:perturbed_momentum_constraint}, it is then not difficult to verify, for $R>0$ chosen sufficiently small, that the constraint map
\begin{align*}
    \mathcal{F}\: :\: E \times F \longrightarrow G \: :\: (Y,Z)\longmapsto \mathcal{F}(Y,Z) 
\end{align*}
is well-defined and analytic, where the spaces $E$, $F$, and $G$ are as defined in the statement of the theorem. 

\subsubsection*{Invertibility of the Flat Laplacian and the Flat Conformal Vector Laplacian}
In order to show that $D_{2}\mathcal{F}(0,0)$ is an isomorphism, it will be necessary to establish that the flat Laplacian and flat conformal vector Laplacian define linear isomorphisms between the zero average spaces $\bar{H}^{s+2}(\Tbb^3)$ to $\bar{H}^{s}(\Tbb^3)$. \newline \par   

To this end, let us first consider the equation
\begin{align}\label{Poisson}
    \Delta u  = f.
\end{align}
where $\Delta u =\delta^{\Omega\Gamma}\del_{\Omega}\del_{\Gamma}$
is the flat Laplacian. Expanding $u$ in a Fourier series as \eqref{Fourier-exp-u} and $f$ as
\begin{equation*} 
f = \sum_{k\in\Zbb^3} \widehat{f}(k) e^{i k_\Omega x^\Omega},
\end{equation*}
and substituting these into \eqref{Poisson} implies that $-|k|^{2}\widehat{u}(k) = \widehat{f}(k)$.
Now, if $f\in \bar{H}^{s}(\Tbb)$, then we can solve this equation to get 
\begin{equation}\label{Poisson-sol}
\widehat{u}(k) = -\frac{1}{|k|^2}\widehat{f}(k), \quad k\in \Zbb^3\setminus\{0\}.
\end{equation}
Using the fact that we can express the $H^s(\Tbb^3)$ norm as
\begin{align*}
    \|f\|_{\bar{H}^{s}}^2 = \sum_{k\in\mathbb{Z}^{3}\setminus  \{0\}}|k|^{2s}|\widehat{f}(k)|^{2},
\end{align*}
we get from \eqref{Poisson-sol} that
$\|u\|_{\bar{H}^{s+2}} = \|f\|_{\bar{H}^{s}}$.
From this equality, we deduce that the flat Laplacian $\Delta : \bar{H}^{s+2}(\Tbb^3)\rightarrow\bar{H}^{s}(\Tbb^3)$ is surjective and has trivial kernel, and hence an isomorphism. \newline \par 

Next, consider the equation 
\begin{align}
\label{eqn:vector_poisson}
    \tDLd X^{\Omega} = f^{\Omega}
\end{align}
where 
\begin{align*}
\tDLd X^{\Omega} 
= \Bigl(\delta^{\Gamma \Theta}\delta^{\Omega}_{\Sigma}+\frac{1}{3}\delta^{\Omega \Gamma}\delta^{\Theta}_{\Sigma}\Bigr)\del_{\Gamma}\del_{\Theta}X^{\Sigma}
\end{align*}
is the flat conformal vector Laplacian. Proceeding in the same fashion as for the Laplace equation \eqref{Poisson}, we expand 
$X^I$ and $f^I$ in a Fourier series and substitute them into \eqref{eqn:vector_poisson}
to obtain the following equations for the Fourier coefficients:
\begin{equation} \label{conf-Lap-Feqn}
Q^\Omega_\Theta(k)\widehat{X}^{\Theta}(k)= \widehat{f}^\Omega(k)
\end{equation}
where
\begin{equation*}
Q^\Omega_\Theta(k)= -|k|^2\delta^{\Omega}_{\Theta}-\frac{1}{3}k^\Omega k_\Theta.
\end{equation*}
Assuming now that $f^\Omega\in \bar{H}^s(\Tbb^3)$, then $\widehat{f}^\Omega(0)=0$, and we need to solve \eqref{conf-Lap-Feqn} for all $k\in \Zbb^3\setminus \{0\}$. Noting the that the matrix $Q^\Omega_\Theta(k)$ is invertible for $k\neq 0$ with the inverse given by
\begin{equation*}
\check{Q}_\Omega^\Theta(k)= -\frac{1}{|k|^2}\Bigl(\delta^\Theta_\Omega -\frac{1}{4|k|^2} k^\Theta k_\Omega\Bigr),
\end{equation*}
which is easily seen to satisfy $\check{Q}_\Omega^\Theta(k)Q^\Omega_\Gamma(k)=\delta^\Theta_\Gamma$, 
we can solve \eqref{conf-Lap-Feqn} to get
\begin{equation*}
\widehat{X}^{\Theta}(k)=-\frac{1}{|k|^2}\Bigl(\delta^\Theta_\Omega -\frac{1}{4|k|^2} k^\Theta k_\Omega\Bigr) \widehat{f}^\Omega(k), \quad k\in \Zbb^3\setminus \{0\}.
\end{equation*}
From this we see that
\begin{align*}
    \|X^{\Theta}\|_{\bar{H}^{s+2}}^2 = \sum_{k \in \mathbb{Z}^{3}\setminus  \{0\}} |k|^{2(s+2)}|\widehat{X}^{\Theta}(k)|^{2} &= \sum_{k \in \mathbb{Z}^{3}\setminus  \{0\}} |k|^{2s}\biggl|\Bigl(\delta^\Theta_\Omega -\frac{1}{4|k|^2} k^\Theta k_\Omega\Bigr)\widehat{f}^{\Omega}(k)\biggr|^{2}  \\
    &\lesssim \sum_{\Omega=1}^3 \sum_{k \in \mathbb{Z}^{3}\setminus  \{0\}} |k|^{2s}|\widehat{f}^{\Omega}(k)|^{2} \\
    &= \sum_{\Omega=1}^3 \|f^{\Omega}\|^2_{\bar{H}^{s}}
\end{align*}
for $\Theta=1,2,3$. 
From this inequality, we conclude that the flat conformal vector Laplacian  $\tDLd : \bar{H}^{s+2}(\Tbb^3)\rightarrow\bar{H}^{s}(\Tbb^3)$ is surjective and has trivial kernel, and hence an isomorphism. 

\subsubsection*{Linearisation of Constraint Map}
Using the results of the previous section, we now show that the linearised constraint map $D_{2}\mathcal{F}(0,0):F \rightarrow G$ is an isomorphism. \newline \par

To begin, we observe that the linearisation of $\Gamma^{2}_{(2)}$ at  $(Y,Z)=(0,0)$ is 
\begin{align*}
D_{Z}\Gamma^{2}_{(2)}|_{(Y,Z)=(0,0)}\cdot \delta Z  = D_{2}(1-j_{\Gamma\Sigma}\nu^{\Gamma}_{(2)}\nu^{\Sigma}_{(2)})^{-1}|_{(Y_{0},Z_{0})}\cdot \delta Z = \frac{2\delta_{\Gamma\Sigma}\check{\nu}_{(2)}^{\Gamma}\delta\mathring{\nu}^{\Sigma}_{(2)}}{(1-\delta_{\Theta\Delta}\check{\nu}_{(2)}^{\Theta}\check{\nu}_{(2)}^{\Delta})^{2}}.
\end{align*}
With the help of this formula, we can then compute the linearisations of $E$, $J^{\Omega}$ to obtain 
\begin{align*}
D_{Z}E|_{(Y,Z)=(0,0)}\cdot \delta Z &=  2\check{\rho}_{(2)}(K_{(2)}+1)\check{\Gamma}^{4}_{(2)}\delta_{\Sigma\Gamma}\check{\nu}_{(2)}^{\Sigma}\delta\mathring{\nu}_{(2)}^{\Gamma} + (\check{\Gamma}^{2}_{(2)}(K_{(2)}+1)-K_{(2)})\delta\mathring{\rho}_{(2)}, \\
D_{Z}J^{\Omega}|_{(Y,Z)=(0,0)}\cdot \delta Z &=\check{\Gamma}^{2}_{(2)}(K_{(2)}+1)\check{\rho}_{(2)}\delta\mathring{\nu}^{\Omega}_{(2)} + 2(K_{(2)}+1)\check{\rho}_{(2)}\check{\Gamma}^{4}_{(2)}\check{\nu}^{\Omega}_{(2)}\delta_{\Gamma\Sigma}\check{\nu}^{\Gamma}_{(2)}\delta\mathring{\nu}^{\Sigma}_{(2)} \nonumber \\
&+ \check{\Gamma}^{2}_{(2)}\check{\nu}^{\Omega}_{(2)}(K_{(2)}+1)\delta\mathring{\rho}_{(2)}, 
\end{align*}
where $\check{\Gamma}^{2}_{(2)} = (1-\delta_{\Delta\Theta}\check{\nu}^{\Delta}_{(2)}\check{\nu}^{\Theta}_{(2)})^{-1}$.
Using the above expressions, we see after a short calculation that the linearisation of the Hamiltonian constraint \eqref{eqn:perturbed_hamiltonian_constraint} can be expressed as
\begin{align}
\label{eqn:ConstraintLinearisation2a}
D_Z C_H|_{(Y,Z)=(0,0)}\cdot \delta Z= & 8\Delta \delta\tilde{\Phi} + 4\check{\rho}_{(2)}(K_{(2)}+1)\check{\Gamma}^{4}_{(2)}\delta_{\Sigma\Gamma}\check{\nu}_{(2)}^{\Sigma}\delta\mathring{\nu}_{(2)}^{\Gamma} \nonumber \\
&+ 2(\check{\Gamma}^{2}_{(2)}(K_{(2)}+1)-K_{(2)})\delta\mathring{\rho}_{(2)},
\end{align}
where in deriving this we used the fact that the background quantities $\check{\rho}_{(\afrak)}$, $\check{\nu}_{(\afrak)}^\Omega$ and $\check{K}$ satisfy the homogeneous constraint equations \eqref{eqn:Homogeneous_ID_Condition1}-\eqref{eqn:Homogeneous_ID_Condition2}.
By a similar calculation, we find that the linearisation of the Momentum constraint \eqref{eqn:perturbed_momentum_constraint} is given by
\begin{align}
  D_Z C_M|_{(Y,Z)=(0,0)}\cdot \delta Z=& \frac{1}{2}\tDLd \delta X^{\Omega} - \check{\Gamma}^{2}_{(2)}(K_{(2)}+1)\check{\rho}_{(2)}\delta\mathring{\nu}^{\Omega}_{(2)} - \check{\Gamma}^{2}_{(2)}\check{\nu}^{\Omega}_{(2)}(K_{(2)}+1)\delta\mathring{\rho}_{(2)}\nonumber\\
  &- 2(K_{(2)}+1)\check{\rho}_{(2)}\check{\Gamma}^{4}_{(2)}\check{\nu}^{\Omega}_{(2)}\delta_{\Gamma\Sigma}\check{\nu}^{\Gamma}_{(2)}\delta\mathring{\nu}^{\Sigma}_{(2)} \label{eqn:ConstraintLinearisation2b}
.
\end{align}

Now,
\begin{equation*}
\delta Z =  (\delta\tilde{\Phi}, \delta X^{\Omega},\delta \mathring{\rho}_{(2)}, \delta \mathring{\nu}^{\Omega}_{(2)})
\end{equation*}
and since $\delta\tilde{\Phi}$, $ \delta X^{\Omega}$ $\in \bar{H}^{s+2}(\Tbb^3)$ by assumption, we have that
\begin{equation*}
\Pi \delta\tilde{\Phi} = \delta\tilde{\Phi}, \quad \Pi^\perp \delta\tilde{\Phi} = 0, \quad \Pi \delta X^{\Omega} = \delta X^{\Omega} \quad \text{and} \quad \Pi^\perp \delta X^{\Omega} = 0.
\end{equation*}
Also we note that
\begin{equation*}
[\Pi,\Delta]=[\Pi^\perp,\Delta]=0 \quad \text{and} \quad [\Pi,\tDLd]=[\Pi^\perp,\tDLd]=0.
\end{equation*}
Using these relations along with the definitions \eqref{eqn:constraintperturbation_decomposition}, we see, after applying $\mathbbm{1}=\Pi+\Pi^\perp$ to \eqref{eqn:ConstraintLinearisation2a}- \eqref{eqn:ConstraintLinearisation2b}, that the linearised constraint operator
\begin{equation*}
D_2\Fc(0,0)\cdot \delta Z= \bigl(D_Z C_H|_{(Y,Z)=(0,0)}\cdot \delta Z,D_Z C_M|_{(Y,Z)=(0,0)}\cdot \delta Z\bigr)
\end{equation*}
can be written matrix form as
\begin{equation*}
D_2\Fc(0,0)\cdot \delta Z= \begin{pmatrix}B_{1} & 0 \\ 0 & B_{2}  \end{pmatrix} \begin{pmatrix}\delta\tilde{\Phi} \\ \delta X^{\Omega} \\ \delta\mathring{\nu}^{\Gamma}_{(2)} \\ \delta\mathring{\rho}_{(2)}       
    \end{pmatrix}
\end{equation*} 
where 
\begin{align*}
    B_{1} &=  \begin{pmatrix} 8\Delta & 0 & \\ 0 & \frac{1}{2}\tDLd  \end{pmatrix}
    \intertext{and}
    B_{2} &= \begin{pmatrix}  4\check{\rho}_{(2)}(K_{(2)}+1)\check{\Gamma}^{4}_{(2)}\delta_{\Sigma\Gamma}\check{\nu}_{(2)}^{\Sigma} & 2(\check{\Gamma}^{2}_{(2)}(K_{(2)}+1)-K_{(2)}) \\  -\check{\Gamma}^{2}_{(2)}(K_{(2)}+1)\check{\rho}_{(2)}\delta^{\Omega}_{\Gamma} - 2(K_{(2)}+1)\check{\rho}_{(2)}\check{\Gamma}^{4}_{(2)}\check{\nu}^{\Omega}_{(2)}\delta_{\Sigma\Gamma}\check{\nu}^{\Sigma}_{(2)} & - \check{\Gamma}^{2}_{(2)}\check{\nu}^{\Omega}_{(2)}(K_{(2)}+1) \end{pmatrix} .
\end{align*}

From the previous section, we know that the linear map
\begin{equation*}
B_1 \: : \: \bar{H}^{s+2}(\Tbb^3)\times H^{s+2}(\Tbb^3,\Rbb^3)\longmapsto \bar{H}^{s}(\Tbb^3)\times H^{s}(\Tbb^3,\Rbb^3)
\end{equation*}
is an isomorphism. Thus, the linearised operator
\begin{equation*}
D_2\Fc(0,0) \: :\: \bar{H}^{s+2}(\Tbb^3)\times H^{s+2}(\Tbb^3,\Rbb^3)\times \Rbb^3\times \Rbb \longmapsto \bar{H}^{s}(\Tbb^3)\times H^{s}(\Tbb^3,\Rbb^3)\times \Rbb^3\times \Rbb 
\end{equation*}
will be an isomorphism if and only if the $4\times 4$-matrix $B_2$ is invertible. Computing the determinant of $B_2$, we get 
\begin{align*}
    \det(B_{2}) = -\frac{(1+K_{(2)})^{3}\check{\rho}^{3}_{(2)}\big(-1+K_{(2)}|\check{\nu}_{(2)}|^2\big)}{\big(-1+|\check{\nu}_{(2)}|^2\big)^{4}}.
\end{align*}
But $\check{\rho}_{(2)}>0$, $|\check{\nu}_{(2)}|<1$, and $\frac{1}{3}<K_{(2)}<1$ by assumption, and so $\det(B_2)\neq 0$, and we conclude that
$D_2\Fc(0,0)$ is a linear isomorphism, which completes the proof.
\end{proof}

\appendix

\section{Tilted Two-Fluid Bianchi I Spacetimes\label{Appendix:two-fluid}}

\subsection{Asymptotics}
In the following analysis of two-fluid Bianchi I spacetimes, we assume, as above, a positive cosmological constant $\Lambda>0$ and linear equations of state \eqref{eos} for each fluid with sound speeds satisfying \eqref{K-range}. Then relative to a foliation by homogeneous spatial hypersurfaces determined by level sets of a time function $\tau$, the metric of a two-fluid Bianchi I solution
$(\gt,\rho_{(\afrak)},\vt_{(\afrak)})$ generated from \textit{tilted initial data} at some time $\tau=\tau_0$, i.e.~$\rho_{(\afrak)}|_{\tau=\tau_0}>0$ and $\htl(\vt_{(\afrak)},\vt_{(\afrak)})|_{\tau=\tau_0}>0$, can be expressed
in a unit lapse and zero shift gauge as
\begin{equation} \label{g-Bianchi}
    \gt = -d\tau\otimes d\tau + \htl,
\end{equation}
where
\begin{equation*}
\htl = \delta_{AB}\thetat^{A}\otimes\thetat^{B}
\end{equation*}
is the spatial metric. 
Here,  $(x^\Lambda)$ denote spatial coordinates on the 
the $\tau=\text{constant}$ hypersurfaces,  the timelike unit normal vector is given by $\et_{0} = \del_{\tau}$,  and  $\thetat^A=\thetat^{A}_\Lambda dx^\Lambda$ is a spatial coframe dual to an orthonormal spatial frame $\et_A=\et_A^\Lambda\del_{\Lambda}$, which we take to be determined via Fermi-Walker transport. 
Furthermore, the fluid four-velocities  $\tilde{v}_{(\afrak)}= \tilde{v}^{a}_{(\afrak)}\et_a$ can be expressed as 
\begin{equation}\label{v-Bianchi}   
\vt_{(\afrak)} = \frac{1}{(1-|\nu_{(\afrak)}|^2)^{\frac{1}{2}}}\bigl(\del_{\tau} + \nu^{A}_{(\afrak)}\et_A\bigr), \quad |\nu_{(\afrak)}|^2= \delta_{AB}\nu_{(\afrak)}^A\nu_{(\afrak)}^B.
\end{equation}

\begin{prop}\label{prop:two-fluid}
Suppose $\frac{1}{3}< K_{(1)} \leq K_{(2)} <1$ and let $\mu_{(\afrak)} = \frac{3K_{(\afrak)}-1}{1-K_{(\afrak)}}$, $(\afrak)=1,2$.
Then every two-fluid Bianchi I solution
$(\gt,\rho_{(\afrak)},\vt_{(\afrak)})$   of 
the Einstein-Euler equations \eqref{eqn:Einstein_Physical}-\eqref{eqn:Euler_Physical} exists globally to the future, and there exist coordinates $(t,x^\Omega)$, where $t>0$ and future timelike infinity is located at $t=0$, and constants $\breve{\frv}_{(\afrak)},\breve{\rho}_{\afrak},\breve{\nu}^{(\afrak)}_{\Omega}\in\Rbb$ with $|\breve{\nu}^{(\afrak)}|=1$, such that  
\begin{align*}
  \gt &= \frac{1}{t^{2}}\Big(-\frac{3}{\Lambda}dt\otimes dt + \bigl(\delta_{\Omega\Gamma}+\mathcal{O}(t)\bigr)dx^{\Omega}\otimes dx^{\Gamma}\Big),\\
  \Kt &= \frac{1}{t^{2}}\biggl(\frac{\sqrt{\Lambda}}{\sqrt{3}}\delta_{\Omega\Gamma} + \mathcal{O}(t)\biggr) dx^{\Omega}\otimes dx^{\Gamma}, \\
  \rho_{(\afrak)}&=  e^{(K_{(\afrak)}+1)(\breve{\mathfrak{v}}_{(\afrak)}+\breve{\rho}_{(\afrak)})}t^{\frac{2(1+K_{(\afrak)})}{1-K_{(\afrak)}}} + \mathcal{O}(t) ,\\
  \vt_{(\afrak)}&= t^{-\mu_{(\afrak)}}e^{-\breve{\mathfrak{v}}_{(\afrak)}}\bigl(1+\mathcal{O}(t)\bigr) dt + t^{-\mu_{(\afrak)}}\sqrt{e^{-2\breve{\mathfrak{v}}_{(\afrak)}} - t^{2\mu_{(\afrak)}}(1+\mathcal{O}(t))}\big(\breve{\nu}_{\Omega}^{(\afrak)}+\mathcal{O}(t)\big) dx^{\Omega}
\end{align*}
where $\Kt$ is the extrinsic curvature of the $t=\text{constant}$ hypersurfaces. 
\end{prop}
\begin{proof}
By \cite{SandinUggla:2008}, we know that the two-fluid Bianchi I solutions exist for all $\tau \geq \tau_0$, that is, globally to the future, satisfy  
$[\et_A,\et_B]=0$, $\rho_{(a)}>0$, and $0<|\nu_{(\afrak)}|<1$ for all $\tau\geq \tau_0$, and the limits 
\begin{equation}\label{nu-Bianchi-limits-A}
\breve{\nu}_{(\afrak)}^A=\lim_{\tau \rightarrow \infty} \nu_{(\afrak)}^A(\tau)
\end{equation}
exist and satisfy $\breve{\nu}_{(\afrak)}^2 =1$.
Furthermore, Wald showed that in the presence of a positive cosmological constant all Bianchi models\footnote{For matter satisfying the strong and dominant energy conditions. In particular, the fluids studied in this article satisfy both of these conditions.} (except type IX)
satisfy \cite{Wald:1983}:
\begin{align}
  \Ht&= \frac{\sqrt{\Lambda}}{\sqrt{3}}  + \mathcal{O}\Big(e^{\frac{-2\sqrt{\Lambda}}{\sqrt{3}}\tau}\Big) \quad \text{and}\quad
   \sigmat_{AB}= \mathcal{O}\Big(e^{\frac{-\sqrt{\Lambda}}{\sqrt{3}}\tau}\Big) \label{Wald.1}
\end{align}
as $\tau \rightarrow \infty$, where $3\Ht$ and $\sigmat_{AB}$ are the trace and trace-free symmetric parts of the extrinsic curvature
relative to $\et_A$.

Now, Fermi-Walker transport implies that the spatial frame vectors evolve according to
\begin{equation}\label{FW-transport-Wald}
\del_{\tau} \et_A^\Omega = -\bigl(\Ht\delta_A^B +\sigmat_A^B\bigr)\et^\Omega_B.
\end{equation}
From this evolution equation, we find, with the help of the Cauchy-Schwartz inequality, that $|\et|^2=\delta_{\Omega\Gamma}\delta^{AB} \et_A^\Omega \et_B^\Gamma$ satisfies
\begin{equation*}
\del_\tau|\et|^2 \leq 2(- \Ht + |\sigmat|)|\et|^2.
\end{equation*}
By \eqref{Wald.1} and Gr\"{o}wall's inequality, we obtain
\begin{equation*}
|\et(\tau)|^2\leq |\et(\tau_0)|^2\exp\biggl(\int_{\tau_0}^\tau -2 \frac{\sqrt{\Lambda}}{\sqrt{3}} + \mathcal{O}\Big(e^{\frac{-\sqrt{\Lambda}}{\sqrt{3}}s}\Big)\, ds\biggr) \lesssim e^{-2 \frac{\sqrt{\Lambda}}{\sqrt{3}}\tau},
\end{equation*}
which, in particular, shows that the rescaled frame components
\begin{equation} \label{e-def-Wald}
e_A^\Omega=e^{\frac{\sqrt{\Lambda}}{\sqrt{3}}\tau}\et_A^\Omega 
\end{equation}
remain bounded, that is,
$|e_A^\Omega| \lesssim |e(\tau_0)|$, $A,\Omega =1,2,3$.
From this bound along with \eqref{Wald.1}, we observe the the frame evolution equation \eqref{FW-transport-Wald} can be expressed as $\del_\tau e_A^\Omega= \mathcal{O}\big(e^{\frac{-\sqrt{\Lambda}}{\sqrt{3}}\tau}\big)$.
Integrating this from $\tau_1$ to $\tau_2$ gives
\begin{equation}\label{spatial-frame-Wald-bnd}
\Bigl|e^{\Omega}_A(\tau_2)-e^{\Omega}_A(\tau_1)\Bigr| \lesssim e^{\frac{-\sqrt{\Lambda}}{\sqrt{3}}\tau_1}\Bigl(e^{\frac{-\sqrt{\Lambda}}{\sqrt{3}}(\tau_2-\tau_1)}-1\Bigr)
\end{equation}
for all $0<\tau_1 <\tau_2$, from which we deduce
that the limit
\begin{equation} \label{eh-def-Wald}
\breve{e}^{\Omega}_A = \lim_{\tau\rightarrow \infty} e_{\Omega\Gamma}(\tau)
\end{equation}
exists. Then letting $\tau_2\rightarrow \infty$ in 
\eqref{spatial-frame-Wald-bnd} yields
\begin{equation}\label{spatial-frame-Wald-asymp-A}
e^{\Omega}_A(\tau)=\breve{e}^{\Omega}_A +  \mathcal{O}\Big(e^{\frac{-\sqrt{\Lambda}}{\sqrt{3}}\tau}\Big).
\end{equation}

Next, from \eqref{FW-transport-Wald}, we observe that
$\del_{\tau}\det(\et)= -3\Ht \det(\et)$,
where in deriving this we used the fact that $\sigmat_A^B$ is trace free, i.e. $\sigmat_A^A=0$. Using \eqref{Wald.1}, we can write this as $\del_{\tau}\ln(\det(\et)) = -\sqrt{3\Lambda} + \mathcal{O}\big(e^{\frac{-2\sqrt{\Lambda}}{\sqrt{3}}\tau}\big)$. Then
integrating this in time yields
\begin{equation*}
\det(\et(\tau_2))\approx \det(\et(\tau_1))e^{-\sqrt{3\Lambda}(\tau_2-\tau_1)}, \quad 0<\tau_1<\tau_2,
\end{equation*}
which, by \eqref{e-def-Wald}, is equivalent to
\begin{equation*}
\det(e(\tau_2))\approx \det(e(\tau_1)), \quad 0<\tau_1<\tau_2.
\end{equation*}
Letting $\tau_2\rightarrow\infty$ in this expression, it then follows from \eqref{eh-def-Wald} that
$\det(\breve{e}) \approx 1$,
and hence, that the asymptotic frame matrix $\breve{e}^\Omega_A$ is non-degenerate. But, Bianchi I spacetimes are spatially flat, and so we can introduce a change of spatial coordinates so that in this new coordinate system \eqref{spatial-frame-Wald-asymp-A} becomes
\begin{equation}\label{spatial-frame-Wald-asymp-B}
e^{\Omega}_A=\delta_A^\Omega +  \mathcal{O}\Big(e^{\frac{-\sqrt{\Lambda}}{\sqrt{3}}\tau}\Big).
\end{equation}
Using this, we can express the spatial metric as 
\begin{align}
\label{eqn:BianchiI_h_asymptotics}
\tilde{h}_{\Omega\Gamma} = e^{\frac{2\sqrt{\Lambda}}{\sqrt{3}}\tau}\biggl(\delta_{\Omega\Gamma} + \mathcal{O}\Big(e^{\frac{-\sqrt{\Lambda}}{\sqrt{3}}\tau}\Big)\biggr),
\end{align}
which shows that the Bianchi I metric \eqref{g-Bianchi} converges exponentially to a de Sitter metric in line with the results of  \cite{Wald:1983}. 
Similarly, from the decay rates \eqref{Wald.1} for $\tilde{H}$ and $\tilde{\sigma}_{AB}$, we observe that the extrinsic curvature satisfies 
\begin{align} \label{eqn:BianchiI_K_asymptotics}
   \tilde{K}_{\Omega\Gamma} = e^{\frac{2\sqrt{\Lambda}}{\sqrt{3}}\tau}\bigg( \frac{\sqrt{\Lambda}}{\sqrt{3}}\delta_{\Omega\Gamma} + \mathcal{O}\Big(e^{\frac{-\sqrt{\Lambda}}{\sqrt{3}}\tau}\Big)\bigg).
\end{align}
Now, a straightforward calculation using the fluid four-velocity representations \eqref{v-Bianchi} and the fact $[\et_A,\et_B]=0$, which holds in Bianchi I spacetimes, and equations (2.48)-(2.49) and (2.58)-(2.59) from \cite{ElstUggla:1997}, shows that the spatial fluid variables $\nu_{(\afrak)}^A$ evolve according to
\begin{align*}
    \del_\tau\nu^{A}_{(\afrak)} &= \frac{1}{ 1 - K_{(\afrak)}|\nu_{(\afrak)}|^2}\Bigl(\tilde{H}\nu_{(\afrak)}^{A}(1-|\nu_{(\afrak)}|^2)(3K_{(\afrak)}-1) + \tilde{\sigma}_{BC}\Big(-\big(1-K_{(\afrak)}|\nu_{(\afrak)}|^2\big)\delta^{AB}\nu_{(\afrak)}^{C} \nonumber \\
    &+ (1-K_{(\afrak)}) \nu_{(\afrak)}^{A}\nu_{(\afrak)}^{\langle B}\nu_{(\afrak)}^{C \rangle}\Big)\Bigr).
\end{align*}
From this, we find that the norm $|\nu_{(\afrak)}|^2$ and normalised velocity 
\begin{equation*}
\hat{\nu}_{(\afrak)}^A= \frac{1}{|\nu_{(\afrak)}|} \nu^A_{(\afrak)}
\end{equation*}
evolve according to
\begin{align}\label{nunorm-ev-A}
    \del_{\tau}|\nu_{(\afrak)}|^2 = -\frac{2}{1 - K_{(\afrak)}|\nu_{(\afrak)}|^2}\Big(\bigl(3K_{(\afrak)}-1\bigr)\tilde{H}|\nu_{(\afrak)}|^2 - \tilde{\sigma}_{BC}\nu_{(\afrak)}^{B}\nu_{(\afrak)}^{C}\Big)\bigl(|\nu_{(\afrak)}|^2-1\bigr)
\end{align}
and 
\begin{align} \label{nuhat-ev-A}
    \del_{\tau}\hat{\nu}^{B}_{(\afrak)} &= \frac{1}{ 1 - K_{(\afrak)}|\nu_{(\afrak)}|^2}\tilde{\sigma}_{BC}\Bigl(-\big(1-K_{(\afrak)}|\nu_{(\afrak)}|^2\big)(\delta^{DB}-\nuh_{(\afrak)}^{D}\nuh_{(\afrak)}^{B})\nuh_{(\afrak)}^{C} \Bigr),
\end{align}
respectively. 
Similarly, after a short calculation using equations (2.48) and (2.58) from \cite{ElstUggla:1997} and \eqref{nunorm-ev-A}, we obtain an evolution equation for the density 
\begin{align} \label{rhot-ev-A}
    \del_{\tau}\tilde{\rho}_{(\afrak)} = \frac{(1+K_{(\afrak)})\big(\tilde{\sigma}_{AB}\nu^{A}_{(\afrak)}\nu^{B}_{(\afrak)} + \tilde{H}(-3+|\nu_{(\afrak)}|^{2})\big)\tilde{\rho}_{(\afrak)}}{1-K_{(\afrak)}|\nu_{(\afrak)}|^{2}}.
\end{align}

Using \eqref{nu-Bianchi-limits-A} and \eqref{Wald.1}, we can express \eqref{nuhat-ev-A} as
$\del_{\tau}\hat{\nu}^{B}_{(\afrak)} = \mathcal{O}\big(e^{\frac{-\sqrt{\Lambda}}{\sqrt{3}}\tau}\big)$.
Integrating this from $\tau$ to $\infty$, it follows from \eqref{nu-Bianchi-limits-A} that
\begin{equation} \label{eqn:BianchiI_hatnu_asymptotics}
\hat{\nu}^{B}_{(\afrak)} = \breve{\nu}^B_{(\afrak)} + \mathcal{O}\Big(e^{\frac{-\sqrt{\Lambda}}{\sqrt{3}}\tau}\Big).
\end{equation}
In a similar fashion, we can, with the help of \eqref{nu-Bianchi-limits-A} and \eqref{Wald.1}, express \eqref{nunorm-ev-A}
as
\begin{align*}
  \frac{\sqrt{3}(1-K_{\afrak}|\nu_{(\afrak)}|^2)}{2\sqrt{\Lambda}(3K_{\afrak}-1)(1-|\nu_{(\afrak)}|^2)|\nu_{(\afrak)}|^2}  \del_{\tau}|\nu_{(\afrak)}|^2 =  1+ \mathcal{O}\Big(e^{\frac{-\sqrt{\Lambda}}{\sqrt{3}}\tau}\Big),
\end{align*}
or equivalently
$\del_\tau \frv_{(\afrak)} = \mathcal{O}\big(e^{\frac{-\sqrt{\Lambda}}{\sqrt{3}}\tau}\big)$
where
\begin{equation*}
\frv_{(\afrak)}  = \frac{1}{2}\ln\left(\frac{(1-|\nu_{(\afrak)}|^2)}{|\nu_{(\afrak)}|^\frac{2}{1-K_{(\afrak)}}}\right)+\frac{\sqrt{\Lambda}\mu_{(\afrak)}}{\sqrt{3}}\tau
\end{equation*}
and $\mu_{(\afrak)}$ is defined as defined above.
Then the same argument used above to derive \eqref{spatial-frame-Wald-asymp-B} shows that there exist constants $\breve{\frv}_{(\afrak)}\in \Rbb$ such that
$\frv_{(\afrak)}= \breve{\frv}_{(\afrak)}+  \mathcal{O}\Big(e^{\frac{-\sqrt{\Lambda}}{\sqrt{3}}\tau}\Big)$,
from which we obtain
\begin{equation} \label{eqn:BianchiI_|nu|_asymptotics}
|\nu_{(\afrak)}|^2=1-e^{-\frac{2\sqrt{\Lambda}\mu_{(\afrak)}}{\sqrt{3}}\tau}e^{2\breve{\frv}_{(\afrak)}}\Bigl(1+  \mathcal{O}\Big(e^{\frac{-\sqrt{\Lambda}}{\sqrt{3}}\tau}\Big)\Bigr).
\end{equation}
Next, a short calculation using \eqref{nunorm-ev-A} and \eqref{rhot-ev-A} shows 
\begin{align*}
    \del_{\tau}\Big(\frac{1}{K_{(\afrak)}+1}\ln(\tilde{\rho}_{(\afrak)}) + \ln(\Gamma_{(\afrak)})\Big) = -3\tilde{H}.
\end{align*}
Combining the above with \eqref{Wald.1} gives 
\begin{align*}
   \del_{\tau}\biggl(\frac{1}{K_{(\afrak)}+1}\ln(\tilde{\rho}_{(\afrak)}) + \ln(\Gamma_{(\afrak)}) +\sqrt{3\Lambda}\tau \biggr)  = \mathcal{O}(e^{-\frac{\sqrt{\Lambda}}{\sqrt{3}}}).
\end{align*}
As above, this implies the existence of a constant $\breve{\rho}\in\mathbb{R}$ such that
\begin{align*}
    \frac{1}{K_{(\afrak)}+1}\ln(\tilde{\rho}_{(\afrak)}) + \ln(\Gamma_{(\afrak)}) +\sqrt{3\Lambda}\tau = \breve{\rho}_{(\afrak)} + \mathcal{O}(e^{-\frac{\sqrt{\Lambda}}{\sqrt{3}}}),
\end{align*}
which, in turn, implies that
\begin{align}
    \tilde{\rho}_{(\afrak)} &= e^{(K_{(\afrak)}+1)(\breve{\mathfrak{v}}_{(\afrak)}+\breve{\rho}_{(\afrak)})}e^{\frac{-\sqrt{\Lambda}}{\sqrt{3}}\tau \frac{2(1+K_{(\afrak)})}{1-K_{(\afrak)}}} + \mathcal{O}(e^{-\frac{\sqrt{\Lambda}}{\sqrt{3}}}).  \label{eqn:BianchiI_rho_asymptotics}
\end{align}

Finally, introducing a new time coordinate $t$ via $t = e^{-\sqrt{\frac{\Lambda}{3}}\tau}$,
it is then straightforward to verify from \eqref{g-Bianchi}, \eqref{v-Bianchi}, \eqref{spatial-frame-Wald-asymp-B}, \eqref{eqn:BianchiI_K_asymptotics}, \eqref{eqn:BianchiI_hatnu_asymptotics}, \eqref{eqn:BianchiI_|nu|_asymptotics} and
\eqref{eqn:BianchiI_rho_asymptotics} that asymptotic formulas from the statement of the theorem hold where we have set
$\breve{\nu}^{(\afrak)}_\Omega = \delta_{\Omega A}\breve{\nu}^A_{(\afrak)}$.
\end{proof}

\subsection{Approximate Background Solutions\label{appendix:approximate_background}}
Rather than directly perturbing two-fluid Bianchi I solutions, we instead perturb a simpler class of solutions that we refer to as \textit{approximate background solutions}, which were previously studied in \cites{Oliynyk:2021,MarshallOliynyk:2022}. These are homogeneous tilted fluid solutions to the relativistic Euler equations with linear equations of state $p=K\rho$, $\frac{1}{3}<K<1$, on a de Sitter spacetime where the spacetime metric $\gt=e^{2\Phi}g$ is determined by the conformal factor $\Phi=\ln(t)$ and the conformal metric
\begin{equation}
\label{eqn:Einstein_Euler_BG_Solution1}
    g = -\frac{3}{\Lambda}dt\otimes dt +\delta_{\Omega\Sigma}dx^{\Omega}\otimes dx^{\Sigma}.
\end{equation}
As shown in \cites{Oliynyk:2021,MarshallOliynyk:2022}, the fluid density and velocity of the solutions are determined by
\begin{align}
\label{eqn:Einstein_Euler_BG_Solution2}
    (\rho_{(\afrak)}, \vt_{(\afrak)}) &= \Bigg(\frac{ \breve{\frd}_{(\afrak)} t^{\frac{2(1+K_{(\afrak)})}{1-K_{(\afrak)}}}}{(t^{2\mu_{(\afrak)}}+e^{2\frg_{(\afrak)}})^{\frac{1+K_{(\afrak)}}{2}}}, -t^{-\mu_{(\afrak)}}\sqrt{e^{2\frg_{(\afrak)}}+t^{2\mu_{(\afrak)}}}\del_t + t^{-\mu_{(\afrak)}}e^{\frg_{(\afrak)}}\xi_{(\afrak)}^\Omega\del_{\Omega} \Bigg), 
\end{align}
where $|\xi_{(\afrak)}|=1$ and  $\frg_{(\afrak)} = \frg_{(\afrak)}(t)$ solves the singular initial value problem\footnote{It is worth noting that in \cites{Oliynyk:2021,MarshallOliynyk:2022} the regular, i.e.~ initial data is specified at $t=t_0>0$, rather than singular value problem was solved to determine solutions of the relativistic Euler equations.}
\begin{align}
\label{eqn:Homogeneous_ODE_1}
    \frg_{(\afrak)}^{\prime}(t) &= \frac{t^{2\mu_{(\afrak)}-1}(\mu_{(\afrak)}-3K_{(\afrak)}+1)}{t^{2\mu_{(\afrak)}}+(1-K_{(\afrak)})e^{2\frg_{(\afrak)}}}, \;\; 0 < t < T_1, \\
\label{eqn:Homogeneous_ODE_2}
    \frg_{(\afrak)}(0) &= \breve{\frg}_{(\afrak)}.
\end{align}

\begin{prop} \label{prop:Asymptotic_ODE}
Suppose $T_1>0$, $1/3<K_{(\afrak)}<1$, $\mu_{(\afrak)} = \frac{3K_{(\afrak)}-1}{1-K_{(\afrak)}}$ and $\breve{\frg}_{(\afrak)} \in \mathbb{R}$. Then there exists a unique solution
$\frg_{(\afrak)} \in C^\infty((0,T_1)) \cap C^0([0,T_{1}))$
to the singular initial value problem \eqref{eqn:Homogeneous_ODE_1}-\eqref{eqn:Homogeneous_ODE_2} satisfying 
\begin{equation*}
|\frg_{(\afrak)}(t)-\breve{\frg}_{(\afrak)}| \lesssim t^{2\mu} \quad \text{and} \quad |\frg_{(\afrak)}'\!(t)| \lesssim t^{2\mu-1}
\end{equation*}
for all $t\in (0,T_1)$. Moreover, for each $\breve{\frd}_{(\afrak)}\in (0,\infty)$ and $\xi_{(\afrak)}^\Omega\in \Rbb$, $\Omega=1,2,3$, satisfying $|\xi_{(\afrak)}|=1$, the solution $\frg$ determines via \eqref{eqn:Einstein_Euler_BG_Solution2} a homogeneous solution of the relativistic Euler 
equations \eqref{eqn:Euler_Physical}  on the de Sitter background 
 \eqref{eqn:Einstein_Euler_BG_Solution1}.
\end{prop}
\begin{proof}
First, we observe that \eqref{eqn:Homogeneous_ODE_1} can be re-expressed as 
    \begin{align} \label{asymptotic-ODE}
        \del_{\tau}(\mathfrak{p}_{(\afrak)}) = 0
    \end{align}
    where
    \begin{align*}
        \tau &= t^{2\mu_{(\afrak)}}\quad \text{and} \quad
        \mathfrak{p}_{(\afrak)} = \frac{K_{(\afrak)}}{2}\ln\big(\tau + e^{2\fru_{(\afrak)}}\big) -\fru_{(\afrak)}.
    \end{align*}
    Now, integrating \eqref{asymptotic-ODE} with respect to $\tau$ yields the implicit solution
    \begin{align*} 
    \frac{K_{(\afrak)}}{2}\ln\big(\tau + e^{2\fru_{(\afrak)}}\big) -\fru_{(\afrak)} = \breve{\mathfrak{p}}_{(\afrak)}
    \end{align*}
    where $\breve{\mathfrak{p}}_{(\afrak)}\in \Rbb$ is an arbitrary integration constant. Setting $\breve{\mathfrak{p}}_{(\afrak)} = (K-1)\breve{\fru}_{(\afrak)}$, we consider the smooth map
    \begin{align*}
        f(\tau,\fru_{(\afrak)}) = \frac{K_{(\afrak)}}{2}\ln\big(\tau + e^{2\fru_{(\afrak)}}\big) -\fru_{(\afrak)} -\breve{\mathfrak{p}}_{(\afrak)},
    \end{align*}
    which clearly satisfies $f(0,\breve{\fru}_{(\afrak)})=0$.  The linearisation of $f$ about $(0,\breve{\fru}_{(\afrak)})$ is given by
    \begin{align*}
        D_{2}f(0,\breve{\fru}_{(\afrak)})\cdot \delta\fru_{(\afrak)} = (K_{(\afrak)}-1)\delta\fru_{(\afrak)}
    \end{align*}
    and is clearly invertible. Thus, by the implicit function theorem \cite{Abraham_et_al:1988}*{\S 2.5}, there exists a unique function $h(\tau) \in C^{\infty}(\tau_{-},\tau_{+})$, $\tau_{-}<0<\tau_{+}$, such that $h(0) = \breve{\fru}_{(\afrak)}$
    and $f(\tau,h(\tau)\big) = 0$ for $\tau \in (\tau_{-},\tau_{+})$. Setting $\fru_{(\afrak)}(t) = h(t^{2\mu})$, it then follows that $\fru_{(\afrak)}$ is the unique solution of \eqref{eqn:Homogeneous_ODE_1}-\eqref{eqn:Homogeneous_ODE_2} with $T_{1} = \tau_{+}^{2\mu}$. Next, integrating \eqref{eqn:Homogeneous_ODE_1}, we find that
    \begin{align} \label{asymptotic:ODE-bnds}
        |\fru_{(\afrak)}(t)-\breve{\fru}_{(\afrak)}| \lesssim \int_{0}^{t}s^{2\mu_{(\afrak)}-1} ds \lesssim t^{2\mu_{(\afrak)}}, \quad 0<t<T_{1},
    \end{align}
    while the bound
    \begin{align*}
        |\fru_{(\afrak)}^{\prime}(t)| \lesssim t^{2\mu_{(\afrak)}-1}, \quad 0<t<T_{1},
    \end{align*}
    follows immediately from \eqref{eqn:Homogeneous_ODE_1}. In fact, \eqref{asymptotic:ODE-bnds} shows that the solution remains bounded at $T_1$ and hence can be continued past this time, and consequently, we can take $T_1>0$ as large as we like. Also, it can be verified via a straightforward calculation that \eqref{eqn:Einstein_Euler_BG_Solution2} determines a homogeneous solution of the relativistic Euler 
equations \eqref{eqn:Euler_Physical} on the de Sitter background 
 \eqref{eqn:Einstein_Euler_BG_Solution1}.
\end{proof}

\begin{rem}\label{rem:Bianchi-approx}
Together, Propositions \ref{prop:two-fluid} and \ref{prop:Asymptotic_ODE} imply that the approximate background solution \eqref{eqn:Einstein_Euler_BG_Solution1}-\eqref{eqn:Einstein_Euler_BG_Solution2} 
will asymptotically approach a tilted two-fluid Bianchi I solution provided the asymptotic data $\breve{\fru}_{(\afrak)}$, $\breve{\mathfrak{d}}_{(\afrak)}$ and $\xi_{(\afrak)}^\Omega$ is chosen to satisfy
    \begin{align*}
        \breve{\fru}_{(\afrak)} = -\breve{\mathfrak{v}}_{(\afrak)}, \quad \breve{\mathfrak{d}}_{(\afrak)} = e^{(K_{(\afrak)}+1)\breve{\rho}} \quad \text{and}\quad
        \xi_{(\afrak)}^\Omega = \delta^{\Omega\Sigma}\breve{\nu}^{(\afrak)}_{\Sigma},
    \end{align*}
respectively.
\end{rem}

\section{Explicit Expressions for the Fluid Stress-Energy Tensor}
\label{sec:Appendix_Stress-Energy-Expansions}

Using the definitions \eqref{eqn:Tab_3+1_lorentz}-\eqref{eqn:Gamma_def}, \eqref{eqn:zeta_hat_defn}, \eqref{eqn:nuh_defn}, \eqref{eqn:u_zeta_defns}, \eqref{eqn:Rescaled_FluidVars}, and the identities
\begin{gather}
\label{eqn:uh_w_defs_app}
\ut = e^{\uh+\frg(t)}, \;\; \wh_{Q} = w_{Q} + t^{-\mu}\xi_{Q}, \\ 
\label{eqn:nuh_w_identity_app}
\nuh_{Q} = t^{\mu}w_{Q} + \xi_{Q},
\end{gather}
we can express the original fluid variables $(\rho_{(\afrak)}, \Gamma^{2}_{(\afrak)},\nu^{(\afrak)}_{A})$ in terms of our modified variables $(\zeta_{(\afrak)},\uh_{(\afrak)},w_{P})$ as follows 
\begin{align}
\label{eqn:rho_appendix_def}
    \rho_{(\afrak)} &= \frac{\rho_{0}t^{\frac{2(1+K)}{1-K}}e^{(K+1)\zeta_{(\afrak)}}}{(t^{2\mu_{(\afrak)}}+e^{2(\uh_{(\afrak)}+\frg_{(\afrak)})})^\frac{K+1}{2}}, \\
\label{eqn:Gamma2_appendix_def}
    \Gamma^{2}_{(\afrak)} & = t^{-2\mu_{(\afrak)}}(t^{2\mu_{(\afrak)}}+e^{2(\uh_{(\afrak)}+\frg_{(\afrak)})}), \\ 
\label{eqn:nu_A_appendix_def}
    \nu^{(\afrak)}_{A} & = \frac{e^{\uh_{(\afrak)}+\frg_{(\afrak)}}}{\sqrt{t^{2\mu_{(\afrak)}}+e^{2(\uh_{(\afrak)}+\frg_{(\afrak)})}}}(t^{\mu_{(\afrak)}}w^{(\afrak)}_{A}+\xi^{(\afrak)}_{A}). 
\end{align}
Hence, in terms of $\zeta_{(\afrak)}$, $\ut_{(\afrak)}$, and $w^{(\afrak)}_{A}$, the stress-energy tensor can be decomposed as 
\begin{align}
\frac{T^{(\afrak)}_{00}}{t^{2}} &= \frac{t^{2}(t^{2\mu_{(\afrak)}}+(1+K_{(\afrak)})e^{2(\uh_{(\afrak)}+\frg_{(\afrak)})})\rho_{0}e^{(K_{(\afrak)}+1)\zeta_{(\afrak)}}}{(t^{2\mu_{(\afrak)}}+e^{2(\uh_{(\afrak)}+\frg_{(\afrak)})})^\frac{K_{(\afrak)}+1}{2}}, \\
\frac{T^{(\afrak)}_{0A}}{t^{2}} &= \frac{t^{2}(1+K_{(\afrak)})e^{\uh_{(\afrak)}+\frg_{(\afrak)}}\rho_{0} e^{(K_{(\afrak)}+1)\zeta}}{(t^{2\mu_{(\afrak)}}+e^{2(\uh_{(\afrak)}+\frg_{(\afrak)})})^{\frac{K_{(\afrak)}}{2}}}(t^{\mu_{(\afrak)}}w^{(\afrak)}_{A}+\xi^{(\afrak)}_{A}), \\
\label{eqn:Tab_Decomp_Wbar_Vars_app}
\frac{T^{(\afrak)}_{AB}}{t^{2}} &= \frac{t^{2}(1+K_{(\afrak)})e^{2(\uh_{(\afrak)}+\frg_{(\afrak)})}\rho_{0} e^{(K_{(\afrak)}+1)\zeta_{(\afrak)}}}{(t^{2\mu_{(\afrak)}}+e^{2(\uh_{(\afrak)}+\frg_{(\afrak)})})^{\frac{K_{(\afrak)}+1}{2}}}(t^{\mu_{(\afrak)}}w^{(\afrak)}_{A}+\xi^{(\afrak)}_{A})(t^{\mu_{(\afrak)}}w^{(\afrak)}_{B}+\xi^{(\afrak)}_{B}) \nonumber \\
&+ \frac{K_{(\afrak)}\rho_{0}t^{\frac{4K_{(\afrak)}}{1-K_{(\afrak)}}}e^{(K_{(\afrak)}+1)\zeta_{(\afrak)}}}{(t^{2\mu_{(\afrak)}}+e^{2(\uh_{(\afrak)}+\frg_{(\afrak)})})^\frac{K_{(\afrak)}+1}{2}}\eta_{AB}.
\end{align}

\section{Explicit Expressions for  $U_{ab}$}
\label{sec:ConformalFactor_Appendix}
By expanding \eqref{eqn:U_conf_def}, $U_{ab}$ can be written in terms of $r_{a}$ and the connection coefficients as
\begin{align*}
 U_{ab} = e_{a}(r_{b}) + e_{b}(r_{a}) - \tensor{\omega}{_a^c_b}r_{c} - \tensor{\omega}{_b^c_a}r_{c} +\eta_{ab}\Big(e_{c}(r^{c})+ \tensor{\omega}{_c^c_i}r^{i} + 2r_{c}r^{c}\Big) -2r_{a}r_{b}.
\end{align*}
First, we consider the case $a=b=0$:
\begin{align*}
    U_{00} = 2e_{0}(r_{0}) - 2\tensor{\omega}{_0^0_0}r_{0} -\Big(-e_{0}(r_{0})-\tensor{\omega}{_c^c_0}r_{0} -2(r_{0})^{2}\Big) - 2(r_{0})^{2}.
\end{align*}
Using \eqref{eqn:connection_identities} and \eqref{ra-def}, we observe that this expression simplifies to
\begin{align*}
    U_{00} &= 3e_{0}(r_{0})+ \tensor{\omega}{_{CB0}}\eta^{CB}r_{0}= 3e_{0}(r_{0}) + 3\mathcal{H}r_{0}.
\end{align*}
Similarly, setting $a=A$ and $b=0$ gives
\begin{align*}
    U_{A0} &= e_{A}(r_{0}) + \tensor{\omega}{_{00A}}r_{0} = \frac{1}{\alpha^{2}t}e_{A}(\alpha) +\frac{\dot{U}_{A}}{\alpha t}.
\end{align*}
Using the constraint \eqref{eqn:C3_constraint}, we can write this as
\begin{align*}
    U_{A0} = \frac{1}{\alpha t}\dot{U}_{A} +\frac{\dot{U}_{A}}{\alpha t} = \frac{2}{\alpha t}\dot{U}_{A}.
\end{align*}
Finally, setting $a=A$ and $b=B$, we find that
\begin{align*}
    U_{AB} &= - \tensor{\omega}{_A^0_B}r_{0} - \tensor{\omega}{_B^0_A}r_{0} + \delta_{AB}\Big(e_{c}(r^{c})+ \tensor{\omega}{_c^c_i}r^{i} + 2r_{c}r^{c}\Big) \nonumber \\
    &= 2S_{AB}r_{0}-\delta_{AB}e_{0}(r_{0}) + \tensor{S}{_I^I}r_{0}\delta_{AB}-2(r_{0})^{2}\delta_{AB}, \nonumber \\
\end{align*}
where 
\begin{align*}
    S_{AB} = - (\mathcal{H}\delta_{AB} + \tensor{\Sigma}{_A_B}).
\end{align*}
Thus, we have
\begin{align*}
    \tensor{U}{_A^A} &= -3e_{0}(r_{0}) -15\mathcal{H}r_{0} -6(r_{0})^{2},
\end{align*}
and hence that
\begin{align*}
    U_{00} +\tensor{U}{_A^A} = -6(2\mathcal{H}+r_{0})r_{0}.
\end{align*}

Now, the symmetric trace-free part of $U_{AB}$ is defined by
\begin{align*}
    U_{\langle AB \rangle} = U_{AB} - \frac{1}{3}\tensor{U}{_D^D}\delta_{AB}.
\end{align*}
With the help of the above expressions for $U_{AB}$ and $\tensor{U}{_D^D}$, we conclude that
\begin{align*}
    U_{\langle AB \rangle} &= 2S_{AB}r_{0} + \tensor{S}{_I^I}r_{0}\delta_{AB} + 5\mathcal{H}r_{0}\delta_{AB} = \frac{2}{\alpha t}\Sigma_{AB}.
\end{align*}

\bibliography{refs.bib}

\end{document}